\PassOptionsToPackage{numbers}{natbib}
\documentclass[graybox,natbib,envcountsect,envcountsame]{svmult}

\usepackage{mathptmx}
\usepackage{helvet}
\usepackage{courier}
\usepackage{makeidx}
\usepackage{graphicx}
\usepackage{multicol}
\usepackage{comment}
\usepackage[bottom]{footmisc}

\usepackage{amsmath}   
\usepackage{amssymb} %
\usepackage[utf8]{inputenc}
\usepackage{tikz}
\usetikzlibrary{%
  cd, %
  matrix %
}

\usepackage{mathrsfs}
\usepackage[all]{xy}
\usepackage{subcaption}

\newcommand{\R}{\mathbb{R}}

\spnewtheorem*{AAT}{Abel's Addition Theorem}{\bfseries}{\itshape}

\smartqed
\usepackage{etoolbox}
\patchcmd{\qed}{\ifmmode\qedsymbol}{\ifmmode\the\qedsymbol}{}{\foobar}
\patchcmd{\qed}{\hfil\qedsymbol}{\hfil\the\qedsymbol}{}{\foobar}
\qedsymbol{\proofSymbol}

\newcommand{\qedhere}{\tag*{\the\qedsymbol}}

\usepackage{xspace}

\def\ddefloop#1{\ifx\ddefloop#1\else\ddef{#1}\expandafter\ddefloop\fi}

\def\ddef#1{\expandafter\def\csname bb#1\endcsname{\ensuremath{\mathbb{#1}}\xspace}}
\ddefloop ABCDEFGHIJKLMNOPQRSTUVWXYZ\ddefloop

\def\ddef#1{\expandafter\def\csname c#1\endcsname{\ensuremath{\mathcal{#1}}\xspace}}
\ddefloop ABCDEFGHIJKLMNOPQRSTUVWXYZ\ddefloop

\def\ddef#1{\expandafter\def\csname m#1\endcsname{\ensuremath{\mathbf{#1}}\xspace}}
\ddefloop ABCDEFGHIJKLMNOPQRSTUVWXYZabcdefghijklmnopqrstvwxyz\ddefloop

\def\ddef#1{\expandafter\def\csname v#1\endcsname{\ensuremath{\boldsymbol{#1}}\xspace}}
\ddefloop ABCDEFGHIJKLMNOPQRSTUVWXYZabcdefghijklmnopqrstuvwxyz\ddefloop

\def\ddef#1{\expandafter\def\csname v#1\endcsname{\ensuremath{\boldsymbol{\csname #1\endcsname}}\xspace}}
\ddefloop {alpha}{beta}{gamma}{delta}{epsilon}{varepsilon}{zeta}{eta}{theta}{vartheta}{iota}{kappa}{lambda}{mu}{nu}{xi}{pi}{varpi}{rho}{varrho}{sigma}{varsigma}{tau}{upsilon}{phi}{varphi}{chi}{psi}{omega}{Gamma}{Delta}{Theta}{Lambda}{Xi}{Pi}{Sigma}{varSigma}{Upsilon}{Phi}{Psi}{Omega}{ell}\ddefloop

\newcommand{\zo}{\ensuremath{{\{0,1\}}}\xspace}

\newcommand{\indc}{\ensuremath{\approx_c}\xspace}

\newcommand{\NP}{\ensuremath{\mathsf{NP}}\xspace}
\newcommand{\BPP}{\ensuremath{\mathsf{BPP}}\xspace}

\newcommand{\polylog}{\mathrm{polylog}}

\newcommand{\msf}[1]{\ensuremath{\mathsf{#1}}}

\renewcommand{\star}{^*}

\newcommand{\set}[1]{\left\{#1\right\}}

\newcommand{\N}{\mathbb{N}}

\renewcommand{\vec}[1]{\ensuremath{\boldsymbol{\mathrm{#1}}}}

\providecommand{\poly}{\msf{poly}}

\mathchardef\mhyphen="2D

\newlength{\protowidth}

\newlength{\cprotowidth}

\renewcommand{\P}{\ensuremath{\mathsf{P}}\xspace}

\def\draft{0}

\ifnum\draft=1
\newcommand{\authnote}[3]{\textsf{\color{#3}  $\ll$ #1: { #2} $\gg$ }  }
\else
\newcommand{\authnote}[3]{}
\fi

\newcommand{\BGW}{{\sf BGW}}
\newcommand{\VSS}{{\sf VSS}}
\newcommand{\trunc}{{\sf trunc}}
\newcommand{\MPC}{{\sf MPC}}

\newcommand{\myparagraph}[1]{\medskip\noindent\textbf{#1}}

\begin{document}

\title*{On the works of Avi Wigderson\thanks{To appear in H. Holden, R. Piene: The Abel Laureates 2018–2022, Springer, 2024}
}

\author{Boaz Barak\inst{1} \and Yael Kalai\inst{2} \and Ran Raz\inst{3} \and Salil Vadhan\inst{1} \and Nisheeth K. Vishnoi\inst{4}} 
\institute{School of Engineering and Applied Sciences, Harvard University. \and
           Microsoft Research and MIT \and 
           Department of Computer Science, Princeton University. \and 
           Department of Computer Science, Yale University.}

\maketitle

\abstract{ %
This is an overview of some of the works of Avi Wigderson, 2021 Abel prize laureate. Wigderson's contributions span many fields of computer science and mathematics.
In this survey we focus on four subfields: \emph{cryptography}, \emph{pseudorandomness}, \emph{computational complexity lower bounds},  and the theory of \emph{optimization over symmetric manifolds}.
Even within those fields, we are not able to mention all of Wigderson's results, let alone cover them in full detail.
However, we attempt to give a broad view of each field, as well as describe how Wigderson's papers have answered central questions, made key definitions, forged unexpected connections, or otherwise made lasting changes to our ways of thinking in that field.
}

\section{Introduction} \label{sec:intro}

In a career that has spanned more than 40 years, Wigderson has resolved long-standing open problems, made definitions that shaped entire fields, built unexpected bridges between different areas, and introduced ideas and techniques that inspired generations of researchers.
A recurring theme in Wigderson's work has been
uncovering the deep connections between computer science and mathematics. 
His papers have both demonstrated unexpected applications of diverse mathematical areas to questions in computer science, and shown how to use theoretical computer science insights to solve problems in pure mathematics.
Many of these beautiful connections are surveyed in Wigderson's own book~\cite{wigderson19}.

In writing this chapter, we were faced with a daunting task. Wigderson's body of work is so broad and deep that it is impossible to do it justice in a single chapter, or even in a single book.
Thus, we chose to focus on a few subfields and, within those, describe only some of Wigderson's central contributions to these fields.

In Section~\ref{sec:crypto} we discuss Wigderson's contribution to \emph{cryptography}. As we describe there, during the second half of the 20th century, cryptography underwent multiple revolutions.
Cryptography transformed from a practical art focused on ``secret writing'' to a science that protects not only communication but also \emph{computation}, and provides the underpinning for our digital society.
Wigderson's works have been crucial to this revolution, 
vastly extending its reach through constructions of objects such as \emph{zero-knowledge proofs} and \emph{multi-party secure computation}. 
In Section~\ref{sec:derand}, we discuss Wigderson's contribution to the field of \emph{pseudorandomness}.
One of the great intellectual achievements of computer science and mathematics alike has been the realization that many deterministic processes can still behave in ``random-like'' or \emph{pseudorandom} manner.  Wigderson has led the field in understanding and pursuing the deep implications of pseudorandomness for problems in computational complexity, such as the power of randomized algorithms and circuit lower bounds, and in developing theory and explicit constructions of ``pseudorandom objects'' like expander graphs and randomness extractors.
Wigderson's work in this field used mathematical tools from combinatorics, number theory, algebra, and information theory to answer computer science questions, and has applied computer science abstractions and intuitions to obtain new results in mathematics, such as explicit constructions of Ramsey graphs.

Section~\ref{sec:lowerbounds} covers Wigderson's contribution to the great challenge of theoretical computer science: proving \emph{lower bounds} on the computational resources needed to achieve computational tasks.
Algorithms to solve computational problems have transformed the world and our lives, but for the vast majority of interesting computational tasks, we do not know whether our current algorithms are \emph{optimal} or whether they can be dramatically improved. %
To demonstrate optimality, one needs to prove such \emph{lower bounds}, and this task has turned out to be exceedingly difficult, with the famous $\P$ vs. $\NP$ question being but one example. %
While the task is difficult, there has been some progress in it, specifically in proving lower bounds for restricted (but still very useful and interesting) computational models.
Wigderson has been a central contributor to this enterprise.

Section~\ref{sec:optimization} covers a line of work by Wigderson and his co-authors on developing and analyzing continuous optimization algorithms for various problems in computational complexity theory, mathematics, and physics. 
Continuous optimization is a cornerstone of science and engineering.
There is a very successful theory and practice of convex optimization.
However,  progress in the area of \emph{nonconvex} optimization has been hard and sparse, despite a plethora of nonconvex optimization problems in the area of machine learning.
Wigderson and his co-authors, in their attempt to analyze some nonconvex optimization problems important in complexity theory, realized that these are no ordinary nonconvex problems -- nonconvexity arises because the objective function is invariant under certain group actions.
This insight led them to synthesize tools from invariant theory, representation theory, and optimization to develop a quantitative theory of optimization over Riemannian manifolds that arise from continuous symmetries of noncommutative matrix groups.
Moreover, this pursuit revealed connections with and applications to  a host of disparate problems in mathematics and physics.

All of us are grateful for having this opportunity to revisit and celebrate Wigderson's work.
More than anything, we feel lucky to have had the joy and privilege of knowing Avi as a mentor, colleague, collaborator, and friend. %

\section{Cryptography} \label{sec:crypto}

Cryptography has been used for thousands of years, going back to ancient Egypt, Sumeria, and Greece.
However, throughout the vast majority of that time, it had two major limitations. 
First, there was no formal analysis of cryptographic schemes, leading to a ``cat and mouse'' game in which ciphers are continuously designed and then broken, leading Edgar Allan Poe to say in 1847 that \emph{``Human ingenuity cannot concoct a cipher which human ingenuity cannot resolve.''}
Second, cryptography was synonymous with ``secret writing'': the design of schemes that enable two parties that share some secret information (i.e., \emph{ secret key}) to communicate by using encryption and decryption.

In the second half of the 20th century, cryptography broke out of these two limitations. First, starting with the work of Shannon~\cite{shannon1949communication}, cryptography was placed on solid mathematical foundations. Second, with their invention of \emph{public key cryptography}, Diffie and Hellman~\cite{DiffieH76} ushered in a new era where cryptography extended far beyond secret writing.
However, neither Shannon nor Diffie and Hellman could imagine how far cryptography would grow.
First, in almost all settings, analyzing cryptographic schemes required going beyond the information-theoretic methods of Shannon, and to use \emph{computational complexity} as a basis.
Second, in the 1980s, cryptography was extended to protect not only \emph{communication} but also \emph{computation}, with a crowning achievement being ``secure multiparty computation'' protocols that allow any number of parties to compute arbitrary functions on their secret inputs, controlling precisely what information would be revealed and to whom.

Avi Wigderson played a key role in these developments. He was instrumental in mapping out the computational assumptions required for many cryptography tools and proving the central feasibility result for secure multiparty computation. 

In this section, we survey some of Wigderson's contributions to cryptography, focusing on two central themes.

\begin{enumerate}
    \item Building cryptographic schemes that are secure under \emph{computational assumptions} (which can be viewed as stronger variants of the famous $\P\neq\NP$ conjecture). This line of works is covered in Section~\ref{sec:crypto:comp}.

    \item Building cryptographic schemes that are proven to be secure \emph{unconditionally}, without relying on any computational assumptions, but rather on certain environmental conditions such as a trusted majority of parties. This field is sometimes known as \emph{information-theoretically secure cryptography}, and some of Wigderson's contributions to it are covered in Section~\ref{sec:crypto:it}.
    
\end{enumerate}

Two objects play a central role in both fields: \emph{zero-knowledge proofs}, and \emph{multi-party secure computation}. 
These are the foundational tools for extending cryptography from only securing \emph{communication} to securing \emph{computation}.
Wigderson has made seminal contributions to constructing these objects in both the computational and information-theoretic regimes, enabling many follow-up works that used these tools to achieve a vast range of cryptographic applications. 
The description below is informal in parts, and many proofs are omitted. However they can be found in Goldreich's excellent textbook \cite[]{GoldreichFOCvol1,GoldreichFOCvol2}.
See also the recent text~\cite{tyagi2023information} for more on information-theoretic cryptography.

\subsection{Cryptography under computational assumptions} \label{sec:crypto:comp}

Cryptography is intimately connected to computational complexity. Indeed, achieving most cryptographic goals requires the existence of functions that are \emph{computationally hard to compute}.
The necessity of computational hardness for encryption was realized early on by Shannon~\cite{shannon1949communication}.
However, in the early 1980s, researchers realized that computational hardness is \emph{sufficient} for achieving applications that extend far beyond encryption.
Avi Wigderson played a key role in this revolution that vastly expanded the domain of cryptography beyond its classical goals of protecting the confidentiality and authenticity of communications.

\subsubsection{Zero knowledge proofs for all languages in $\NP$}\label{sec:comp-ZK}

\emph{Zero-knowledge proofs} achieve the seemingly paradoxical notion of convincing a party (known as the ``verifier'') that a particular statement X is true, without giving any information to the verifier as to \emph{why} that statement is true.
For example,  a prover that knows the factorization of the number  $N = 1,013,883,390,263,903$  as $N= 32,722,259 \times 30,984,517$ can easily prove the statement ``$N$ is composite and has a factor with least-significant digit $7$'' by providing the factorization, but a zero-knowledge proof allows them to prove this fact \emph{without} revealing the prime factors.\footnote{Simply proving that a number $N$ is composite can be done easily, since the verifier can use an efficient \emph{primality testing} algorithm~\cite{AgrawalKaSa04}, and so even an empty proof suffices. Also, current classical (i.e., non quantum) algorithms can be used to efficiently factor numbers with up to a few hundred digits~\cite{lenstra1993development,BoudotGGHTZ22}.}

In 1982, Goldwasser, Micali and Rackoff~\cite{GoldwasserMR85} defined the notion of zero knowledge proofs, and gave such proofs for particular examples of languages, such as quadratic residuosity, for which no efficient algorithm is known.
To do this, \cite{GoldwasserMR85} needed to extend the notion of proofs to include \emph{interaction}--- the prover and verifier exchange messages rather than just a static piece of text--- and \emph{randomization}--- the verifier's algorithm is randomized, and it is only convinced of the proof validity with high probability. 

We now formally define interactive proofs in general and zero-knowledge proofs in particular.
We restrict our attention to proof systems for languages in $\NP$, which is the case of most practical relevance in cryptography.
Recall that the class $\NP$ consists of all languages for which membership can be \emph{efficiently verified}.
Formally, we define an $\NP$-relation to be a relation $R \subseteq \{0,1\}^* \times \{0,1\}^*$ such that there is a polynomial-time algorithm to check whether a pair $(x,y)$ is in $R$ and such that there is a polynomial $p$ such that for every pair $(x,y)\in R$, $|y| \leq p(|x|)$.
For a relation $R$, we define the \emph{lanaguage} corresponding to $R$ to be $L_R = \{ x | (x,y)\in R \}$.
The class $\NP$ is the set of all languages $L$ such that $L=L_R$ for some $\NP$-relation $R$.
Given a string $x$, one can \emph{prove} that $x$ is in $L$ by providing the string $y$ such that $(x,y)\in R$.
By the definition of an $\NP$ relation, the string $y$ will be of length at most polynomial in $|x|$ and the membership of $(x,y)$ in $R$ can be verified efficiently.

An \emph{interactive algorithm} is a (potentially randomized) algorithm that, given a current state $s_i$ and the message received $m_i$, outputs the updated state $s_{i+1}$ and the message it sends $m_{i+1}$. An interaction between two interactive algorithms $A$ and $B$ given inputs $a,b$, respectively, proceeds in the natural way: the initial states of $A$ and $B$ are $a,b$, respectively. Then we iterate between each algorithm, computing its new state and message sent based on the previous state and message received from the other party. (For concreteness, we assume that $A$ sends the first message, and thus gets the empty message and its initial state as input to compute it.)
We say that an interactive algorithm  $A$ is \emph{polynomial time} if there are some polynomials $p,q$ such that, letting $n$ be the length of $A$'s input: (1) the total number of rounds $A$ will interact with before halting, as well as the length of each message it sends and its internal state, is at most $p(n)$, (2) $A$ computes every message and updated state using at most $q(n)$ operations.

\begin{definition}\label{def:IPNP} (Interactive proofs for $\NP$ languages) Let $R$ be an $\NP$ relation and $L=L_R$ the corresponding $\NP$ language. An (efficient) \emph{interactive proof system} for $R$ is a pair of interactive randomized polynomial-time algorithms $P,V$ that satisfy the following properties:

\begin{description}
\item[Completeness:] For every $(x,y)\in R$, if $P$ gets $x,y$ as input and $V$ gets $x$ as input, then at the end of the interaction, $V$ will output $1$ with probability $1$.

\item[Soundness:] For every interactive algorithm $P^*$ (even inefficient one) and every $x\not\in L$, if $V$ gets $x$ as input and interacts with $P^*$, then the probability that $V$ outputs $1$ at the end of the interaction is at most $\tfrac{1}{2}$.\footnote{The probability $1/2$  of error can be reduced to $2^{-k}$ by the standard trick of repeating the protocol $k$ times sequentially.}
\end{description}
\end{definition}

The formal definition of Zero-Knowledge proofs uses the notion of a \emph{simulator}. 
The idea is that to demonstrate that a verifier $V$ did not learn anything from an interaction with a prover $P$, we show that $V$ \emph{could have simulated the interaction by itself}. 

\begin{definition}\label{def:ZKP} (Zero-knowledge proofs)  Let $R$ be an $\NP$ relation and $(P,V)$ be an efficient interactive proof system for $R$. 
We say that $(P,V)$ is \emph{zero knowledge} if the following holds. For every polynomial-time interactive algorithm $V^*$ there exists a (non-interactive) randomized polynomial-time algorithm $S^*$ such that: for every $(x,y)\in R$, if we let $s^*$ be the random variable corresponding to $V^*$'s state, then $s^*$ is \emph{computationally indistinguishable} from the random variable $S^*(x)$.
\end{definition}

Let $\{ X_\alpha \}_{\alpha \in I}$ and $\{ Y_\alpha \}_{\alpha \in I}$ be two parameterized collections of random variables, with $I\subseteq \{0,1\}^*$, and  $X_\alpha,Y_\alpha$ supported on strings of length at most polynomial in $|\alpha|$. We say that $\{ X_\alpha \}$ and $\{ Y_\alpha \}$ are \emph{computationally indistinguishable}, denoted by $\{ X_\alpha \} \indc \{ Y_\alpha \}$  if there exists some function $\mu:\N \rightarrow (0,1]$ such that $\mu(n)=n^{-\omega(1)}$ (i.e., $\lim_{n\rightarrow \infty} \tfrac{\log \mu(n)}{\log n} = -\infty$) and such that for every Boolean circuit $C_\alpha$ of size at most $1/\mu(|\alpha|)$, $\left| \Pr[ C_\alpha(X_\alpha)=1 ] - \Pr[ C_\alpha(Y_\alpha)=1 ] \right|<\mu(|\alpha|)$.
We often omit the subscript $\alpha$ when it is clear from context and so use the notation $X_\alpha \indc Y_\alpha$ or simply $X \indc Y$. For example, the condition of Definition~\ref{def:ZKP} is that there is some function $\mu:\N\rightarrow [0,1]$ such that $\mu(n)=n^{-\omega(1)}$ and such that for every Boolean circuit $C$ of size $\leq 1/\mu(|x|)$, $| \Pr[ C(s^*)=1] -\Pr[ C(S^*(x))=1]| < \mu(|x|)$.

\medskip

At the time of Goldwasser et al's result, it was not clear that zero-knowledge proofs are not a mere ``curiosity'' restricted to very specific examples. (Indeed, the \cite{GoldwasserMR85} paper famously took three years before it was accepted for publication.)
In Wigderson's work with Goldreich and Micali~\cite{GoldreichMW86} they showed that this is decidedly not the case.
Rather, they showed that (under standard cryptographic assumptions) \emph{every} language in $\NP$ has a zero-knowledge proof.

One way to define $\NP$ is that it consists of languages $L$ such that the membership of a string $x$ in $L$ can be proven by an efficient mathematical proof (i.e., a piece of text at most polynomially long in $|x|$, which can be verified in polynomial time).
Hence \cite{GoldreichMW86}'s result can be thought of as saying that if a statement $x$ has an efficient proof at all, then it also has an efficient proof in which the verifier learns nothing except that the statement is true.\footnote{This way of phrasing is a bit cheating, since the first instance of ``proof'' corresponds to a standard mathematical proof--- a static deterministically-verifiable piece of text--- while the second one corresponds to the extended notion of \cite{GoldwasserMR85} which includes interaction and randomness. A followup work \cite{Ben-OrGGHKMR88} extended this by showing that everything that can be proven by a randomized interactive proof, also has such a proof which is zero knowledge.}
The way that \cite{GoldreichMW86} proved their theorem was ingenious.
They used the celebrated Cook-Levin Theorem, which is typically considered as a \emph{negative} or \emph{impossibility} result, to show give a positive result.

We now describe their protocol. The Cook-Levin Theorem says that there are concrete problems that are \emph{$\NP$-complete} in the sense that any other problem in $\NP$ reduces to them. 
One example of an $\NP$-complete language is \emph{three coloring} or $3COL$, which is defined as follows.
For a graph $G=(V,E)$, $G \in 3COL$ if and only if there exists $\chi:V \rightarrow \{1,2,3\}$ such that for every $\{u,v\}\in E$, $\chi(u) \neq \chi(v)$.
A classical result is the following:

\begin{theorem}\label{thm:threecolnpc} (Cook-Levin-Karp \cite{Cook71,levin1973,Karp72})  $3COL$ is $\NP$-complete. That is, for every $L\in\NP$ there exists a polynomial-time reduction $r:\{0,1\}^* \rightarrow \{0,1\}^*$ such that for every $x\in \{0,1\}^*$,
\[
x\in L \Leftrightarrow r(x) \in 3COL \;.
\]
Moreover, if $R$ is the $\NP$-relation corresponding to $L$, there exists polynomial-time algorithms $r',r''$ such that
\begin{enumerate}
\item For every $(x,y)\in R$, $r'(x,y)$ is a valid 3-coloring for the graph $r(x)$.

\item For every $x\in \{0,1\}^*$ and $G=r(X)$, if $\chi$ is a valid 3-coloring for $G$ then $(x,r''(G,\chi)) \in R$.
\end{enumerate}
\end{theorem}

The ``moreover'' part of Theorem~\ref{thm:threecolnpc} was already implicit in the classical works of \cite{Cook71,levin1973,Karp72}, and was explicitly discussed by Levin (which is why a triple $(r,r',r'')$ as above is sometimes known as a ``Levin reduction'').
Wigderson and his coauthors used this insight to show that in order to give a zero-knowledge proof system for all $L\in \NP$, it suffices to give a zero-knowledge proof system for $3COL$.
Specifically, for every $\NP$ language $L$, let $r_L, r'_L,r''_L$ the reductions from $L$ to $3COL$ as in Theorem~\ref{thm:threecolnpc}.
Given a zero-knowledge protocol $(P_{3COL},V_{3COL})$ for $3COL$ we can obtain a protocol $(P_L,V_L)$ as follows: 

\begin{itemize}
\item Verifier and prover get $x$ as input, and the prover gets in addition $y$ such that $(x,y)\in R_L$
\item Verifier and prover compute $G=r_L(x)$ and prover computes $\chi = r'_L(x,y)$.
\item Verifier and prover run the protocol  $(P_{3COL},V_{3COL})$  with inputs $G$ and  $(G,\chi)$ respectively.
\end{itemize}

\subsubsection{Computationally secure multiparty computation} \label{sec:MPC-GMW}

Obtaining a zero-knowledge proof system for every problem in $\NP$ is an intellectually satisfying result on its own merits. But does it have further applications? 
In another work of Wigderson with Goldreich and Micali~\cite{GoldreichMW87}, they showed that the answer is a resounding \emph{yes}.
They introduced a general technique to use zero-knowledge proofs as a way to \emph{compile} protocols that achieve a very weak form of security into ones that achieve a very strong one.
Using their technique, \cite{GoldreichMW87} proved what is arguably ``The Fundamental Theorem of Cryptography''--- a protocol for \emph{ secure multiparty computation}.

Secure multiparty computation (MPC) is a vast generalization of many tasks in cryptography, including encryption, electronic voting, voting, privacy-preserving data mining, and more.
Full proofs and even the precise definition of MPC with all its variants is beyond the scope of this section.
Lindell's survey~\cite{Lindell21} gives an excellent introduction, while the books~\cite{GoldreichFOCvol2,CramerMPC15} go into more detail.

The setup is that there are $n$ parties holding private inputs $x_1 \in \mathcal{X}_1,\ldots, x_n \in \mathcal{X}_n$, and they wish to compute a (potentially probabilistic) map
\[
F:\mathcal{X}_1\times \cdots \times{X}_n \rightarrow \mathcal{Y}_1\times \cdots \times \mathcal{Y}_n
\]
such that (roughly speaking) the following properties hold:

\begin{description}
\item[\textbf{Completeness}] Every party $i$ will learn the value $y_i$ where $(y_1,\ldots,y_n) = F(y_1,\ldots,y_n)$.

\item[\textbf{Privacy}] A party $i$ will not learn anything else apart from $y_i$ about the private inputs of the other parties. More generally, every adversary that controls some set $A \subseteq [n]$ of parties will not learn more about the private inputs $\{ x_i \}_{i\not\in A}$ of the other parties beyond what could be derived from the outputs $\{ y_i \}_{i\in A}$.

\item[\textbf{Soundness}] An adversary that controls $A$ as above cannot modify the outputs $y_i$ of $i\not\in A$ beyond its choice of the inputs $\{ x_i \}_{i\in A}$.\footnote{The adversary might also be able to abort the protocol; we ignore this issue of aborts in this section, but it is discussed extensively in the literature. }

\end{description}

Up to considerations of computational efficiency, as well as allowing for interactive communication between parties, we can cast almost any cryptography problem as an instance of MPC. 
For example, the encryption task can be thought of as computing the function $F(x,\empty)=(\empty,x)$ where $\empty$ is the ``empty'' input/output. That is, computing $F$ corresponds to ensuring that the second party learns $x$ and that no one learns anything else.
Conducting an auction could correspond to computing the function $F(x_1,\ldots,x_n) = (y_1,\ldots,y_n)$ where $y_i= 1_{i=\arg\max \{ x_i \}}$. (That is, each party only learns whether or not they were the highest bidder.)

Yao~\cite{Yao86} gave a version of MPC that was restricted in two ways.
First, Yao's protocol was only for two parties.
Second, and more importantly, Yao's protocol assumed a very restricted (and unrealistic) adversary: one that follows precisely the protocol's instructions but tries to extract information from the communication it is involved in.
Such adversaries are known in cryptographic parlance as \emph{passive} or \emph{honest-but-curious}.
Since in general, we have no reason to expect attackers to obey our protocol's instructions, honest-but-curious is not a realistic model for security.

Wigderson's work~\cite{GoldreichMW87} solved both issues.
First, they gave a general MPC protocol for $n$ parties in the hones-but-curious model.
Second, they used zero-knowledge proofs to provide a general transformation or ``compiler'' from protocols that are only secure against honest-but-curious adversaries into ones that are secure against general (also known as \emph{malicious}) adversaries.
Since their work, the general paradigm of using zero-knowledge proofs to ``boost'' security from passive to active adversaries has found numerous uses in theory and practice. 

The details of \cite{GoldreichMW87}'s protocol are complex, and we omit them here.
However, some of the techniques are illustrated in Section~\ref{sec:crypto:it}, which describes a different multiparty secure computation protocol of Wigderson in the \emph{information-theoretic} setting.
In both cases, the general idea is that (1) we can describe a general function $F$ as a \emph{Boolean circuit}, which is a composition of simple \emph{gates}, and (2) once we do so, we can achieve a secure computation protocol by performing a gate-by-gate computation of intermediate values that are ``encrypted'' in the sense that no party (or strict subset of parties) can recover them on its own.

Arguably, it is \cite{GoldreichMW87}'s honest-but-curious to malicious \emph{compiler} which had had the most significant impact.
The idea behind this compiler is simple yet ingenious: Every party in the protocol will use zero-knowledge proofs to prove that it has followed the protocol's instructions.
For example, suppose that at a given step in the protocol, the party $i$ has private input $x_i$, and has received messages $m_1,\ldots,m_t$. 
Suppose that according to the protocol's instructions, the party should compute its next message as 
\begin{equation}
m_{t+1} = \Pi_i(x_i,m_1,\ldots,m_t) \label{eq:nextmessagefunc}
\end{equation} where $\Pi_i$ is some known polynomial-time function that is specified by the protocol.
A simple way for $i$ to convince the other parties that it computed $m_{t+1}$ correctly is to reveal all the inputs used in the computation, including the private input $x_i$. 
But that would, of course, violate $i$'s privacy.
However, the statement ``there exists $x_i$ that satisfies (\ref{eq:nextmessagefunc})'' is an $\NP$ statement.
Hence, it can be proved in \emph{zero knowledge}, and so in a way that does not reveal $x_i$.

While this is the general idea, implementing it involved additional complications, including ensuring consistency (that the same private input $x_i$ is used in all messages), dealing with randomized protocols (that are inherent to cryptography), and more.
Using tools such as \emph{commitment schemes} and \emph{coin-tossing protocols}, \cite{GoldreichMW87} overcame those obstacles and proved that (under standard cryptographic assumptions), there exists a secure multiparty computation protocol for \emph{every} polynomial time (potentially probabilistic) map  $F:\mathcal{X}_1\times \cdots \times{X}_n \rightarrow \mathcal{Y}_1\times \cdots \times \mathcal{Y}_n$.
This is one of the most fundamental theorems in all of cryptography and shows that if we are willing to allow for (polynomial-time) computation and communication overhead, \emph{every} protocol problem can be solved.
Although the road from such a theoretical proof of existence to practical constructions is long and arduous, the results of \cite{GoldreichMW87}, as well as the techniques they introduced, served as guiding lights for theorists and practitioners alike.

\subsection{Information Theoretic Cryptography} \label{sec:crypto:it}

In the previous section we showed how to construct zero-knowledge interactive proofs for all of $\NP$, and how to use them to construct secure multi-party computation protocols.  These works \cite{GoldreichMW86,GoldreichMW87} rely on cryptographic assumptions (such as the existence of a one-way function) and assume that the malicious parties are computationally bounded and cannot break the underlying cryptographic assumption.  In particular, these protocols do not offer {\em everlasting security}.  Namely, even if during the execution of the protocol the parties did not learn any information (beyond the validity of the statement or the output of the computation), if many years later the computers become stronger and manage to break the underlying cryptographic assumption, then at that point information can be leaked. 

This motivates the question of whether we can obtain everlasting security.  In other words, can we construct a zero-knowledge interactive proof and a secure $\MPC$ protocol that do not rely on cryptographic assumptions and provide security against all-powerful adversaries?  Wigderson, together with Ostrovsky \cite{OstrovskyW93}, proved that the answer is no for zero-knowledge interactive proofs. 
Namely, they showed that one-way functions are necessary for constructing zero-knowledge interactive proofs for all of $\NP$ (assuming $\NP\neq \BPP$).  Similarly, Wigderson, together with Ben-Or and Goldwasser \cite{Ben-OrGW88}, showed that cryptographic assumptions are necessary for secure $\MPC$ (with the security guarantees as presented in Section~\ref{sec:crypto:comp}).  

For Wigderson and his coauthors, these lower bounds were nothing but an invitation to surpass them.  In the case of zero-knowledge they managed to do so by changing the model in a clever and interesting way.  Specifically, to obtain information theoretic zero knowledge, Wigderson, together with Ben-Or, Goldwasser and Kilian~\cite{Ben-OrGKW88}, considered a new proof model:  Rather than considering a verifier that is interacting with a single prover, they considered a verifier that is interacting with {\em two non-communicating provers}.   They constructed  information theoretic zero-knowledge 2-prover interactive proofs.  We elaborate on this construction and on the immense impact of the 2-prover model in Section~\ref{sec:crypto:IT:MIP}.  For the case of secure $\MPC$, Wigderson, together with Ben-Or, Goldwasser \cite{Ben-OrGW88} managed to get an information theoretic security by restricting the fraction of parties the adversary is allowed to corrupt to be less than $1/2$ in the honest-but-curious setting and less than $1/3$ in the general malicious setting. They constructed an ingenious $\MPC$ protocol that is information theoretic secure assuming the adversary is restricted as above, and assuming that each pair of parties is connected via a secure channel. 
This result is a true breakthrough and has served as a foundation for numerous subsequent works.  We elaborate on it in Section~\ref{sec:crypto:IT:MPC}.

 \subsubsection{Multi-Prover Zero-Knowledge Interactive Proofs}\label{sec:crypto:IT:MIP}
In the multi-prover interactive proof model, there are multiple provers who can cooperate and communicate between 
them to decide on a common optimal strategy {\em before} the interaction with the verifier starts. But, 
once they start to interact with the verifier, they 
can no longer interact nor can they see the 
messages exchanged between the verifier and the 
other provers. 

\begin{definition}[2-prover interactive proof]
A 2-prover interactive proof for a language $L\in \NP$ consists of two provers $(P_1,P_2)$ and a probabilistic polynomial time verifier $V$ such that the verifier takes as input an instance $x$ and each prover takes as input both $x$ and a corresponding witness $w$.  The verifier samples two queries $(q_1,q_2)$ and sends $q_i$ to prover $P_i$.   Each prover $P_i$ computes an answer  $a_i=P_i(x,w,q_i)$ and sends $a_i$ to the verifier $V$, who then outputs a verdict bit $b=V(x,q_1,q_2,a_1,a_2)$ indicating accept or reject. 
 
 The following two properties are required to hold:
\begin{itemize}
    \item {\bf Completeness.}  For every $x\in L$ and any corresponding witness $w$ s.t.\ $(x,w)\in R_L$, the verifier $V(x)$, who generates queries $q_1$ and $q_2$, accepts the answers $a_1=P_1(x,w,q_1)$ and $a_2=P_2(x,w,q_2)$, with probability~$1$.
    \item {\bf Soundness.} For every $x\notin L$ and any two (computationally unbounded) cheating provers $P^*_1$ and $P^*_2$, the probability that the verifier $V(x)$, who generates queries $q_1$ and $q_2$, accepts the answers $a_1=P^*_1(x,q_1)$ and $a_2=P^*_2(x,q_2)$, is at most $1/2$. 
\end{itemize}
\end{definition}
\begin{theorem}
Every language $L\in \NP$ has a two-prover 
perfect zero-knowledge interactive proof-system. 
\end{theorem}

\noindent\textbf{\emph{Proof idea.}} Recall that in Section~\ref{sec:comp-ZK} we showed how to construct a computational zero-knowledge interactive proof.  The computational aspect follows from the use of a commitment scheme, whose hiding property is only computational.  

The main new ingredient in the 2-prover zero-knowledge proof is an information theoretic commitment.  Recall that in the zero-knowledge proof presented in Section~\ref{sec:comp-ZK}, in the first step the prover sends a commitment (in the case 3-coloring, this is a commitment to a legal coloring).  To achieve zero-knowledge we need this commitment scheme to be hiding and for soundness this commitment must be binding.  It is known that any commitment scheme that is statistically binding can only be {\em computationally hiding}, and this is precisely where the cryptographic hardness assumption comes in.  

Wigderson et.~al.~\cite{Ben-OrGKW88} get around this barrier by constructing a commitment scheme that is both statistically binding and statistically hiding in a model where there are two committers, who are assumed to be non-communicating.   
In what follows, we present a slightly simplified version of their commitment scheme.  We show how to commit to a single bit, and one can commit to arbitrarily many bits by repetition.

\subsubsection{Bit commitment scheme in the 2-prover setting.}

In what follows we show how two provers $P_1$ and $P_2$ commit to a bit $b\in\{0,1\}$. First, before the protocol begins, they share a random string $w\leftarrow\{0,1\}^n$ and a single random bit $d\leftarrow\{0,1\}$ which are hidden from the verifier~$V$.  $n$ controls the binding failure; taking $n=1$ will guarantee binding with probability $1/2$ whereas taking a general $n$ will guarantee binding with probability $2^{-n}$.
\begin{description}
\item {\em Commitment phase:}
\begin{enumerate}
\item The verifier $V$ chooses a random string $r\leftarrow\{0,1\}^n$, and sends $r$ to $P_1$.  He sends nothing to $P_2$.
\item Prover $P_1$ sends $x=(d\cdot r)\oplus w$ and the prover $P_2$ sends $z=b\oplus d$.

\end{enumerate}
\item {\em Opening phase:}
\begin{enumerate}
    \item Prover $P_2$ sends to the verifier the committed bit $b$ along with~$w$.
    \item The verifier $V$ accepts if and only if $x=((b\oplus z)\cdot r)\oplus w$.
\end{enumerate}
\end{description}

\myparagraph{\emph{Analysis.}}
In what follows we argue that this commitment scheme is information theoretic hiding and is also information theoretical binding (assuming the two provers do not interact).  
\begin{description}
    \item {\em Hiding.}  Note that $(x,z)\equiv U_{n+1}$ where $U_{n+1}$ denotes the uniform distribution over $n+1$ bits:  
    $$
    (x,z)=((d\cdot r)\oplus w, b\oplus d)\equiv (U_n, b\oplus d) \equiv U_{n+1}. 
    $$
    \item {\em Binding.}  We show that any pair of cheating provers can break the binding property with probability at most $2^{-n}$.  To this end, consider any pair of cheating provers $P^*_1$ and $P^*_2$ that send $(x,z)$ to $V$, and that later $P_2$ can open successfully to both $0$ using $w_0$ and to $1$ using $w_1$.  This means that 
    $$
    (z\cdot r)\oplus w_0= (1\oplus z)\cdot r\oplus w_1
    $$
    which in turn implies that
    $$
    w_0\oplus w_1=r,
    $$
    and thus $P^*_2$ can predict $r$, which should happen with probability $2^{-n}$. 
    
    \end{description}

\myparagraph{Information theoretic 2-prover zero-knowledge proof.} Equipped with this information theoretic commitment scheme, the 2-prover zero-knowledge construction is essentially the same as that presented in Section~\ref{sec:comp-ZK} while replacing the computational commitment scheme with the information theoretic one presented above.

\subsection{The Importance of the Multi-Prover Interactive Proof Model} \label{sec:PCP}

As mention above, the original motivation of \cite{Ben-OrGKW88} for considering the model of multi-prover interactive proofs was for constructing  statistical zero-knowledge proofs, a goal that they achieved with utter success.  However, already in their original paper, \cite{Ben-OrGKW88} realized the potential power of such a proof model, and they posed the following open problem: {\em ``It is interesting to consider what is the power of this new model solely with respect to language
recognition." }   Their intuition for why this model is powerful stems from the fact that {\em ``the verifier can check the provers against each other.''}  In particular, the example they give is that of suspects that try to cover up a crime.  It seems hard to cheat in a consistent manner.  

Indeed, Babai et.~al.~\cite{BabaiFL90}  showed that this proof model is extremely powerful.  In particular, they showed that the correctness of any (deterministic or non-deterministic) 
time-$T$ computation can be verified in a 2-prover interactive proof model, where the  communication complexity is only $\polylog(T)$  and where the running time of the verifier is only $\polylog(T)$ plus quasi-linear in the input length.  In particular, a polynomial time verifier can verify the correctness of exponentially long (deterministic or non-deterministic) computations.

The power of this proof model had groundbreaking consequences, leading to the notable PCP theorem \cite{FeigeGLSS91, BabaiFLS91, AroraS92, AroraLMSS92}.  In particular, Fortnow et.~al.~\cite{FortnowRS88} realized that if in the proof for a time-$T$ computation the messages from the verifier to the provers are of size $O(\log T)$ (which they indeed in the construction of ~\cite{BabaiFL90}) then each prover can generate a list consisting of its answers to {\em all} possible verifier messages, and this list will be of size $\poly(T)$.  The lists of the two provers can be thought of as a list of size $\poly(T)$ that can be verified by reading only two blocks of size $\polylog(T)$.  

This simple observation is spectacular.  In the context of non-deterministic computations, it means  that one can take any proof, and convert it into a new proof which is polynomially longer, but which can be verified by randomly reading only poly-logarithmically many bits of the proof.  Indeed this observation created an immense splash in the theory community,  leading to the 
notable PCP theorem which says that for any $\NP$ language $L$, with a corresponding $\NP$ relation $R$, it holds that there is an efficient transformation that given any $(x,w)\in R$ generates a {\em probabilistically checkable proof} (PCP) $\pi$ of size polynomial in the size of $(x,w)$, such that if $x\notin L$ then after (randomly) reading only {\em 3 bits} of $\pi$ the verifier will reject the proof with probability $7/8$, and if $(x,w)\in L$ and the proof $\pi$ was honestly generated then it is accepted by the probabilistic verifier with probability~$1$.

The 2-prover interactive proof model and the PCP theorem had numerous applications to theory of computation and beyond.  They form the foundation for all known  hardness of approximation results, and are at the heart of all known succinct computationally sound proof system (also known as argument systems).  
Succinct arguments have played an important role in cryptography over the last 15 years. This significance is underscored by the hundreds of papers that have been published, the dozens of systems that have been built, and their deployment by numerous blockchain corporations, including prominent ones like Ethereum.

\subsubsection{Information Theoretic Secure Multi-Party Computation}\label{sec:crypto:IT:MPC}

Recall that in \ref{sec:MPC-GMW} we elaborated on the work of \cite{GoldreichMW87}, which showed how a set of $n$ parties can compute any function of their (secret) inputs securely, where the security guarantees {\em computational}, i.e., security holds against {\em computationally bounded} adversaries. 
The focus of this section is on obtaining information theoretic security (also known as perfect security), where security holds even against {\em all powerful} adversaries. Indeed, Wigderson, together with Ben-Or and Goldwasser \cite{Ben-OrGW88}, showed that any function can be computed with {\em perfect security} assuming each pair of parties is connected with a private channel, as long as the malicious adversary controls less than $1/3$ of the parties. If the adversary is restricted to be semi-honest then security is guaranteed as long as the adversary controls a minority of the parties. Moreover, they showed that such a corruption rate is tight, both in the malicious setting and the semi-honest setting.

  This result is truly remarkable.  Indeed, it is a cornerstone in the field of secure multi-party computation, and has paved the way for a lot of subsequent work, making it a highly influential and a groundbreaking contribution.  As is often the case in the literature, we refer to this protocol as the $\BGW$ protocol.

\myparagraph{High-level overview of the $\BGW$ protocol.} 
Suppose a set of $n$ parties wish to securely compute a function $f$ from $n$ inputs to $n$ outputs. Let $C$ be an arithmetic circuit computing $f$. 
On a high level, the $\BGW$ protocol securely emulates the computation of~$C$ in a gate-by-gate manner, starting from the input gates all the way up to the output gates. More specifically, it proceeds with the following steps:
\begin{enumerate}
    \item {\bf Secret sharing of the inputs.}
The protocol begins with the parties sharing their inputs with
each other using a secure secret sharing scheme \cite{Shamir79}.  If the adversary is assumed to be semi-honest, and thus is assumed to corrupt less than $n/2$ parties, then the parties share each bit of their secret input using the Shamir  $t$-out-of-$n$ secret sharing scheme \cite{Shamir79} with $t=\lceil{n/2\rceil}-1$.  Such a scheme ensures that if at most $t$ shares are revealed then no information about the secret is revealed, whereas $t+1$ shares can be efficiently combined to reveal the secret.  We elaborate on Shamir's secret sharing scheme below.  

If the adversary is malicious then one needs to use a {\em verifiable secret
sharing } ($\VSS$) scheme.  The first information theoretically secure $\VSS$ scheme was constructed in the work of \cite{Ben-OrGW88}, and we elaborate on it towards the end of this section.

\item {\bf Gate-by-gate emulation.} The parties then emulate the computation of each gate of
the circuit, computing secret shares of the gate’s output from the secret shares of the gate’s inputs. As we shall see, the Shamir secret sharing scheme, as well as the $\VSS$ scheme of \cite{Ben-OrGW88}, have the property that computing shares corresponding to addition gates  can be done {\em locally}, without any interaction. Thus, the parties only interact in order to emulate the
computation of multiplication gates. This step is the most involved part of the protocol, and we elaborate on it below.

\item {\bf Output reconstruction.} 
 Finally,
the parties reconstruct the value of each output wire from the shares of the that wire.  Namely, if an output wire belongs to party $P_i$ then all the parties send party $P_i$ their shares corresponding to the wire and $P_i$ uses all these shares to reconstruct the output.
\end{enumerate}

We next describe the $\BGW$ protocol in more detail.  We first focus on the semi-honest setting, and only towards the end of the section we discuss the malicious setting. We start by recalling Shamir's secret sharing scheme.

\myparagraph{Shamir Secret Sharing Scheme.} Suppose a party (often referred to as the dealer) wishes to share a secret input among $n$ parties, with the guarantee that any $t+1$ of the parties can use their shares to efficiently reconstruct the secret and yet any $t$ shares do not reveal any information about the secret.  In what follows we assume for simplicity, and without loss of generality, that the secret is a single bit (and thus can be embedded in any finite field).

Let $\mathbb{F}$ be a finite field of size greater than $n$, let $\alpha_1,\ldots,\alpha_n$ arbitrary distinct non-zero elements in $\mathbb{F}$.
In order to share a secret $s\in\mathbb{F}$ a random degree $t$ polynomial $p(x) \in \mathbb{F}[x]$ is chosen such that $p(0)=s$.  The share of party $P_i$ is 
set to be $p(\alpha_i)$. By interpolation, given
any $t + 1$ points it is possible to reconstruct the polynomial $p$ and compute the secret $s = p(0)$. Furthermore, since $p$ is random subject to $p(0)=s$, and thus has $t$ random coefficients, its values at any $t$ or less of the $\alpha_i$’s give no information about the secret~$s$.

\myparagraph{Gate-by-Gate Emulation.} 
We next show how to use the structure of Shamir's secret sharing scheme to do the gate-by-gate emulation.  The first observation is
that addition gates can be computed locally. That is, given shares $p(\alpha_i)$ and $q(\alpha_i)$ of the two input wires corresponding to an addition gate, it holds that $r(\alpha_i) = p(\alpha_i) + q(\alpha_i)$ is a valid sharing of the output wire.
This is due to the fact that the polynomial $r(x)=p(x)+q(x)$ has the same
degree as both $p(x)$ and $q(x)$, and $r(0) = p(0) + q(0)$.
\begin{remark}\label{remark:BGW-linear}
    Note that if the function~$f$ is linear, and thus can be computed using a circuit that has only addition gates, then the gate-by-gate emulation step is completely non-interactive. 
\end{remark}
Regarding multiplication gates, a natural attempt is to compute the product of the shares, namely, party $P_i$ computes $r(\alpha_i) = p(\alpha_i) \cdot q(\alpha_i)$.  Indeed, the constant term is $r(0)=p(0)\cdot q(0)$, as desired.  However, the degree of $r(x)$ becomes $2t$, as opposed to~$t$.  This is a problem since the reconstruction algorithm works as long as
the polynomial used for the sharing is of degree at most~$t$.  We therefore need to reduce the degree of $r$ down to~$t$. To solve this, Wigderson and his coauthors devised a beautiful and elegant degree reduction protocol. This protocol also ensures that 
the new degree $t$ polynomial is a {\em random} degree $t$ polynomial with free coefficient $p(0)\cdot q(0)$.  This is crucial for security, since if the polynomial is not random then $t$ shares may reveal undesired information.

Specifically, the degree reduction protocol first  randomizes the degree-$2t$ polynomial~$r=p\cdot q$ so that it is uniformly distributed, and
then it reduces its degree to~$t$ while maintaining its uniformity. Specifically, a multiplication gate is computed via the following protocol: 
\begin{enumerate}
\item {\bf Local multiplication.} Each party locally multiplies its input shares. Namely, party $P_i$ computes $r(\alpha_i)=p(\alpha_i)\cdot q(\alpha_i)$.

\item  {\bf Randomizing the polynomial~$r$.} Each party $P_i$ generates a random degree $2t$ polynomial $z_i$ such that $z_i(0)=0$, and sends  to each party $P_j$ the share $z_i(\alpha_j)$. 
Then, each party $P_i$ adds all the shares it received and the original share it computed to obtain
$$\sum_{j=1}^n z_j(\alpha_i)+r(\alpha_i).$$  
The resulting shares define a random degree $2t$ polynomial $R$ such that $R(0)=p(0)\cdot q(0)$.

\item {\bf Degree reduction.} The parties run a private protocol where each party $P_i$ converts its share $R(\alpha_i)$ into the share $R_{\trunc}(\alpha_i)$, where $R_{\trunc}$ is simply a truncation of the polynomial $R$ to a degree $t$ polynomial.  Namely, if $R(x)=\sum_{j=0}^{2t} a_ix^i$ then $R_{\trunc}(x)=\sum_{j=0}^{t} a_ix^i$. 

A priori it is not clear how a party $P_i$ can compute $R_{\trunc}(\alpha_i)$ from $R(\alpha_i)$.  Indeed, $P_i$ cannot do this on its own, and needs the help of all other parties $P_j$, who have shares $R(\alpha_j)$. Note that this computation needs to be done in a private manner, which is the problem we are trying to solve in the first place, and thus it seems that we are back to square one! 
However, the magical observation of  \cite{Ben-OrGW88} is that this truncation function, which converts shares of $R$ to shares of $R_{\trunc}$, is {\em linear}.  As mentioned in Remark~\ref{remark:BGW-linear}, linear functions we already know how to compute securely since they do not require any multiplication gates!
Thus, it remains to argue the linearly of this function, which is argued in the claim below.  

\begin{claim}
   There exists a fixed (public) matrix $A\in \mathbb{F}^{n\times n}$ such that for every degree $2t$ polynomial $R:\mathbb{F}\rightarrow \mathbb{F}$ and for every distinct non-zero elements $\alpha_1,\ldots,\alpha_n \in \mathbb{F}$,
   $$
   A\cdot (R(\alpha_1),\ldots,R(\alpha_n))^T =(R_\trunc(\alpha_1),\ldots,R_\trunc(\alpha_n))^T, 
   $$
   where $R_\trunc$ is defined as above.
\end{claim}

\begin{proof} Let $\vec{R}=(R_0,R_1,\ldots,R_{2t},0,\ldots,0)\in \mathbb{F}^n$ denote the vector of coefficients of the polynomial~$R$.
Let $V_{\vec{\alpha}}$ be the Vondermonde matrix corresponding to $\vec{\alpha}=(\alpha_1,\ldots,\alpha_n)$. Namely, for every $i,j\in [n]$,  $V_{\vec{\alpha}}(i,j)=\alpha_{i}^{j-1}$. Note that 
$$
V_{\vec{\alpha}}\cdot \vec{R}^T =  (R(\alpha_1),\ldots,R(\alpha_n))^T
$$
It is well known that the Vondermonde matrix $V_\alpha$ is invertible  if $\alpha_1,\ldots,\alpha_n \in \mathbb{F}$ are all distinct and non-zero.  Therefore,
\begin{equation}\label{eqn:BGW-R}
\vec{R}^T=V^{-1}_{\vec{\alpha}}\cdot (R(\alpha_1),\ldots,R(\alpha_n))^T.
\end{equation}
Let $P$ be the linear projection function that takes as input a vector $(a_1,\ldots,a_n)\in\mathbb{F}^n$ and outputs $(a_1,\ldots,a_{t+1},0\ldots,0)\in\mathbb{F}^n$.  Namely, in matrix representation, $P(i,j)=1$ if $i=j$ and both are in $\{1,\ldots,t+1\}$, and $P(i,j)=0$ otherwise.
Thus, denoting by $\vec{R}_\trunc$ the $t+1$ coefficients of the degree~$t$ polynomial $R_\trunc$ followed by zeros, i.e.,  $\vec{R}_\trunc=(R_0,R_1\ldots,R_{t},0,\ldots,0)$, by the definition of $P_t$ and by Equation~\eqref{eqn:BGW-R}
$$
\vec{R}_\trunc^T=P\cdot V^{-1}_{\vec{\alpha}}\cdot (R(\alpha_1),\ldots,R(\alpha_n))^T.
$$
This, together with the definition of $V_{\vec{\alpha}}$, implies that 
$$
(R_\trunc(\alpha_1),\ldots,R_\trunc(\alpha_n))^T= V_{\vec{\alpha}}\cdot P\cdot V^{-1}_{\vec{\alpha}}\cdot (R(\alpha_1),\ldots,R(\alpha_n))^T.
$$
We thus conclude the proof of this claim by setting $A=V_{\vec{\alpha}}\cdot P\cdot V^{-1}_{\vec{\alpha}}$.
\end{proof}

\end{enumerate}

This concludes the description of the $\BGW$ protocol in the honest-but-curious setting, where the adversary is assumed to follow the protocol.  We next show how Wigderson and his co-authors modify this protocol to obtain security against a {\em malicious} adversary who controls less than $1/3$ of the parties and may deviate arbitrarily from the protocol.  
The main new tool is a \emph{verifiable secret sharing} scheme.

\subsubsection{Verifiable Secret Sharing}

A verifiable secret sharing ($\VSS$) scheme, originally defined by Chor et al. \cite{CGMA85}, is a secret sharing scheme that is secure even in the presence of {\em malicious} adversaries. Recall that a secret sharing scheme (with threshold~$t$)
is made up of two stages:  A sharing stage and a reconstruction stage.  In the sharing stage, the dealer shares a secret among the $n$ parties so that 
any $t + 1$ parties can efficiently reconstruct the secret from their shares, while any subset of $t$ or fewer shares reveal no information about the secret. In the reconstruction stage, a set of $t + 1$ or more parties reconstruct the secret. If we consider Shamir’s secret-sharing scheme, much can go
wrong if the dealer or some of the parties are malicious. Recall, that to share a secret $s$, the dealer is supposed to choose a random
polynomial $q$ of degree $t$ with $q(0) = s$ and then hand each party $P_i$
its share $q(\alpha_i)$. However, a malicious dealer can choose a polynomial of higher degree, and as a result different subsets of $t + 1$ parties may reconstruct different values. Thus, the shared
value is not well defined. In addition,  in the reconstruction phase a corrupted party can provide an arbitrary malicious value instead of the prescribed value $q(\alpha_i)$, thus effectively changing the value
of the reconstructed secret.

A verifiable secret sharing scheme is aimed at solving precisely these issues.  Chor et al. \cite{CGMA85} constructed a $\VSS$ scheme with {\em computational} security, i.e., assuming the malicious parties are computationally bounded (and assuming the hardness of some computational problem). 
Wigerson with his coauthors \cite{Ben-OrGW88} constructed an  information theoretically secure $\VSS$ scheme, which ensures security against {\em all powerful} adversaries, assuming that at most  $t < n/3$ of the parties are corrupted.  More specifically, the security guarantee is that the shares received by the honest parties are guaranteed to be $q(\alpha_i)$ for a well-defined degree-$t$ polynomial $q$, even if
the dealer is corrupted.  To achieve this guarantee, the secret sharing stage is followed by a verification stage which is an interactive stage where the parties ``correct" their shares if needed.  This correction protocol, which we elaborate on below, is extremely beautiful!   

Given this security guarantee it is possible to use techniques from the field of error-correcting codes in order
to reconstruct $q$ (and thus $q(0) = s$) as long as $n-t$ correct shares are provided and $t < n/3$. This
is due to the fact that Shamir’s secret-sharing scheme when looked at in this context is exactly a
Reed-Solomon code.

\myparagraph{$\VSS$ via bivariate polynomials.}  

The $\VSS$ protocol of \cite{Ben-OrGW88} consists of three stages.
\begin{enumerate}
    \item {\bf Secret sharing stage.}  Loosely speaking, in this stage the dealer embeds the Shamir's secret sharing polynomial in a 
{\em bivariate} polynomial $S(x, y)$. Specifically, to share a secret $s\in\{0,1\}$ the dealer chooses a random bivariate polynomial $S(x, y)$ of degree $t$ in each variable, such that 
$S(0, 0) = s$.  Note that $q(z)=S(0,z)$ is a polynomial corresponding to the Shamir secret sharing scheme and the values $q(\alpha_1),...,q(\alpha_n)$ are the Shamir shares embedded into
$S(x, y)$. Similarly, $p(z)=S(z,0)$ is a polynomial corresponding to the Shamir secret sharing scheme and the values $p(\alpha_1),...,p(\alpha_n)$ are the Shamir shares embedded into $S(x, y)$. The dealer sends each party $P_i$ {\em two} univariate polynomials as shares;
these polynomials are $f_i(x) = S(x, \alpha_i)$ and $g_i(y) = S(\alpha_i,y)$.  The Shamir-share of party $P_i$ is $f_i(0)=S(0,\alpha_i)$, and the polynomials $f_i$ and $g_i$ are given only for the sake of verification.

\item {\bf Verification stage.}  At this point the parties engage in an interactive verification protocol. First, each party $P_i$ sends each party $P_j$ the value $s_{i,j}= f_i(\alpha_j)$.  Note that if the dealer is honest, then the elements $s_{1,j},\ldots,s_{n,j}$ sent to party $P_j$ should be the value of the polynomial $g_j$ on  $\alpha_1,\ldots,\alpha_n$, respectively. Each party $P_j$ checks that indeed for every $i\in [n]$, $s_{i,j}=g_j(\alpha_i)$.  If this is not the case, it broadcasts a request for the dealer to reveal $s_{i,j}$.  If $P_j$ has more than $t$ requests then the dealer is clearly malicious, in which case $P_j$ broadcasts a ``complaint" thereby asking the dealer to reveal his private shares $f_i$ and $g_i$.  Finally, after the dealer broadcasts all the information requested, each party $P_i$ checks that all the public and private information he received from the dealer are consistent. If $P_i$ finds any inconsistencies he broadcasts a complaint thereby asking the dealer to reveal his private shares. If at this point more than $t$ parties have asked to make their shares public, the dealer is clearly malicious and all the parties pick the default zero polynomial as
the dealer’s polynomial. Likewise, if the dealer did not answer all the broadcasted requests he is declared malicious. On the other hand, if $t$ or less parties have complained then there are at least $t + 1$ honest parties who are satisfied (this follows from the fact that $t<n/3$).  The shares of these parties uniquely define a polynomial $S(x,y)$ and this polynomial conforms with all the information that was made public (otherwise one of these honest parties would have complained). In this case the complaining parties take the public information as their
share.

\item {\bf Reconstruction stage.}  At this point each party sends its updated share $f_i(0)$, and the secret $s$ is reconstructed by running the Reed-Solomon decoding algorithm. 
\end{enumerate}

Note that if the dealer is honest then no information about the shares of any honest party is revealed during the verification process. If however
the dealer is malicious, we do not need to protect the privacy of his information, and
the verification procedure ensures that all the honest parties values lie on some
polynomial of degree~$t$.

\myparagraph{Gate-by-Gate Emulation in the Malicious Setting}

As was done in the honest-but-curious setting, addition gates are computed locally by adding the corresponding shares, whereas computing multiplication gates is significantly more involved. 
Recall that in the honest-but-curious setting the multiplication step was done by multiplying the shares locally, thus obtaining shares of a degree-$2t$ polynomial.  Then the parties rerandomized this polynomial and then truncated it.  In the malicious setting, the rerandomization step needs to be made secure against malicious adversaries.  In addition, to apply the degree reduction step, we need to argue that the truncation is a linear function, but for this 
we must make sure that the all the parties use as input to this function their correct
point on the product polynomial $h(x) = f(x)g(x)$. To guarantee that this is indeed
the case, error correcting codes are used yet again.

\section{Pseudorandomness} \label{sec:derand}

\def\now{0}
\def\solutions{0}

\spnewtheorem{fact}[theorem]{Fact}{\bfseries}{\itshape}
\spnewtheorem{openprob}[theorem]{Open Problem}{\bfseries}{\itshape}
\spnewtheorem{compprob}[theorem]{Computational Problem}{\bfseries}{\itshape}

\ifnum\now=1
\newtheorem{problem}{Problem}[chapter] %
\newtheorem{algorithm}[theorem]{Algorithm}
\newtheorem{construction}[theorem]{Construction}
\newtheorem{example}[theorem]{Example}
\newtheorem{remark}[theorem]{Remark}
\else
\spnewtheorem{remk}[theorem]{Remark}{\bfseries}{\itshape}
\spnewtheorem{exmp}[theorem]{Example}{\bfseries}{\itshape}
\spnewtheorem{algorithm}[theorem]{Algorithm}{\bfseries}{\rmfamily}
\spnewtheorem{construction}[theorem]{Construction}{\bfseries}{\itshape}

\fi

\newcommand{\alginput}[1]{\ \\ \noindent Input: #1\\ \ \\}

\newenvironment{algenumerate}{\vspace{-5ex}\begin{enumerate}}{\end{enumerate}}

\def\FullBox{\hbox{\vrule width 8pt height 8pt depth 0pt}}

\def\qed{\ifmmode\qquad\FullBox\else{\unskip\nobreak\hfil
\penalty50\hskip1em\null\nobreak\hfil\FullBox
\parfillskip=0pt\finalhyphendemerits=0\endgraf}\fi}

\def\qedsketch{\ifmmode\Box\else{\unskip\nobreak\hfil
\penalty50\hskip1em\null\nobreak\hfil$\Box$
\parfillskip=0pt\finalhyphendemerits=0\endgraf}\fi}

\newenvironment{proofof}[1]{\begin{proof}[of #1]}{\end{proof}}

\newenvironment{proofsketch}{\begin{proof}[sketch]}{\end{proof}}

\ifnum\solutions=0
\excludecomment{solution}
\fi

\newenvironment{claimproof}{\begin{quotation} \noindent
{\bf Proof of claim:~~}}{\qedsketch\end{quotation}}

\newcommand{\eqdef}{\mathbin{\stackrel{\rm def}{=}}}
\newcommand{\F}{{\mathbb F}}
\newcommand{\Efield}{{\mathbb E}}
\newcommand{\Q}{{\mathbb Q}}
\newcommand{\suchthat}{{\;\; : \;\;}}
\newcommand{\pr}[2][]{\Pr_{#1}\left[#2\right]}
\newcommand{\deffont}[1]{{\em #1}}
\newcommand{\getsr}{\mathbin{\stackrel{\mbox{\tiny R}}{\gets}}}
\newcommand{\Exp}{\mathop{\mathrm{E}\/}\displaylimits}
\newcommand{\Var}{\mathop{\mathrm{Var}\/}\displaylimits}
\newcommand{\xor}{\oplus}

\def\textprob#1{\textmd{\textsc{#1}}}
\newcommand{\mathprob}[1]{\mbox{\textmd{\textsc{#1}}}}
\newcommand{\QuadRes}{\textprob{Quadratic Residuosity}}
\newcommand{\Sampling}{\textprob{Sampling}}
\newcommand{\OracleApprox}[1]{\textprob{$[+#1]$-Approx Oracle Average}}
\newcommand{\CircApprox}[1]{\textprob{$[+#1]$-Approx Circuit Average}}
\newcommand{\CircRelApprox}[1]{\textprob{$[\times (1+#1)]$-Approx \#CSAT}}
\newcommand{\DNFRelApprox}[1]{\textprob{$[\times (1+#1)]$-Approx \#DNF}}
\newcommand{\CA}[1]{\mathprob{CA}^{#1}}
\newcommand{\CRA}[1]{\mathprob{CSAT}^{#1}}

\newcommand{\GraphNoniso}{\textprob{Graph Nonisomorphism}}
\newcommand{\GNI}{\mathprob{GNI}}
\newcommand{\GraphIso}{\textprob{Graph Isomorphism}}
\newcommand{\GI}{\mathprob{GI}}
\newcommand{\MinCut}{\textprob{MinCut}}
\newcommand{\MaxCut}{\textprob{MaxCut}}
\newcommand{\MaxSAT}{\textprob{MaxSAT}}
\newcommand{\MaxFlow}{\textprob{Max Flow}}
\newcommand{\IdentityTest}{\textprob{Polynomial Identity Testing}}
\newcommand{\ArithCircIdentityTest}{\textprob{Arithmetic Circuit Identity Testing}}
\newcommand{\GraphConn}{\textprob{Graph Connectivity}}
\newcommand{\Primality}{\textprob{Primality Testing}}
\newcommand{\USTConn}{\textprob{Undirected S-T Connectivity}}
\newcommand{\STConn}{\textprob{S-T Connectivity}}
\newcommand{\PerfMatch}{\textprob{Perfect Matching}}
\newcommand{\PM}{\mathprob{PM}}
\newcommand{\ACIT}{\mathprob{ACIT}}
\newcommand{\LargeCut}{\textprob{Large Cut}}
\newcommand{\Determinant}{\textprob{Determinant}}
\newcommand{\MatrixMult}{\textprob{Matrix Multiplication}}
\newcommand{\Par}{\mathprob{Par}}
\newcommand{\Majority}{\textprob{Majority}}
\newcommand{\Permanent}{\textprob{Permanent}}
\newcommand{\PolynomialFactorization}{\textprob{Polynomial Factorization}}
\newcommand{\ApproxCountMatchings}{\textprob{Approximately Counting Matchings}}

\newcommand{\class}[1]{\mathsf{#1}}
\newcommand{\ioclass}[1]{\class{i.o.\mbox{-}#1}}
\newcommand{\coclass}[1]{\class{co\mbox{-}#1}} %
\newcommand{\RP}{\class{RP}}
\newcommand{\coRP}{\coclass{RP}}
\newcommand{\ZPP}{\class{ZPP}}
\newcommand{\RNC}{\class{RNC}}
\newcommand{\RL}{\class{RL}}
\renewcommand{\L}{\class{L}}
\newcommand{\coRL}{\coclass{RL}}
\newcommand{\BPL}{\class{BPL}}
\newcommand{\NL}{\class{NL}}
\newcommand{\NSPACE}{\class{NSPACE}}
\newcommand{\SPACE}{\class{SPACE}}
\newcommand{\DSPACE}{\class{DSPACE}}
\newcommand{\IP}{\class{IP}}
\newcommand{\AM}{\class{AM}}
\newcommand{\MA}{\class{MA}}
\renewcommand{\P}{\class{P}}
\newcommand\prBPP{\class{prBPP}}
\newcommand\prRP{\class{prRP}}
\newcommand\prAM{\class{prAM}}
\newcommand\prSUBEXP{\class{prSUBEXP}}
\newcommand\prP{\class{prP}}
\newcommand{\Ppoly}{\class{P/poly}}
\newcommand{\ioPpoly}{\ioclass{P/poly}}
\newcommand{\DTIME}{\class{DTIME}}
\newcommand{\TIME}{\class{TIME}}
\newcommand{\ETIME}{\class{E}}
\newcommand{\NE}{\class{NE}}
\newcommand{\coNE}{\coclass{NE}}
\newcommand{\BPTIME}{\class{BPTIME}}
\newcommand{\EXP}{\class{EXP}}
\newcommand{\SUBEXP}{\class{SUBEXP}}
\newcommand{\BPSUBEXP}{\class{BPSUBEXP}}
\newcommand{\ioSUBEXP}{\ioclass{SUBEXP}}
\newcommand{\qP}{\class{\tilde{P}}}
\newcommand{\PH}{\class{PH}}
\newcommand{\NC}{\class{NC}}
\newcommand{\NCO}{\class{NC^0}}
\newcommand{\PSPACE}{\class{PSPACE}}
\newcommand{\quasiP}{\class{\tilde{P}}}
\newcommand{\ACO}{\class{AC^0}}
\newcommand{\ACOPAR}{\class{AC^0[2]}}
\newcommand{\ACC}{\class{ACC}}
\newcommand{\BPACO}{\class{BPAC^0}}
\newcommand{\qACO}{\class{\widetilde{AC}^0}}
\newcommand{\PP}{\class{PP}}
\newcommand{\NEXP}{\class{NEXP}}
\newcommand{\NSUBEXP}{\class{NSUBEXP}}
\newcommand{\MAEXP}{\class{MAEXP}}
\newcommand{\coNP}{\coclass{NP}}
\newcommand{\ShP}{\class{\#P}}

\newcommand{\Diam}{\mathrm{Diam}}
\newcommand{\cut}{\mathrm{cut}}
\newcommand{\CP}{\Col}
\newcommand{\Supp}{\mathrm{Supp}}
\newcommand{\indeg}{\mathrm{indeg}}
\newcommand{\outdeg}{\mathrm{outdeg}}
\newcommand{\hit}{\mathrm{hit}}

\newcommand{\fail}{\mathtt{fail}}
\newcommand{\halt}{\mathtt{halt}}

\newcommand{\HFam}{\mathcal{H}}
\newcommand{\FFam}{\mathcal{F}}
\newcommand{\GFam}{\mathcal{G}}
\newcommand{\Dom}{\mathcal{D}}
\newcommand{\Rng}{\mathcal{R}}

\newcommand{\Hall}{\mathrm{H}}
\newcommand{\Hmin}{\mathrm{H}_{\infty}}
\newcommand{\avgHmin}{\mathrm{H}_{\infty}}
\newcommand{\HRen}{\mathrm{H}_2}
\newcommand{\HSha}{\mathrm{H}_{\mathit{Sh}}}
\newcommand{\Ext}{\mathrm{Ext}}
\newcommand{\Disp}{\mathrm{Disp}}
\newcommand{\Con}{\mathrm{Con}}
\newcommand{\Samp}{\mathrm{Samp}}
\newcommand{\Code}{\mathcal{C}}
\newcommand{\IsPath}{\textrm{IsPath}}
\newcommand{\hatn}{{\hat{n}}}
\newcommand{\hatf}{{\hat{f}}}
\newcommand{\hatell}{{\hat{\ell}}}
\newcommand{\hatL}{{\hat{L}}}
\newcommand{\hatA}{{\hat{A}}}
\newcommand{\Ham}{d_H}
\newcommand{\oLIST}{\mathrm{LIST}}
\newcommand{\cLIST}{\overline{\mathrm{LIST}}}
\newcommand{\LIST}{\mathrm{LIST}}
\newcommand{\Agr}{\mathrm{agr}}
\newcommand{\ECC}{\mathrm{ECC}}
\newcommand{\Amp}{\mathrm{Amp}}
\newcommand{\Red}{\mathrm{Red}}
\newcommand{\Contra}{{\Rightarrow\Leftarrow}}

\newcommand{\cond}[3]{#1\rightarrow_{#3}#2}

\newcommand{\zigzag}{\mathbin{\raisebox{.2ex}{
      \hspace{-.4em}$\bigcirc$\hspace{-.75em}{\rm
      z}\hspace{.15em}}}}
\newcommand{\replacement}{\mathbin{\raisebox{.2ex}{
      \hspace{-.4em}$\bigcirc$\hspace{-.75em}{\rm
      r}\hspace{.15em}}}}

\newcommand{\para}{\parallel}

\newcommand{\eps}{\varepsilon}
\newcommand{\ci} {\stackrel{\rm{c}}{\equiv}}
\newcommand{\Time}{\mathrm{time}}
\newcommand{\Space}{\mathrm{space}}
\newcommand{\Adv}{\mathrm{adv}}
\newcommand{\inner}[2]{\left[#1,#2\right]}
\newcommand{\tDec}{t_{\mathrm{Dec}}}
\renewcommand{\bbH}{\mathbb{H}}
\newcommand{\card}[1]{|#1|}

\newcommand{\GL}{\operatorname{GL}}
\newcommand{\maj}{\operatornamewithlimits{maj}}

\newcommand{\tO}{\tilde{O}}
\newcommand{\tTheta}{\tilde{\Theta}}
\newcommand{\tOmega}{\tilde{\Omega}}

\newcommand{\Bip}{\mathrm{Bip}}

\newcommand{\liftB}{\tilde{B}}
\newcommand{\liftA}{\hat{A}}

\newcommand{\Class}{\mathcal{C}}
\newcommand{\IIDBits}{\mathrm{IIDBits}}
\newcommand{\IndBits}{\mathrm{IndBits}}
\newcommand{\UnpredBits}{\mathrm{UnpredBits}}

\newcommand{\Kmax}{\mathit{K_{max}}}
\newcommand{\Mon}{\mathrm{Mon}}
\newcommand{\kmax}{\mathit{k_{max}}}

\newcommand{\Uni}{\mathcal{U}}

\renewcommand{\Im}{\mathrm{Im}}

\newcommand{\uniform}[1]{U_{#1}}

\newcommand{\nodes}{\#\mathrm{nodes}}

\newcommand{\Bin}{\mathrm{Bin}}
\newcommand{\PBin}{\mathrm{PBin}}

\newcommand{\Lovasz}{{Lov{\'a}sz}}
\newcommand{\Erdos}{{Erd{\H o}s}}
\newcommand{\Renyi}{R\'enyi}
\newcommand{\Hastad}{{H{\aa}stad}}

\newcommand{\INW}{\mathrm{INW}}
\newcommand{\fH}{\mathcal{H}}
\newcommand{\G}{\mathbf{G}}
\newcommand{\Gt}{\widetilde{\mathbf{G}}}
\newcommand{\bH}{\mathbf{H}}

\newcommand{\gen}{\INW}
\newcommand{\circleds}{\mathbin{\circledS}}

A major theme in Wigderson's work is to understand the power of randomness in efficient computation, addressing questions such as:
\begin{itemize}
\item Can randomized algorithms solve problems much more efficiently than deterministic algorithms, or can every randomized algorithm be converted into a deterministic algorithm with only a small loss in efficiency?
\item Can we give explicit, deterministic constructions of combinatorial objects whose existence is proven via the Probabilistic Method?
\item Can we convert weak random sources, which may have biases and correlations, into high-quality random bits that can be used for running randomized algorithms or protocols?
\end{itemize}
In this section, we will survey the answers that Wigderson's work has given to these fundamental questions, and the close connections between the questions that he has helped uncover and exploit.  For more details, we recommend the broader surveys of pseudorandomness~\cite{Goldreich10-prgprimer,Vadhan12-fnttcs,HatamiHo23}.\footnote{Some of our text is taken verbatim from \cite{Vadhan12-fnttcs}, with permission.}

\subsection{Hardness vs. Randomness}

\subsubsection{Motivation}

In the 1970's and 1980's, randomization was discovered to be an extremely powerful tool in theoretical computer science.  By allowing algorithms to ``toss coins,'' we could potentially solve problems much more efficiently than before.  In particular, polynomial-time randomized algorithms were found for a number of problems that were only known to have exponential-time deterministic algorithms, such as \PolynomialFactorization\ (over finite fields)~\cite{Berlekamp70}, \Primality~\cite{SolovaySt77,Miller76,Rabin80}, \IdentityTest~\cite{DeMilloLi78,Schwartz80,Zippel79}, and \ApproxCountMatchings\ in graphs~\cite{JerrumSi89}.    However, it was unclear whether this apparent exponential savings provided by randomization was real, or just a reflection of our ignorance: could there be polynomial-time deterministic algorithms for these problems that we just hadn't discovered or proven correct yet?  For example, already Miller~\cite{Miller76} gave a deterministic polynomial-time algorithm for \Primality\ based on the Extended Riemann Hypothesis, and three decades later, Agrawal, Kayal, and Saxena~\cite{AgrawalKaSa04} gave an unconditional deterministic polynomial-time algorithm.  Thus, the following problem remained open:

\begin{openprob}  \label{openprob:BPPneqP} Are there problems that can be solved by randomized algorithms in polynomial time that cannot be solved by deterministic algorithms in polynomial time? \end{openprob}

We now formalize this question using complexity classes that capture the power of efficient deterministic and randomized algorithms. As is common in complexity theory,
these classes are defined in terms of decision problems, where instances are given by binary strings $x\in \zo^* \eqdef \bigcup_{n=0}^\infty \zo^n$ and the set of instances where the answer should be ``yes'' is specified by a
\deffont{language} $L\subseteq \zo^*$, or equivalently a boolean function $f : \zo^*\rightarrow \zo$.
However, the definitions generalize
in natural ways to other types of computational problems, such as computing
functions or solving search problems.

Recall that we say a deterministic algorithm $A$ {\em runs in time}  $t : \N\rightarrow
\N$ if $A$ takes at most $t(|x|)$ steps on every input $x$ (where $|x|$ is the length of $x$ in bits), and it
{\em runs in polynomial time} if it runs time $t(n)=O(n^c)$ for a
constant $c$.  Polynomial time is a theoretical approximation to
feasible computation, with the advantage that it is robust to
reasonable changes in the model of computation and representation of
the inputs.

\begin{definition} $\P$ is the class of languages $L$ for which there exists a
deterministic polynomial-time algorithm $A$ such that for all instances $x$,
\begin{itemize}
\item $x\in L \Rightarrow \mbox{$A(x)$ accepts}$, and
\item $x\notin L \Rightarrow \mbox{$A(x)$ rejects}$.
\end{itemize}
\end{definition}

\begin{definition} $\BPP$ is the class of languages $L$ for which there exists a
probabilistic polynomial-time algorithm $A$ such that
for all instances $x$,
\begin{itemize}
\item  $x \in L \Rightarrow \Pr[\mbox{$A(x)$ accepts}] \geq 2/3$, and

\item  $x \not\in L \Rightarrow \Pr[\mbox{$A(x)$ accepts}] \leq 1/3$,
\end{itemize}
where the probabilities are taken over the random coin tosses of the algorithm $A$.
\end{definition}

The choice of the thresholds $\ell=1/3$ and $u=2/3$ is arbitrary, and any two distinct constants $\ell<u$ yields an equivalent definition, since the error probability of a randomized algorithm can be made arbitrarily small by running the algorithm several times and accepting if at least an $(\ell+u)/2$ fraction of the executions accept.

The cumbersome notation $\BPP$ stands for ``bounded-error
probabilistic polynomial-time,'' due to the unfortunate fact that
$\PP$ (``probabilistic polynomial-time'') refers to the definition
where the inputs in $L$ are accepted with probability
greater than 1/2 and inputs not in $L$ are accepted with
probability at most 1/2. Despite its name, $\PP$ is not a reasonable
model for randomized algorithms, as it takes exponentially many
repetitions to reduce the error probability. $\BPP$ is considered
the standard complexity class associated with probabilistic
polynomial-time algorithms, and thus a driving question of Wigderson's work surveyed in this section is the following formalization of Open Problem~\ref{openprob:BPPneqP} (negated).

\begin{openprob} \label{openprob:BPPequalsP} Does $\BPP=\P$?   \end{openprob}

More generally, we are interested in quantifying how much savings randomization provides. One
way of doing this is to find the smallest possible upper
bound on the deterministic time complexity of languages in $\BPP$.
For example, we would like to know which of the following complexity classes
contain $\BPP$:

\begin{definition}[Deterministic Time Classes]
\label{def:DTIME}
\vspace{-4ex}
$$
\begin{array}{rcll}
\DTIME(t(n)) &=& \multicolumn{2}{l}{\{L : \mbox{$L$ can be decided deterministically in
time $O(t(n))$}\}} \\
\P & = & \cup_c\DTIME(n^c) & \mbox{(``polynomial time'')}\\
\quasiP &=&  \cup_c\DTIME(2^{(\log{n})^c}) & \mbox{(``quasipolynomial time'')}\\
\SUBEXP &=& \cap_{\varepsilon}\DTIME(2^{n^\varepsilon}) & \mbox{(``subexponential time'')}\\
\EXP & = & \cup_c \DTIME(2^{n^c}) & \mbox{(``exponential time''),}
\end{array}
$$
where the unions and intersections are taken over all $c,\eps\in (0,\infty)$
\end{definition}

As a baseline, we can always remove randomization with at most an exponential slowdown:

\begin{proposition} \label{prop:BPPinEXP}
$\BPP\subseteq \EXP.$
\end{proposition}

\begin{proof}
If $L$ is in $\BPP$, then there is a probabilistic polynomial-time
algorithm $A$ for $L$ running in time $t(n)$ for some polynomial $t$.
Let $m(n)\leq t(n)$ be an upper bound on the number of random bits used by $A$ on inputs
of length $n$.  Thus we can view $A$ as a deterministic algorithm on two inputs
--- its regular input $x\in \zo^n$ and its coin tosses $r\in \zo^{m(n)}$.
Writing
$A(x;r)$ for $A$'s output on input $x\in \zo^n$ and
coin tosses $r\in \zo^{m(n)}$, we have
\[ \Pr_r[A(x;r)\mbox{
accepts}]=\frac{1}{2^{m(n)}}\sum_{r\in\{0,1\}^{m(n)}}A(x;r)\]
We can compute the right-hand side of the above expression in deterministic time
$2^{m(n)}\cdot t(n)$.
\end{proof}

We see that the enumeration method is {\em general} in that it
applies to all $\BPP$ algorithms, but it is {\em infeasible} (taking
exponential time).  However, if the algorithm uses only a small
number of random bits, it becomes feasible:

\begin{proposition} \label{prop:enumeration}
If $L$ has a probabilistic polynomial-time algorithm that runs in
time $t(n)$ and uses $m(n)$ random bits, then $L\in \DTIME(t(n)\cdot
2^{m(n)})$.  In particular, if $t(n)=\poly(n)$ and $m(n)=O(\log n)$,
then $L\in \P$.
\end{proposition}

Thus an approach to proving $\BPP=\P$ is to show that the number of random
bits used by any $\BPP$ algorithm can be reduced to $O(\log n)$. This is the angle of attack
pursued in Wigderson's work, as surveyed in the next section. However, to date,
Proposition~\ref{prop:BPPinEXP} remains the best unconditional upper-bound we
have on the deterministic time-complexity of $\BPP$.

\begin{openprob} Is $\BPP$ ``closer'' to $\P$ or $\EXP$?
Is $\BPP\subseteq \quasiP$?  Is $\BPP\subseteq \SUBEXP$?
\end{openprob}

\subsubsection{Wigderson's Contributions} 

\paragraph{\underline{Derandomization from Circuit Lower Bounds}}
In the early 1980's, the answer to Open Problem~\ref{openprob:BPPequalsP} seemed very likely to be no, that $\BPP\neq \P$, since there were many examples of problems where randomization provided an exponential speedup over the best deterministic algorithms known at the time.  The first evidence that randomization might not be so powerful came from Yao~\cite{Yao82}, who showed that if there exist ``cryptographically secure'' pseudorandom generators, as defined by Blum and Micali~\cite{BlumMi84}, then $\BPP\subseteq \SUBEXP$.  In a series of works, Wigderson and his collaborators obtained much stronger derandomization results, convincing the theoretical computer science community that indeed $\BPP=\P$.

\begin{theorem}[\cite{NisanWi94,BabaiFoNiWi93,ImpagliazzoWi97}] \label{thm:IW}
\begin{enumerate}
    \item If $\EXP$ has a function of circuit complexity $n^{\omega(1)}$, then $\BPP\subseteq \SUBEXP$. \label{itm:lowend}
    \item If $\ETIME\eqdef\DTIME(2^{O(n)})$ has a function of circuit complexity $2^{\Omega(n)}$, then $\BPP=\P$. \label{itm:highend}
\end{enumerate}
\end{theorem}
To define the ``circuit complexity'' referred to in the theorem, we associate a language $L\subseteq \zo^*$ in $\EXP$ or $\ETIME$ with its characteristic function $f : \zo^*\rightarrow \zo$. For each $n$, we consider the restriction $f_n : \zo^n\rightarrow \zo$ of $f$ to instances of length $n$, and ask how many boolean operations (AND, OR, NOT) are needed to compute $f_n$, i.e. what is the size of the smallest {\em boolean circuit} computing $f_n$, as a function of $n$. (See Section~\ref{sec:circuit-complexity} for a more formal definition.)  An algorithm running in time $t(n)$ can be simulated by boolean circuits of size $\tO(t(n))\eqdef t(n)\cdot \polylog(t(n))$, but the converse is not true, since circuits are a {\em nonuniform} model of computation, essentially allowing a different program for each input length (rather than a single set of instructions that can solve problems of arbitrary size). Thus Theorem~\ref{thm:IW} can be interpreted as saying ``if nonuniformity cannot speed up all (exponential-time) algorithms too much, then randomization never provides too much of a speed up.'' Or, in more of a ``win-win'' formulation, ``either we can speed up all (exponential-time) algorithms with nonuniformity, or we can efficiently derandomize all probabilistic algorithms.''

The two items in Theorem~\ref{thm:IW} are special cases of a more general quantitative result that relates circuit complexity to derandomization. At the ``low end,'' Item~\ref{itm:lowend} says that if there are problems solvable in exponential time that have superpolynomial circuit complexity, then we get a subexponential-time derandomization of $\BPP$.  This is the same as Yao's aforementioned result~\cite{Yao82} but with a much weaker hypothesis than the existence of cryptographically secure pseudorandom generators.  At the ``high end,'' Item~\ref{itm:highend} says that if instead we have problems with exponential circuit complexity, then in fact we get polynomial-time derandomization. (We also need to make the relatively minor switch from $\EXP$ to $\ETIME$.  If we used $\EXP$ instead, the conclusion would be  
$\BPP\subseteq \quasiP$.)

These circuit complexity hypotheses are very plausible. The $\NP$-complete problems are promising candidates; they can be solved in exponential time and are conjectured to require superpolynomial and even exponential circuit complexity, though we are very far from proving it. (Such a result would resolve the famous $\P$ vs. $\NP$ question.)
See Chapter~\ref{sec:lowerbounds} for a survey of Wigderson's fundamental contributions to circuit lower bounds.

Sections~\ref{sec:prgs}--\ref{sec:worstcase-avgcase} give an overview of the proof of Theorem~\ref{thm:IW}.

\paragraph{\underline{Circuit Lower Bounds from Derandomization}}

Theorem~\ref{thm:IW} of Wigderson et al. establishes a unidirectional implication between two major projects in theoretical computer science: {\em if} we can prove circuit lower bounds, {\em then} we can provably derandomize $\BPP$. 
 Since proving circuit lower bounds is so difficult, it is natural to wonder whether derandomization could be easier.  With Impagliazzo and Kabanets~\cite{ImpagliazzoKaWi02}, Wigderson proved that if we ``add nondeterminism'' to the complexity classes, then in fact circuit lower bounds are {\em equivalent} to derandomization.

\begin{theorem}[\cite{ImpagliazzoKaWi02}] \label{thm:IKW}
$\MA$ (a randomized analogue of $\NP$) has a nontrivial derandomization (namely, $\MA\neq \NEXP$, where $\NEXP$ is an exponential-time analogue of $\NP$) {\em if and only if} $\NEXP$ does not have polynomial-sized circuits.
\end{theorem}
It follows from Theorem~\ref{thm:IKW} that derandomization of ordinary randomized polynomial-time algorithms (without nondeterminism) also implies circuit lower bounds.  Specifically, it turns out that if we can derandomize the generalization of $\BPP$ to ``promise problems'' (i.e. partial boolean functions, where we don't define or care about the output on certain inputs), then we can also derandomize $\MA$ and hence deduce from Theorem~\ref{thm:IKW} that $\NEXP$ does not have polynomial-sized circuits. 
(This implication of Theorem~\ref{thm:IKW} also follows from the earlier work of \cite{BuhrmanFoTh98}.)
Building on Theorem~\ref{thm:IKW}, Kabanets and Impagliazzo~\cite{KabanetsIm04} proved that even if the specific problem of \IdentityTest\ (which is in $\BPP$) has a nontrivial derandomization, then $\NEXP$ does not have polynomial-sized boolean circuits {\em or} the \Permanent\ does not have polynomial-sized arithmetic circuits, either of which would be breakthroughs in complexity theory.  (See Section~\ref{sec:partialderivatives} for Wigderson's work on arithmetic circuit lower bounds and Section~\ref{sec:ncit} for Wigderson's work on variants of \IdentityTest.)

A more positive interpretation of Theorem~\ref{thm:IKW} is that we might be able to prove new circuit lower bounds by coming up with new methods for derandomizing algorithms.  This possibility was realized in Williams' program of proving circuit lower bounds by designing faster SAT algorithms~\cite{Williams10} and his breakthrough result that $\NEXP$ does not have polynomial-sized $\ACC$ circuits~\cite{Williams11}, both of which built on the results and techniques of \cite{ImpagliazzoKaWi02}.

\paragraph{\underline{Optimizing the Hardness vs. Randomness Tradeoff}}

As mentioned above, there is a full spectrum of ``hardness vs. randomness'' implications between the ``low end'' and ``high end'' derandomizations stated in Theorem~\ref{thm:IW}. 
With Impagliazzo and Shaltiel~\cite{ImpagliazzoShWi06}, Wigderson pointed out that in the intermediate regime, 
the proof of Theorem~\ref{thm:IW} yielded results that were suboptimal in a sense that can be made formal, and initiated a line of work that culminated in optimal hardness vs. randomness tradeoffs achieved by Shaltiel and Umans~\cite{ShaltielUm05,Umans03}.
More recently, researchers have turned to more finely quantifying how much slowdown is needed to derandomize algorithms. Under suitably strong complexity assumptions, recent works~\cite{DoronMoOhZu22,ChenTe21-fast} give evidence that every randomized algorithm running in time $T(n)$ can be converted to a deterministic algorithm running in time $n^{1+\epsilon}\cdot T(n)$ for an arbitrarily small constant $\epsilon>0$.  Thus, it seems that randomization saves at most an almost-linear factor in runtime!

\paragraph{\underline{Derandomization from Uniform Assumptions}}

Theorem~\ref{thm:IKW} and the results discussed after it show that some  derandomizations of $\BPP$ require novel circuit lower bounds, at least for $\NEXP$.
Nevertheless, in a remarkable paper with Impagliazzo~\cite{ImpagliazzoWi01}, Wigderson showed that a nontrivial derandomization of $\BPP$ {\em is} possible under the uniform assumption $\EXP\neq \BPP$.  Specifically, for every language $L$ in $\BPP$, they obtain a deterministic subexponential-time algorithm that correctly decides $L$ on all but a $1/\poly(n)$ fraction of inputs of length $n$, for infinitely many values of $n$.  Moreover, this holds not just for the uniform distribution on instances of length $n$, but simultaneously for every efficiently samplable distribution on instances.
 
An intriguing feature of the Impagliazzo--Wigderson uniform derandomization is that it is (necessarily~\cite{TrevisanVa07}) a ``non-black-box'' construction.  That is, the construction and proof actually make use of the {\em code} of the programs that compute the assumed hard function $f\in \ETIME$ and that decide the language $L\in \BPP$ that is being derandomized.  In contrast, Theorem~\ref{thm:IW}, like most results in complexity theory, treats these algorithms as ``black boxes,'' only using the fact that they can be solved by efficient programs to deduce that other programs using them as subroutines are efficient.

The Impagliazzo--Wigderson uniform derandomization of $\BPP$ is
a ``low-end'' result; assuming only a superpolynomial lower bound for $\EXP$, we get (only) a subexponential-time derandomization of $\BPP$.  It remains an open problem to have a high-end (or nearly high-end) analogue of their result, for example to get a polynomial-time (or quasipolynomial-time) average-case derandomization of $\BPP$ under the assumption that $\ETIME\nsubseteq 
\BPTIME(2^{o(n)})$ (or $\EXP\nsubseteq \BPSUBEXP$).  In \cite{TrevisanVa07}, a uniform analogue of the high-end worst-case-to-average-case hardness for $\ETIME$ (Theorem~\ref{thm:amplifyhardness}) was given.

Subsequent works have given uniform average-case derandomizations of randomized algorithms with one-sided error ($\RP$)~\cite{Kabanets01} and constant-round interactive proofs (a.k.a. Arthur--Merlin games, $\AM$)~\cite{Lu01,ImpagliazzoKaWi02,GutfreundShTa03,ShaltielUm09}.Recent work has identified ``almost-everywhere'' uniform hardness assumptions (e.g. computational problems where every uniform probabilistic polynomial-time algorithm fails to solve the problem on all but finitely many inputs) that are {\em equivalent} to worst-case derandomization of $\BPP$ (generalized to ``promise problems'')~\cite{Goldreich11,ChenTe21-uniform,LiuPa22}.

\subsubsection{Pseudorandom Generators}
\label{sec:prgs}

The approach to derandomizing algorithms suggested by Yao~\cite{Yao82} and pursued by Wigderson is by constructing {\em pseudorandom generators}.  These are defined in terms of computational indistinguishability, which was introduced in Section~\ref{sec:crypto} and will be convenient to reformulate here in a non-asymptotic form:

\begin{definition}[computational indistinguishability~\cite{GoldwasserMi84}]
\label{def:compind}
Random variables $X$ and $Y$ taking values in
$\zo^m$ are {\em $(s,\eps)$ indistinguishable} if for every boolean circuit $T : \zo^m\rightarrow \zo$ of size at most $s$, we have
$$
\left| \Pr[T(X)=1] - \Pr[T(Y)=1] \right| \le \eps
$$
The left-hand side above is called also the {\em advantage} of $T$ in distinguishing $X$ and $Y$.
\end{definition}
If we set $s=\infty$ (or even $s=2^m$), then we allow all $2^{2^m}$ boolean functions $T$ as statistical tests, and $(s,\eps)$-indistinguishability is equivalent to requiring that $X$ and $Y$ have total variation distance at most $\eps$.  (See Definition~\ref{def:tvd}.)  However, by restricting to computationally efficient tests, e.g., with $s=\poly(m)$, then we obtain a significantly relaxed definition, where even random variables $X$ and $Y$ with disjoint supports can be indistinguishable.  At the same time, for all efficient purposes (i.e. tasks that can be done by a boolean circuit of size $s$), $X$ and $Y$ are interchangeable.

A pseudorandom generator is a procedure that stretches a short seed if truly random bits into a long string that is computationally indistinguishable from uniform.

\begin{definition}[pseudorandom generator~\cite{BlumMi84,Yao82}] \label{def:PRG}
A deterministic function $G:\{0,1\}^d\rightarrow \{0,1\}^m$ is an
\emph{$(s,\eps)$ pseudorandom generator (PRG)} if
\begin{enumerate}
\item $d < m$, and
\item $G(U_d)$ and $U_m$ are $(s,\eps)$ indistinguishable, where $U_k$ denotes a random variable uniformly distributed over $\zo^k$.
\end{enumerate}
\end{definition}
If a test $T : \zo^m\rightarrow \zo$ has advantage at most $\eps$ in distinguishing $G(U_d)$ from $U_m$, we say that $G$ {\em $\eps$-fools} $T$.

People attempted to construct pseudorandom generators long before this definition was formulated.
Their generators were tested against a battery of statistical tests (e.g. the number of
1's and 0's are approximately the same, the longest run is of length $O(\log m)$, etc.), but
these fixed set of tests provided no guarantee that the generators would perform well in
an arbitrary application.
Indeed, most classical constructions (e.g. linear congruential generators, as implemented
in the standard C library) are known to fail in some applications.

Intuitively, the above definition guarantees that the pseudorandom bits produced by the generator
are as good as truly random bits for {\em all} efficient purposes (where efficient means computable by a circuit of size at most $s$).
In particular, we can use such a generator to reduce the number of random bits used by any algorithm from $m$ to $d(m)$ provided that the algorithm runs in time at most $t=s/\polylog(s)$, because the behavior of any such algorithm on any input $x$ can be simulated by a boolean circuit of size $s$.  For the resulting algorithm to be efficient, we will also need the generator
to be efficiently computable.

\begin{definition}
We say a sequence of generators $\{G_m : \zo^{d(m)}\rightarrow \zo^m\}$ is {\em computable in
time $t(m)$} if there is a {\em uniform and deterministic} algorithm $M$ such that for every $m\in \N$ and
$x\in \zo^{d(m)}$, we have $M(m,x)=G_m(x)$ and $M(m,x)$ runs in time at most $t(m)$. %
\end{definition}

Note that even when we define the pseudorandomness property of the generator with respect to
nonuniform boolean circuit, the efficiency requirement refers to uniform algorithms. For
readability,
we will usually refer to a single generator $G : \zo^{d(m)}\rightarrow \zo^m$, with it being
implicit that we are really discussing a family $\{G_m\}$.

\begin{theorem} \label{thm:PRGderandomize}
Suppose that for all $m$ there exists an $(m,1/8)$ pseudorandom generator
$G : \zo^{d(m)}\rightarrow \zo^m$ computable in time $t(m)$.  Then
$\BPP\subseteq \bigcup_c \DTIME(2^{d(n^c)}\cdot (n^c+t(n^c)))$.
\end{theorem}

\begin{proof}
Let $L$ be any language in $\BPP$.  Then there is a constant $c$ such that $L$ is decided by a bounded-error randomized algorithm in time $t(n)=O(n^{c-1})$ on inputs of length $n$.   

The idea is to replace the random bits used by $A$ with pseudorandom bits generated by $G$, use the pseudorandomness property to show that the algorithm will still be correct
with high probability, and finally enumerate over all possible seeds to obtain a deterministic algorithm.

\begin{claim}  For all sufficiently large $n$ and every $x\in \zo^n$, $A(x;G(U_{d(n^c)}))$ errs
with probability smaller than $1/2$.
\end{claim}

\begin{claimproof} Suppose that there exists some $x\in \zo^n$ on which $A(x;G(U_{d(n^c)}))$ errs
with probability at least $1/2$.  Then, $T(\cdot) = A(x, \cdot)$ 
distinguishes $G(U_{d(n^c)})$ from $U_{n^c}$ with advantage at least
$1/2-1/3 > 1/8$.  Since algorithms running in time $t(n)$ can be simulated by boolean circuits of size at most $\tO(t(n))$, $T(\cdot)$ can be computed by a boolean circuit of size at most $n^c$, for sufficiently large $n$. This contradicts the pseudorandomness property of $G$. 
\end{claimproof}

Now, enumerate over all seeds of length $d(n^c)$ and take a majority vote.  There are
$2^{d(n^c)}$ of them, and for each we have to run both $G$ and
$A$.
\end{proof}

In the definition of a cryptographic pseudorandom generator used by Yao~\cite{Yao82}, the requirement was that $G$ is computable in time polynomial in its input length, i.e. $t(m)\leq \poly(d(m))$.  This implies that $d(m)\geq t(m)^\delta \geq m^\delta$ for a constant $\delta > 0$, so the running time of the derandomization in Theorem~\ref{thm:PRGderandomize} is at least $2^{d(n^c)} \geq 2^{n^{\delta c}}$, and we can at best conclude $\BPP\subseteq \SUBEXP$.

Thus a key to Theorem~\ref{thm:IW} was the realization by Nisan and Wigderson~\cite{NisanWi94} that, for derandomization, we can relax the efficiency requirements of a cryptographic generator in two ways.  First, we can afford for the generator to be computable in time exponential in its seed length, since anyway we enumerate over all seeds when derandomizing.  Second (and relatedly), we can afford for the generator to run in more time than the algorithms it fools.  Indeed, in Theorem~\ref{thm:PRGderandomize}, we only need to fool circuits of size $m$, but we are happy for a generator computable in time $\poly(m)$.  In contrast, a cryptographic generator actually requires fooling circuits of size $m^{\omega(1)}$, ones that are superpolynomially larger than the output length and the running time of the generator.
Thus, they proposed the following efficiency requirement:
\begin{definition}[\cite{NisanWi94}] 
A generator $G : \zo^{d(m)}\rightarrow \zo^m$ is {\em quick} (a.k.a. {\em mildly explicit})  if it is computable in time $\poly(m,2^{d(m)})$.
\end{definition}

They demonstrated the benefits of these relaxed requirements with the beautiful pseudorandom generator construction described in the next section (which is a key component of the proof of Theorem~\ref{thm:IW}).

\subsubsection{The Nisan--Wigderson Generator} \label{sec:NW}

The Nisan--Wigderson generator constructs a quick pseudorandom generator from any function in $\ETIME$ that is sufficiently {\em hard on average}:

\begin{definition} \label{def:average-case-hard}
For $s\in \N$ and $\alpha>0$, a function $f : \zo^\ell \rightarrow \zo$ is {\em $(s,\alpha)$ average-case hard} if for every boolean circuit $A$ of size at most $s$, we have
$$\Pr[A(U_\ell) \neq f(U_\ell)] > \alpha.$$
\end{definition}
Note that, in contrast to the definition of $\BPP$, here the probabilities are taken over the {\em input} to the algorithm $A$, rather than its random coin tosses.  When $\alpha=0$, Definition~\ref{def:average-case-hard} simply says that $f$ has circuit complexity greater than $s$, but when $\alpha$ is nonzero it is a significantly stronger hardness requirement on $f$.  Note that $\alpha=1/2$ is impossible, since a constant function (of size $s=1$) can always compute $f$ correctly on at least half of the inputs.

We now state the Nisan--Wigderson theorem, restricted to the ``high-end'' regime, where hardness is against circuits of exponential size.

\begin{theorem}[\cite{NisanWi94}] \label{thm:NWspecific}
Suppose that there is a constant $\delta>0$ and a function $f \in \ETIME$ such that for every input length $\ell\in \N$, $f_\ell$ is $(2^{\delta \ell},1/2-1/2^{\delta\ell})$ average-case hard.  Then for every $m\in \N$, there is a quick
$(m,1/m)$ pseudorandom generator $G : \zo^{d(m)}\rightarrow \zo^m$
with seed length $d(m) = O(\log m)$.  In particular, $\BPP=\P$.
\end{theorem}
Similar to Theorem~\ref{thm:IW}, this is a specific instance of a quantitative tradeoff between hardness and derandomization. In particular, if we replace both occurrences of the exponential bound $2^{\delta\ell}$ with a superpolynomial bound $\ell^{\omega(1)}$, we obtain the ``low-end'' conclusion that $\BPP\subseteq \SUBEXP$.
However, the hypothesis in Theorem~\ref{thm:NWspecific} is significantly stronger in that it assumes average-case hardness rather than worst-case hardness, and very strong average-case hardness at that: no small circuit can compute $f$ with probability much better than random guessing.  In the next section, we will discuss how Wigderson and collaborators relaxed the average-case hardness assumption to a worst-case one in order to obtain Theorem~\ref{thm:IW}.

The starting point for Theorem~\ref{thm:NWspecific} is the realization, implicit in Yao~\cite{Yao82}, that if $f$ is $(s,1/2-\eps)$ average-case hard, then $G(x)=(x,f(x))$ is an $(s-O(1),\eps)$ pseudorandom generator. That is, by applying $f$ once on a uniformly random input, we obtain one pseudorandom bit (beyond the $d=\ell$ truly random bits in the seed).  So, to obtain may pseudorandom bits, we can try applying $f$ many times.  For this to provide a generator with large stretch (i.e. with output length superlinear in the input length), we cannot evaluate $f$ on independent random inputs, but rather need to generate many correlated inputs, but ensure that the correlations don't destroy the pseudorandomness.

The idea, building on Nisan~\cite{Nisan91-AC}, is to use inputs to $f$ that share very few bits.
Specifically, the sets of seed bits used for each input to $f$ will be given by a {\em design}:
\begin{definition}
$S_1,\cdots,S_m \subseteq [d]$ is an {\em $(\ell,a)$-design} if
\begin{enumerate}
\item $\forall i, |S_i|=\ell$
\item $\forall i \not= j, |S_i \cap S_j| \leq a$
\end{enumerate}
\end{definition}
It turns out that there exist designs with
lots of sets having small intersections over a small universe:

\begin{lemma} \label{lem:design}
For every every $\ell,m \in \N$, there exists an
$(\ell,a)$-design $S_1,\cdots, S_m \subseteq [d]$ with
$d=O\left(\frac{\ell^2}{a} \right)$ and $a = \log_2 m$. Such a design
can be constructed deterministically in time $\poly(m,d)$.
\end{lemma}
The important points are that intersection sizes are only logarithmic in the number of sets, and
the universe size $d$ is linear in $\ell$ in case we take $m=2^{\Omega(\ell)}$.

\begin{construction}[Nisan--Wigderson Generator] \label{const:NW}
Given a function $f: \zo^\ell \to
\zo$ and an $(\ell,a)$-design $S_1,\cdots, S_m \subseteq [d]$, define the {\em Nisan--Wigderson generator} $G: \zo^d \to \zo^m$ as
$$G(x) = f(x|_{S_1}) f(x|_{S_2}) \cdots f(x|_{S_m})$$
where if $x$ is a string in $\zo^d$ and $S \subseteq [d]$, then $x|_S$ is the
string of length $|S|$ obtained from $x$ by selecting the bits indexed by
$S$.
\end{construction}

This elegant construction is analyzed as follows.

\begin{theorem} \label{thm:NWgeneral}
Let $G : \zo^d\rightarrow \zo^m$ be the Nisan--Wigderson generator
based on a function $f : \zo^\ell\rightarrow \zo$ and some $(\ell,a)$ design.
If $f$ is $(s,1/2 - \eps/m)$ average-case hard, then
$G$ is a $(s',\eps)$ pseudorandom generator, for $s'=s-m\cdot 2^a$.
\end{theorem}

Theorem~\ref{thm:NWspecific} follows from Theorem~\ref{thm:NWgeneral}
by setting $\eps=1/m$, $a=\log_2 m$, and $s=2^{\delta \ell}$, and observing that
for $\ell=(1/\delta)\cdot\log_2(2m^2) = O(\log m)$, we have $$s' = s-m\cdot 2^a = 2m^2 - m^2\geq m,$$
and $\eps/m \leq 1/2^{\delta \ell}$,
so we have an $(m,1/m)$ pseudorandom generator.  The seed length
is $d=O(\ell^2/\log m) = O(\log m).$

\begin{proof}
Suppose for contradiction that $G$ is not an $(s',\eps)$ pseudorandom generator. By the equivalence of pseudorandomness and next-bit unpredictability~\cite{Yao82},
there is a size $s'$
circuit 
$P$ such that
\begin{equation}
\Pr[P(f(X|_{S_1}) f(X|_{S_2}) \cdots f(X|_{S_{i-1}}))=f(X|_{S_i})] >
\frac{1}{2} + \frac{\eps}{m},
\end{equation}
for some $i\in [m]$ and a uniformly random $X\gets \zo^d$.
From $P$, we will construct a small circuit $A$
that computes $f$ on a uniformly random input with probability greater than $1/2+\eps/m$.

Let $Y=X|_{S_i}$.  By averaging, we can fix
all bits of $X|_{\overline{S_i}}=z$ (where $\overline{S_i}$ is the complement of $S$) such that
the prediction probability remains greater than $1/2+\eps/m$
over $Y$.
Define $f_j(y)=f(x|_{S_j})$ for $j
\in \{1,\cdots,i-1\}$. (That is, $f_j(y)$ forms $x$ by placing $y$
in the positions in $S_i$ and $z$ in the others, and then applies
$f$ to $x|_{S_j}$).  Then

$$\Pr_Y[P(f_1(Y) \cdots f_{i-1}(Y))=f(Y)] > \frac{1}{2} +
\frac{\eps}{m}.$$

Note that $f_j(y)$ depends on only $|S_i \cap S_j| \leq a$ bits of
$y$. Thus, we can compute each $f_j$ with a look-up table hardwired into our circuit. Indeed, every function on $a$ bits can be computed
by a boolean circuit of size at most $2^a$.  (In fact,  size
at most $O(2^a/a)$ suffices.)

Then, by considering $A(y)=P(f_1(y)\cdots f_{i-1}(y))$, we
deduce that
$f$ can be computed with error probability smaller than $1/2-\eps/m$ by a boolean circuit of size at most $s'+(i-1)\cdot 2^a < s'+m\cdot 2^a = s.$
This contradicts the
hardness of $f$.  Thus, we conclude $G$ is an $(m,\eps)$ pseudorandom generator.
\end{proof}

\subsubsection{Pseudorandom Generators from Worst-Case Lower Bounds}
\label{sec:worstcase-avgcase}

As we saw in the previous section, the Nisan--Wigderson construction gives us pseudorandom generators from boolean functions that are very hard on average, where every boolean circuit of size $2^{\delta \ell}$ must err with probability
greater than $1/2-1/2^{\delta \ell}$ on a random input.
In works with Babai, Fortnow, and Nisan~\cite{BabaiFoNiWi93} and Impagliazzo~\cite{ImpagliazzoWi97}, Wigderson showed how to relax the assumption to worst-case hardness, yielding Theorem~\ref{thm:IW}.
This was done by showing how to convert worst-case hard functions into average-case hard functions, which again we state only in the high-end regime of parameters:

\begin{theorem}[worst-case to average-case hardness for $\ETIME$~\cite{ImpagliazzoWi97}] \label{thm:amplifyhardness}
Suppose that for a constant $\delta>0$, there is a function in $\ETIME$ that has circuit complexity at least $2^{\delta\ell}$ on inputs of length $\ell$. Then there is
a constant $\delta'>0$ and a function 
in $\ETIME$
that is $(2^{\delta' \ell},1/2-1/2^{\delta' \ell})$ average-case hard.
\end{theorem}
Combining Theorem~\ref{thm:amplifyhardness} and Theorem~\ref{thm:NWspecific} yields the high-end part of Theorem~\ref{thm:IW}.

Beyond the application to pseudorandomness and derandomization, the relationship between worst-case complexity and average-case complexity is a central question in complexity theory. (See the survey \cite{BogdanovTr06}.) In particular, whether a similar result is true for $\NP$ (rather than $\ETIME$) remains a major open problem.

\subsection{Expanders, Extractors, and Ramsey Graphs}

Another area in which randomness has proved very useful is in 
the Probabilistic Method~\cite{AlonSp16}, whereby mathematical objects with interesting properties are proven to exist by showing that a randomly chosen object has the desired property with high (or at least nonzero) probability.  A famous example is \Erdos' existence proof for Ramsey graphs --- graphs with no large clique or independent set~\cite{Erdos47}.

In such cases, the problem of derandomization becomes one of finding {\em explicit constructions} of objects with the desired properties.  The search for explicit constructions is of pure mathematical interest, as a way of developing and testing our understanding of the mathematical properties at hand.  They are also important for many computer science applications, where we need efficient algorithms to describe and work with the objects.

In this section, we survey Wigderson's contributions to explicit constructions, in particular to the constructions of expander graphs, randomness extractors, and Ramsey graphs, as well as identifying and exploiting the connections between these.

\subsubsection{Expander Graphs} \label{sec:expanders}

{\em Expander graphs} are graphs that are ``sparse'' yet very ``well-connected.''  They are ubiquitous in theoretical computer science, with applications including communication and routing networks~\cite{Pinsker73,PelegUp89}, derandomizing algorithms~\cite{AjtaiKoSz83,Reingold08}, error-correcting codes~\cite{Guruswami06-EATCS}, lower bounds on circuit complexity~\cite{Valiant76} and proof complexity~\cite{Ben-SassonW01}, integrality gaps for optimization problems~\cite{LinialLoRa95,AumannRa98}, data structures~\cite{BuhrmanMiRaVe02}, fault-tolerant storage~\cite{UpfalWi87} and more.  A rich mathematical theory has developed around constructing expanders and understanding their properties; we refer to Wigderson's survey with Hoory and Linial~\cite{HooryLiWi06} as well as \cite{Vadhan12-fnttcs,Trevisan17,Lau22} for many aspects that we will not be able to cover here.

We will typically interpret the properties of expander graphs in an asymptotic
sense. That is, there will be an infinite family of graphs $G_i$,
with a growing number of vertices $N_i$.  By ``sparse,'' we mean
that the (maximum or average) degree $D_i$ of $G_i$ should be very slowly growing as a
function of $N_i$.  The ``well-connectedness'' property has a variety
of different interpretations, which we will discuss below.
Typically, we will drop the subscripts of $i$ and the fact that we
are talking about an infinite family of graphs will be implicit in
our theorems.  We will state many of
our definitions for {\em directed multigraphs} (which we'll call {\em digraphs} for short), though in the end we will
mostly study undirected multigraphs. %

The most intuitive definition of expansion is the following.
\begin{definition} \label{def:vertex-expansion}
A digraph $G$ is a $(K,A)$ \emph{vertex expander} if  for all sets
$S$ of at most $K$ vertices, the {\em (out-)neighborhood} $N(S) \eqdef
\set{u|\exists v\in S \mbox{ s.t. } (u,v)\in E}$ is of size at least
$A\cdot |S|$.
\end{definition}
Ideally, we would like graphs with degree $D=O(1)$, and $(K,A)$ vertex expansion with $K = \Omega(N)$ where $N$ is the number of vertices, and $A$ as close
to $D$ as possible.

It is often useful to work instead with a linear-algebraic measure of expansion.  For simplicity, we restrict attention to regular graphs in presenting the definition.

\begin{definition}
    Let $G$ be an $N$-vertex $D$-regular digraph with random-walk matrix $M$ (so $M_{ij}$ equals the number of edges from $i$ to $j$ divided by $D$).  Let $\sigma_2(G)\in [0,1]$ denote the second-largest singular value of $M$.  The {\em spectral expansion} of $G$ is $\gamma(G) = 1-\sigma_2(G)$.\footnote{In some other sources, the term spectral expansion refers to $\sigma_2(G)$ rather than $\gamma(G)$.  Here we use $\gamma(G)$, because it has
 the more natural feature that larger values of $\gamma$ correspond to the graph being ``more expanding''.}
\end{definition}
Ideally, we would like an infinite family of graphs with degree $D=O(1)$ and $\gamma(G)=\Omega(1)$.
Alon~\cite{Alon86a} proved that this linear-algebraic measure of expansion is
equivalent to the combinatorial measure of vertex expansion for
common parameters of interest. 

\begin{theorem}[\cite{Alon86a}] \label{cor:VertSpec}
Let $\GFam$ be an infinite family of $D$-regular multigraphs, for a constant $D\in \N$.  Then the following two conditions are equivalent:
\begin{itemize}
\item There is a constant $\delta>0$ such that every $G\in \GFam$ is an $(N/2,1+\delta)$ vertex expander.
\item There is a constant $\gamma>0$ such that every $G\in \GFam$ has spectral expansion at least $\gamma$.
\end{itemize}
\end{theorem}
When people informally use the term
``expander,'' they often mean a family of regular graphs of {\em constant degree} $D$ satisfying one of the two equivalent conditions above.  However, we note that the quantitative relationship between vertex expansion and spectral expansion is lossy, so optimizing one of these measures of expansion need not yield optimality with respect to the other.

We can get more intuition for spectral expansion by considering some equivalent formulations of it.
Since $G$ is regular, the uniform distribution, written as a row vector $u=(1/N,\ldots,1/N)$, is an eigenvector of the random-walk matrix $M$ of eigenvalue 1, i.e. $uM = u$.  By the Perron-Frobenius Theorem, the largest singular value of $M$ equals 1, and thus we have the following variational characterization of spectral expansion. 

\begin{lemma}
 $1-\gamma(G) = \sigma_2(G) = \max_{x\perp u} \frac{\|xM\|}{\|x\|} = \max_\pi \frac{\|\pi M-u\|}{\|\pi-u\|},$
where the first maximum is over all nonzero row vectors $x\in \R^N$ that are orthogonal to $u$, and the second maximum is over all probability distributions $\pi\in [0,1]^N$ (also written as row vectors).
\end{lemma}
That is, if we start at any probability distribution $\pi$ on the vertices of $G$ and take one step of the random walk to end up at probability distribution $\pi M$, the $\ell_2$ distance to uniform will shrink by at least a factor of $1-\gamma(G)$.  So if $\gamma(G)$ is bounded away from 0, then random walks on $G$ will converge quickly to the uniform distribution.

Another useful characterization of $\gamma(G)$ is as follows.
\begin{lemma} \label{lem:MatrixDecomposition}
$\gamma(G)\geq \gamma$ iff we can write $M=\gamma J + (1-\gamma) E$, where $J$ is the matrix with every entry $1/N$ and $\|E\|\leq 1$, where $\|E\|$ is the spectral norm of $E$.
\end{lemma}
Notice that $J$ is the random-walk matrix for the complete graph on $N$ vertices with self-loops, which is intuitively the most expanding possible graph (albeit not sparse).  Thus, Lemma~\ref{lem:MatrixDecomposition} says that an expander can be viewed as a sparse approximation of the complete graph.

It can be shown that a random $D$-regular undirected graph on $N$ vertices is an excellent expander with high probability, for $D=O(1)$ and $N\rightarrow \infty$.  For example, it achieves spectral expansion 
$\gamma(G) = 1- 2\sqrt{D-1}/D + o(1)$~\cite{Friedman08}, which is optimal up to the $o(1)$~\cite{Nilli91}, and achieves $(\alpha A, D-1-\eps)$ vertex expansion for any $\eps>0$ and $\alpha=\alpha(D,\eps)$~\cite{Pinsker73,Bassalygo81}.
For some applications of expanders, however, we cannot afford to choose the graph at random, because it may be too costly in memory, communication, or randomness. Indeed, some applications even require exponentially large expander graphs, in which case a random graph would be completely infeasible to manipulate.  Thus, we seek {\em explicit constructions} of expanders.

\begin{definition} \label{def:explicit-expanders}
Let $\GFam = \{G_i\}$ be an infinite family of digraphs where $G_i$ has $N_i$ vertices and is $D_i$-regular.  We say that $\GFam$ is {\em (fully) explicit} if given $N_i$, $u\in [N_i]$, and $j\in [D_i]$, the $j$'th neighbor of $u$ in $G_i$ can be computed deterministically in time $\poly(\log N_i)$.
\end{definition}
That is, we require a very efficient local description of the graph, where computing neighbors can be done in time polynomial in the bitlength of vertices, rather than in time polynomial in the number of vertices.

Thus, starting with Margulis~\cite{Margulis73}, there is a long and beautiful line of work on explicit constructions of constant-degree expanders, with
one highlight being the optimal spectral expanders of Lubotzky, Phillips, and Sarnak~\cite{LubotzkyPhSa88}  and Margulis~\cite{Margulis88}, known as {\em Ramanujan graphs}.  Many of these constructions were based on deep results from algebra and number theory, and it was of interest to have more combinatorial approaches to constructing expanders.

With Reingold and Vadhan~\cite{ReingoldVaWi01}, Wigderson gave a combinatorial construction of expanders based on a new graph operation, called the {\em zig-zag product}.  Although these expanders did not match the spectral expansion of Ramanujan graphs, the additional flexibility offered by the construction found numerous applications, which we will survey below.

Specifically, their approach to constructing expanders is to start with a constant-sized expander of appropriate parameters and repeatedly apply graph operations to build larger and larger graphs while preserving the degree and spectral expansion. 

Two standard operations on an $N$-vertex $D$-regular graph $G$ with random-walk matrix $M$ are the following:
\begin{description}
\item[\textit{Squaring:}] $G^2$ is the graph on $N$ vertices whose random-walk matrix is $M^2$. That is, edges in $G^2$ are walks of length 2 in $G$. If $G$ has spectral expansion at least $\gamma=1-\sigma$, then $G^2$ has spectral expansion at least $1-\sigma^2 = 2\gamma-\gamma^2$

\item[\textit{Tensoring:}] $G\otimes G$ is the graph on $N^2$ vertices whose random-walk matrix is $M\otimes M$ (the Kronecker product).  That is, random walks in $G\otimes G$ correspond to two independent random walks in $G$.  If $G$ has spectral expansion at least $\gamma$, then $G\otimes G$ also has spectral expansion at least $\gamma$.
\end{description}
Squaring has the benefit of improving expansion and tensoring has the benefit of creating larger graphs, but both have the downside of increasing the degree $D$ to $D^2$.  Thus, we need an operation that decreases the degree, without hurting the expansion too much.  This is what the zig-zag product achieves.

\paragraph{\underline{The Zig-Zag Product}}

Let $G$ be a $D_1$-regular digraph on $N_1$ vertices and $H$ be a regular digraph on $D_1$ vertices.
The {\em zig-zag product} of $G$ and
$H$, denoted $G \zigzag H$, is defined as follows. The nodes of $G \zigzag H$ are the pairs
$(u,i)$ where $u \in V(G)$ and $i \in V(H)$. 
We think of this each vertex $u$ of $G$ with a copy of $V(H)$, which we refer to as a {\em cloud}, and associate each vertex of $H$ with one of the edges incident to $u$.  The edges in $G\zigzag H$ then correspond to taking an $H$-step within a cloud, using a $G$-step to move between clouds, and an $H$-step in the resulting cloud.  See Figure~\ref{fig:zigzag} for an illustration.

\begin{figure}[h]
    \centering
    \begin{subfigure}[]{0.45\textwidth}
        \includegraphics[width=\textwidth]{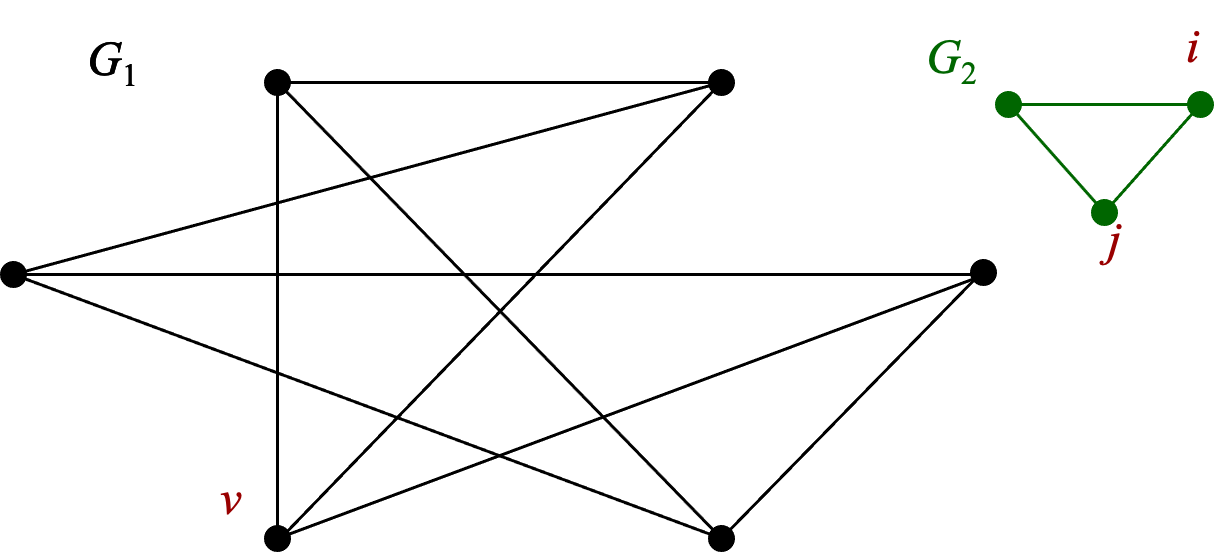}
        \caption{Initial $D_1$-regular graph $G_1$ and graph $G_2$ on $D_1$ vertices.}
    \end{subfigure}
    \hfill
    \begin{subfigure}[]{0.45\textwidth}
        \includegraphics[width=\textwidth]{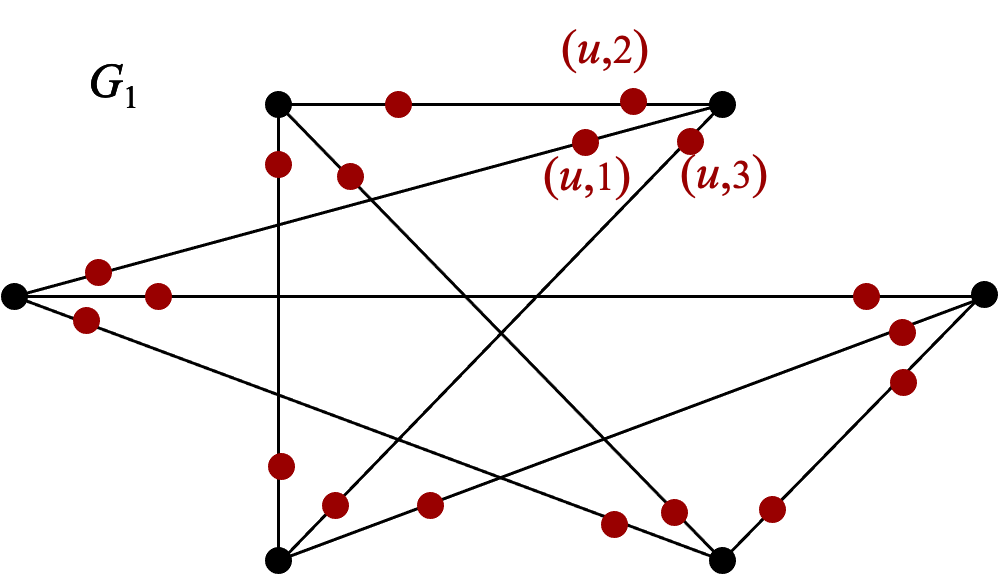}
        \caption{First we construct a cloud with the same shape as $G_2$ for each vertex $u \in G_1$.}
    \end{subfigure}
    \hfill
    \begin{subfigure}[]{0.45\textwidth}
        \includegraphics[width=\textwidth]{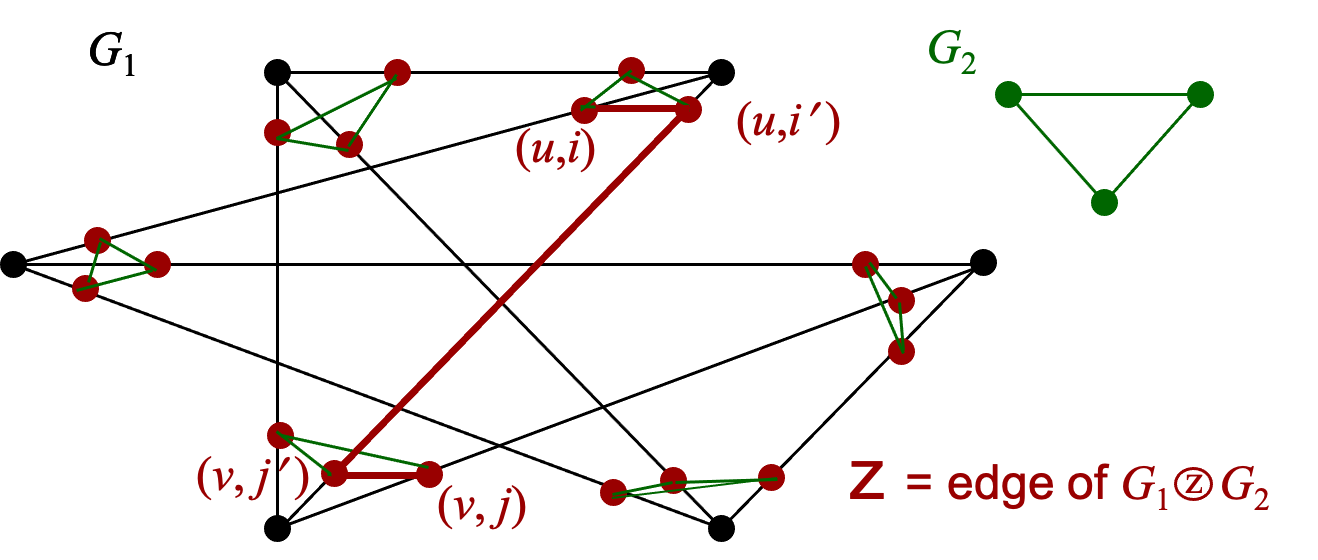}
        \caption{Three steps connecting $(u,i)$ and $(v,j)$ following Definition~\ref{def:zigzag}}
    \end{subfigure}
    \hfill
    \begin{subfigure}[]{0.45\textwidth}
        \includegraphics[width=\textwidth]{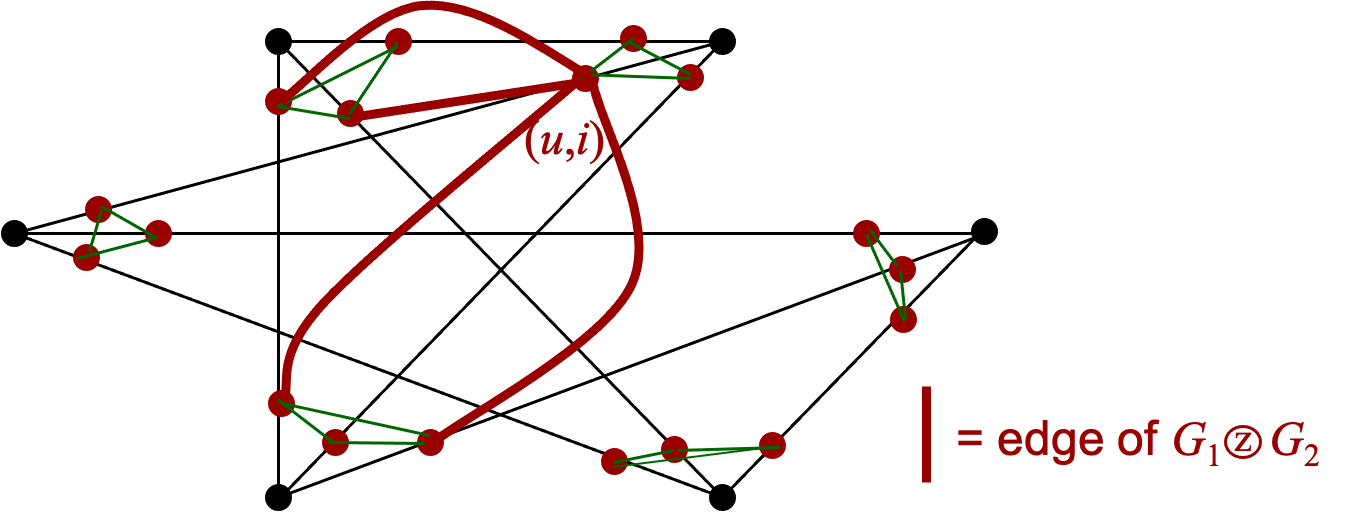}
        \caption{We construct an edge between $(u,i)$ and $(v,j)$ and similarly for other edges.}
    \end{subfigure}
    \caption{Illustration of the edge construction in the Zig-Zag product $G_1 \zigzag G_2$.}
    \label{fig:zigzag}
\end{figure}

\begin{definition}[Zig-zag Product] \label{def:zigzag} 
Let $G$ be an $D_1$-regular digraph on $N_1$ vertices, and $H$ a $D_2$-regular digraph on $D_1$ vertices.  Then
$G \zigzag H$ is the following $D_2^2$-regular graph on $N_1D_1$ vertices. The vertices are pairs $(u,i) \in [N_1]\times [D_1]$, and %
for $a,b\in [D_2]$, the $(a,b)$'th neighbor of a vertex $(u,i)$ is the vertex $(v,j)$ computed
as follows:
\begin{enumerate}
\item Let $i'$ be the $a$'th neighbor of $i$ in $H$.  (That is, take an $H$-step to move from $(u,i)$ to $(u,i')$.)
\item Let $v$ be the $i'$'th neighbor of $u$ in $G$, so $e=(u,v)$ is the $i'$'th edge leaving $u$.  Let $j'$ be such that $e$ is the $j'$'th edge
entering $v$ in $G$.  (That is, take a $G$-step to move from $(u,i')$ to $(v,j')$.)
\item Let $j$ be the $b$'th neighbor of $j'$ in $H$. (That is, take an $H$-step to move from $(v,j')$ to $(v,j)$.)
\end{enumerate}
\end{definition}

Note that the graph  $G\zigzag H$ depends on how the edges leaving and entering each vertex of $G$ are numbered.
Thus it is best thought of as an operation on labelled graphs. 
Nevertheless, the following lower bound on its expansion holds regardless of the labelling:

\begin{theorem}[\cite{ReingoldVaWi01}]\label{thm:zigzag} If $G$ has spectral expansion at least $\gamma_1$ and $H$ has spectral expansion at least $\gamma_2$, then $G\zigzag H$ has spectral expansion at least $\gamma_1\gamma_2^2$
\end{theorem}
$G$ should be thought of as a big graph and $H$ as a small
graph, where $D_1$ is a large constant and $D_2$ is a
small constant.  Observe that when $D_1>D_2^2$ the degree is reduced
by the zig-zag product.

Before giving intuition for Theorem~\ref{thm:zigzag}, let's see how it can be used to construct an infinite family of constant-degree expanders. %

\begin{construction}[Zig-Zag Based Expanders] \label{const:mild}
Let $H$ be a fixed $D$-regular graph on $D^4$ vertices with spectral expanion at least 7/8.\footnote{Since the number of vertices is polynomially related to the degree, such graphs are much easier to construct than constant-degree expanders, and there are a number of simple constructions. Alternatively, since we think of $D$ as a constant, $H$ can be found by exhaustive search.}
Define
\begin{eqnarray*}
G_1 &=& H^2\\
G_t &=& G_{t-1}^2\zigzag H \qquad \text{for $t>1$}
\end{eqnarray*}
\end{construction}

A straightforward induction, using Theorem~\ref{thm:zigzag} and the properties of squaring, shows that this is an infinite family of expanders: 
\begin{proposition}
For all $t$, $G_t$ is a $D^2$-regular graph on $D^{4t}$ vertices with spectral expansion at least $1/2$. 
\end{proposition}

Although simple to describe, Construction~\ref{const:mild} does not quite meet our definition of explicitness (Definition~\ref{def:explicit-expanders}), since the natural recursive way to compute neighbors in $G_t$ (by doing two neighbor computations in $G_{t-1}$) appears to take time exponential in $t$, which is polynomial in $N_t=D^{4t}$, rather than polylogarithmic. This can be remedied by tensoring in addition to squaring, so that the number of vertices grows much more quickly than the depth of the recursion.

There are two different intuitions underlying the expansion of the zig-zag product:
\begin{enumerate}
\item Given an initial distribution $(U,I)$ on the vertices of $G_1\zigzag G_2$ that is far from uniform,
there are two extreme cases. (Here we use capital letters to denote random variables corresponding to the lower-case
values in Definition~\ref{def:zigzag}.)  Either
\begin{enumerate}
\item All the (conditional) distributions $I|_{U=u}$ within
the clouds are far from uniform, {\em or} \label{case:zigzag-conditionalfar}
\item All the (conditional) distributions $I|_{U=u}$ within the clouds of size $D_1$ are uniform
(in which case the marginal distribution $U$ on the clouds must be far from uniform).
\label{case:zigzag-marginalfar}
\end{enumerate}
In Case~\ref{case:zigzag-conditionalfar}, the first $H$-step
$(U,I)\mapsto (U,I')$
already brings us closer to the uniform distribution, and the other two steps cannot hurt (as they
are steps on regular graphs).  In Case~\ref{case:zigzag-marginalfar}, the first $H$-step has no effect, but the $G$-step $(U,I')\mapsto (V,J')$
has the effect
of making the marginal distribution on clouds closer to uniform, i.e. $V$ is closer to uniform
than $U$.  But note that the joint distribution $(V,J')$ isn't actually any closer
to the uniform distribution on the vertices of $G_1\zigzag G_2$ because the $G$-step is a permutation.
Still, if the marginal distribution $V$ on clouds is closer to uniform, then the conditional distributions within the
clouds $J'|_{V=v}$ must have become further from uniform, and thus the second $H$-step
$(V,J')\mapsto (V,J)$ brings us closer to uniform.

This intuition can be turned into a formal proof, and with a careful analysis (which can be found in \cite{ReingoldVaWi01}) yields slightly better expansion bounds than stated in Theorem~\ref{thm:zigzag}. 

\item  A second intuition, which follows \cite{RozenmanVa05,ReingoldTrVa06}, leads to a very short of Theorem~\ref{thm:zigzag}.  Here
we think of the expander $H$ as behaving similarly to the complete graph on
$D_1$ vertices, via Lemma~\ref{lem:MatrixDecomposition}.  In the case that $H$ equals the complete graph, then it is easy to see that $G\zigzag H =
G\otimes H$, whose spectral expansion is equal to $\gamma(G)$ (since the complete graph has spectral expansion 1).  For general $H$, we use Lemma~\ref{lem:MatrixDecomposition} to decompose the random-walk matrix for $H$ into a convex combination of the random-walk matrix for the complete graph and an error matrix of spectral norm at most 1, with the coefficient on the complete graph being $\gamma(H)$.  Doing this for both steps on $H$ in the zig-zag product leads to a spectral expansion lower bound of $\gamma(H)^2\cdot \gamma(G)$.
\end{enumerate}

As mentioned earlier, Construction~\ref{const:mild} does not achieve an optimal relationship between spectral expansion and degree (which is $\gamma(G) = 1-\Theta(1/\sqrt{D})$, achieved by random graphs~
\cite{Friedman08} or explicit Ramanujan graphs~\cite{LubotzkyPhSa88,Margulis88}). However, in subsequent work with Capalbo, Reingold, and Vadhan~\cite{CapalboReVaWi02}, Wigderson used a variant of the zig-zag product to construct near-optimal directed or bipartite {\em vertex expanders}, namely constant-degree graphs where sets of size up to $K=\Omega(N)$ expand by a factor of $A=(1-\eps)\cdot D$. (Viewed as bipartite graphs, the expansion is from the left side of the graph to the right side of the graph, corresponding the use of out-neighborhoods in Definition~\ref{def:vertex-expansion}.) %
This variant of the zig-zag product comes from viewing expanders as forms of randomness extractors (as discussed in the next section), and builds on the first intuition for the zig-zag product given above.  This was the first explicit construction of constant-degree graphs with expansion factor $A>D/2$, which has a qualitative implication that is important in a number of applications: they are also {\em unique-neighbor expanders}, where every left-set $S$ of size at most $K$ has at least one neighbor (in fact, at least $(1-2\eps)D$ neighbors) on the right that is incident to exactly one vertex in $S$.

A different variant of the zig-zag product was introduced by Ben-Aroya and Ta-Shma~\cite{BenAroyaTa08} and used to give a combinatorial construction of ``almost-Ramanujan'' expanders (namely with $\gamma(G) = 1-1/D^{1/2-o(1)}$, where the $o(1)\rightarrow 0$ as $D\rightarrow \infty$). This same variant was then used by Ta-Shma~\cite{TaShma17} in his breakthrough construction of linear error-correcting codes (aka small-biased sets) that nearly meet the Gilbert--Varshamov bound. 

With Alon and Lubotzky~\cite{AlonLuWi01}, Wigderson gave an intriguing algebraic interpretation of the zig-zag product: Under certain conditions, if $G$ and $H$ are Cayley graphs, then $G\zigzag H$ is a Cayley graph for the {\em semi-direct product} of the underlying groups.  Using this connection, they answered a question of Lubotzky and Weiss~\cite{LubotzkyWe92} and proved that expansion of Cayley graphs is {\em not} a group property: 
a group can have two constant-sized sets of generators, such that the Cayley graph defined by one is expanding and the other is not.  With Meshulam~\cite{MeshulamWi04} and Rozenman and Shalev~\cite{RozenmanShWi06}, Wigderson further used this group-theoretic zig-zag to obtain iterative constructions of expanding Cayley graphs.

Perhaps the most striking application of the zig-zag product is Reingold's algorithm for \USTConn~\cite{Reingold06}, which we will see in Section~\ref{sec:unconditional}, which in turn inspired Dinur's celebrated combinatorial proof of the PCP Theorem~\cite{Dinur07}.  (See Section~\ref{sec:PCP} for discussion of the PCP Theorem.)

Wigderson has also formulated and initiated the study of many other variants of expansion, such as expanding hypergraphs~\cite{FriedmanWi95}, monotone expanders~\cite{DvirWi10}, and notions of expansion for collections of linear maps~\cite{LiQiWiWiZh23}.

\subsubsection{Randomness Extractors}

Randomness extractors are functions that extract almost-uniform bits from sources
of biased and correlated bits.
The original
motivation for extractors was to simulate randomized
algorithms with weak random sources as might arise in nature.  This
motivation is still compelling, but extractors have taken on a much
wider significance in the years since they were introduced.  They
have found numerous applications in theoretical computer science
beyond this initial motivating one, in areas from
cryptography to distributed algorithms to hardness of approximation.  (See the surveys~\cite{NisanTa99,Vadhan12-fnttcs,Shaltiel04}.)
In this section, we will survey Wigderson's numerous contributions to the theory of extractors, their constructions, and their applications.  Many of these contributions involve developing and exploiting the close connection between randomness extractors and expander graphs.

We begin with some probability definitions that are needed to introduce randomness extractors.

\begin{definition} \label{def:tvd}
For random variables $X$ and $Y$ taking values in $\Uni$, their {\em
statistical difference}  (also known as {\em total variation distance})
is $\Delta(X,Y) = \max_{T \subseteq \Uni} |\Pr[X \in T] - \Pr[Y \in
T]|$. We say that $X$ and $Y$ are {\em $\eps$-close} if
$\Delta(X,Y) \leq \eps$.
\end{definition}
Recall that random variables being $\eps$-close is equivalent to them being $(\infty,\eps)$-indistinguishable (Definition~\ref{def:compind}).

\begin{definition}[entropy measures]
Let $X$ be a discrete random variable.  Then
\begin{itemize}
\item the {\em Shannon entropy} of $X$ is:
\[\HSha(X)=\Exp_{x\getsr X}\left[\log \frac{1}{\pr{X=x}}\right].\]

\item the {\em \Renyi\ entropy} of $X$ is:
\[\HRen(X)=\log \left(\frac{1}{\Exp_{x\getsr X}[\pr{X=x}]}\right) \text{ and}\]

\item the {\em min-entropy} of $X$ is:
\[\Hmin(X)=\min_x\left\{\log\frac{1}{\pr{X=x}}\right\},\]
\end{itemize}
where all logs are base 2.
\end{definition}

\begin{fact} \label{fact:entropies}
\begin{enumerate}
\item  For every random variable $X$, $$\Hmin(X)\leq \HRen(X)\leq \HSha(X),$$ with equality iff $X$ is uniform on its support.
\item  For every random variable $X$, $\HRen(X)\leq 2\Hmin(X)$, and for every $\eps>0$, there is a random variable
$X'$, such that $\HRen(X')\leq \Hmin(X)+\log(1/\eps)$. \label{part:minvsRenyi}
\end{enumerate}
\end{fact}

To illustrate the differences between the three notions, consider a source $X$ such that
$X=0^n$ with probability $0.99$ and $X=U_n$ with probability $0.01$.
Then $\HSha(X)\geq 0.01n$ (contribution from the uniform distribution), $\HRen(X)\leq \log
(1/.99^2)<1$ and $\Hmin(X)\leq \log (1/.99)<1$
(contribution from $0^n$).
Note that even though $X$  has Shannon entropy linear in $n$, we cannot expect
to extract bits that are close to uniform or carry out any useful randomized computations
with one sample from $X$, because it gives us nothing useful 99\% of the time.
Thus, we should use the stronger measures of entropy given by
$\HRen$ or $\Hmin$.  These entropy measures were introduced into the randomness extraction literature by Cohen and Wigderson~\cite{CohenWi89} and Chor and Goldreich~\cite{ChorGo88}, respectively.

We will consider the task of extracting randomness from sources where all we know is a lower bound
on the min-entropy (which is equivalent to a lower bound on \Renyi\ entropy by Fact~\ref{fact:entropies}):
\begin{definition}
A random variable $X$ is a {\em $k$-source} if $\Hmin(X)\geq k$, i.e., if
$\pr{X=x}\leq 2^{-k}$ for all $x$.
\end{definition}

A typical setting of parameters is $k=\delta n$ for some fixed $\delta$, e.g., $\delta=1/10$.  We
call $\delta$ the {\em min-entropy rate}.  Some different ranges that
are commonly studied (and are useful for different applications): $k=\polylog(n)$, $k=n^{\gamma}$ for a constant $\gamma\in (0,1)$,
$k=\delta n$ for a constant $\delta\in (0,1)$, and $k=n-O(1)$.  The middle two
($k=n^\gamma$ and $k=\delta n$) are the most natural for simulating randomized algorithms with
weak random sources.

An ideal goal for a randomness extractor is to take one sample from an unknown $k$-source as input and output almost-uniformly distributed bits.  Unfortunately, this is impossible to achieve: 
\begin{proposition}
For any $\Ext:\zo^n\rightarrow\zo$ there exists an $(n-1)$-source $X$ so that $\Ext(X)$
is constant.
\end{proposition}

\begin{proof}
There exists $b\in \zo$ so that $|\Ext^{-1}(b)|\geq 2^n/2=2^{n-1}$. Then let
$X$ be the uniform distribution on $\Ext^{-1}(b)$.
\end{proof}

Thus, instead researchers turned to the problem of simulating randomized algorithms with a weak random source.  That is, suppose we have a language $L\in \BPP$.  The $\BPP$ algorithm for $L$ assumes a source of truly uniform and independent bits.  Can we decide membership in $L$ in polynomial time if we are instead given one sample from a $k$-source $X$ with large enough min-entropy $k$?  Of course, the answer is yes if $\BPP=\P$, but here we want unconditional results, not assuming circuit lower bounds like Theorem~\ref{thm:IW}.  With Cohen~\cite{CohenWi89}, Wigderson gave the first positive answer to this question for sources of constant entropy rate, namely $\delta=k/n>3/4$. This was then improved to any constant entropy rate $\delta>0$ by Zuckerman~\cite{Zuckerman96}, and then these approaches were abstracted by Nisan and Zuckerman~\cite{NisanZu96} into the following elegant definition of a randomness extractor:

\begin{definition}[seeded extractors~\cite{NisanZu96}] \label{def:extractor}
A function $\Ext : \{0,1\}^n \times \{0,1\}^d  \rightarrow
\{0,1\}^m$ is a {\em $(k, \eps)$-extractor} if for every
$k$-source $X$ on $\{0,1\}^n$,  $\Ext(X, U_d)$ is $\eps$-close
to $U_m$.
\end{definition}

That is, an extractor extracts almost-uniform bits given one sample from a $k$-source and a {\em seed} consisting of $d$ truly random bits.  The point is that if $d$ is small enough, such as $d=O(\log n)$, we can eliminate the seed entirely by trying all $2^d$ possibilities rather than choosing it at random, similarly to Proposition~\ref{prop:enumeration}.\footnote{The similarity of this approach to derandomization via pseudorandom generators is not a coincidence.  Trevisan~\cite{Trevisan01} showed that Wigderson et al.'s conditional construction of pseudorandom generators from circuit lower bounds (Theorem~\ref{thm:IW}) can also be interpreted as an unconditional construction of randomness extractors! Indeed, the same holds for any construction of pseudorandom generators from a ``black-box'' hard function $f$, and thus Wigderson's two lines of work on pseudorandom generators and extractors were unified.}

Indeed, using the Probabilistic Method, it can be shown that seed length $d=O(\log n)$ is possible: 
\begin{theorem}[\cite{Sipser88,Zuckerman97}] \label{thm:probabilistic-extractor}
For every $n\in \N$, $k\in [0,n]$ and $\eps > 0$, there
exists a $(k, \eps)$-extractor $\Ext : \{0,1\}^n \times
\{0,1\}^d  \rightarrow \{0,1\}^m$ with $m = k + d - 2
\log(1/\eps) - O(1)$ and $d = \log(n-k) + 2 \log
(1/\eps) + O(1)$.  Indeed, a randomly chosen function $\Ext$ with these parameters is a $(k,\eps)$-extractor with high probability.
\end{theorem}
Both the lower bound on the output length $m$ and upper bound on the seed length $d$ can be shown to be optimal up to additive constants for almost all settings of parameters~\cite{RadhakrishnanTa00}.
A small constant $\eps$, say $\eps=1/8$, can be shown to be sufficient for simulating randomized algorithms with a weak random source.  In this case, the seed length is $d=\log(n-k)+O(1)$ and we extract all but $O(1)$ of the $k+d$ bits of entropy that is fed is into the extractor as input.

However, like with expanders, for applications of extractors, we typically need explicit constructions, ones where $\Ext$ is computable in polynomial time. There was a long line of work giving increasingly improved constructions of extractors, and a milestone was achieved by Wigderson, together with Lu, Reingold, and Vadhan~\cite{LuReVaWi03}, who gave explicit extractors that are optimal up to constant factors.

\begin{theorem}[\cite{LuReVaWi03}] \label{thm:LRVW}
For all constants $\eps,\alpha>0$, and all $n,k\in \N$, there is an explicit
$(k,\eps)$-extractor $\Ext : \zo^n\times \zo^d\rightarrow \zo^m$ with $d=O(\log n)$ and $m=(1-\alpha)\cdot k$.
\end{theorem}
In fact, the error parameter $\eps$ in Theorem~\ref{thm:LRVW} can be made subconstant, even almost polynomially small.  Constructions with no constraint on $\eps$ were later given by Guruswami, Umans, and Vadhan~\cite{GuruswamiUmVa09} and by Dvir and Wigderson~\cite{DvirWi11}, the latter being based on Dvir's resolution of the Kakeya problem in finite fields~\cite{Dvir09}.  For taking the entropy loss rate $\alpha$ parameter to be subconstant, the first construction was given earlier than Theorem~\ref{thm:LRVW} by Wigderson and Zuckerman~\cite{WigdersonZu99}, but had seed length $d=\polylog(n)$ rather than $d=O(\log n)$.  Subsequent to Theorem~\ref{thm:LRVW}, Dvir, Kopparty, Saraf, and Sudan~\cite{DvirKoSaSu09} achieved $d=O(\log n)$ with $\alpha=1/\polylog(n)$.

\paragraph{\underline{Extractors vs. Expanders}}

The works of Wigderson with Cohen~\cite{CohenWi89} and with Friedman~\cite{FriedmanWi95} showed that explicit constructions of certain kinds of imbalanced bipartite expanders suffice for simulating randomized algorithms with weak sources of randomness. Building on this connection, the Nisan--Zuckerman definition of seeded extractors~\cite{NisanZu96} can be interpreted graph-theoretically as follows. Given any function
$\Ext:\{0,1\}^n\times\{0,1\}^d\longrightarrow\{0,1\}^m$, we can view $\Ext$ as a
bipartite graph $G$ with $N=2^n$ vertices on the left, $M=2^m$ vertices on the right, and left-degree $D=2^d$, where the $y$'th neighbor of $x\in \zo^n$ is $\Ext(x,y)$.

Suppose $\Ext$ is a $(k,\eps)$-extractor.  Then given any set $S\subseteq \{0,1\}^n$ of size $K=2^k$, the uniform distribution on $S$, which we'll denote $U_S$, is a $k$-source.  The extractor property tells us that $\Ext(U_S,U_{[D]})$ is $\eps$-close to uniform on $[M]$.  That is, a random neighbor of a random element of $S$ is $\eps$-close to uniform on the right-hand vertices of $G$.  In particular, $|N(S)|\geq (1-\eps)M$. This property is just like
vertex expansion, except that it ensures a large neighborhood for sets of size exactly $K$ (rather than all sets of size at most $K$).
Indeed, this variant of vertex expansion was introduced in
graph-theoretic form in \cite{Pippenger85,Santha87,Sipser88}, and is equivalent to the following relaxation of extractors.

\begin{definition}[dispersers] \label{def:disperser}
A function $\Disp : \{0,1\}^n \times \{0,1\}^d  \rightarrow
\{0,1\}^m$ is a {\em $(k, \eps)$-disperser} if for every
$k$-source $X$ on $\{0,1\}^n$,  $\Disp(X, U_d)$ has a support of
size at least $(1-\eps)\cdot 2^m$.
\end{definition}

Despite this connection, the parameters most commonly studied for extractors/dispersers and expanders 
are quite different.  Extractors and dispersers typically have polylogarithmic degree (e.g. $D=\polylog(N)$, corresponding to seed length $d=O(\log n)$), are very imbalanced (e.g. $M=N^\delta$ for a constant $\delta\in (0,1)$), and often do not actually `expand' (i.e. $|N(S)|<|S|$, since we are generally satisfied with retaining entropy, not necessarily increasing it).
Nevertheless, in the ``high min-entropy regime"  $k=(1-o(1))n$, extractors and expanders become more closely related, and indeed Goldreich and Wigderson~\cite{GoldreichWi97} showed that by taking a power of a constant-degree spectral expander, we obtain the following ``high min-entropy extractors'':

\begin{theorem}[\cite{GoldreichWi97}] \label{thm:GW}
For every $n,k\in \N$ and $\eps>0$, there is an explicit $(k,\eps)$-extractor
$\Ext:\{0,1\}^n\times\{0,1\}^d\longrightarrow\{0,1\}^n$ with
$d=O(n-k+\log(1/\eps))$.
\end{theorem}

Note that the seed length of this extractor is linear rather than logarithmic, but importantly it is linear in $n-k$ rather than just $n$.  So when $k=n-o(\log n)$, the seed length is shorter than that of Theorem~\ref{thm:LRVW}.  The origin of Wigderson's zig-zag product described in Section~\ref{sec:expanders} was in the context of extractors, to compose extractors such as given in
Theorem~\ref{thm:LRVW} and in Theorem~\ref{thm:GW} to obtain a ``best of both'' seed length of $O(\log(n-k))$~\cite{ReingoldVaWi00-old}.

Wigderson's constant-degree expanders with expansion $(1-\eps)D$~\cite{CapalboReVaWi02} came from considering a common generalization of expanders and extractors.  In applying an extractor, any distribution $X$ that has large enough (min-)entropy gets transformed into one that is close to uniform.  In contrast, a random step on expander transforms any distribution $X$ that does {\em not} have too much entropy into one with higher entropy.  Formally, a spectral expander can be interpreted as one that increases \Renyi\ entropy (noting that the expression $\Exp_{x\getsr X}[\pr{X=x}]$ that appears in the definition of \Renyi\ entropy equals the squared $\ell_2$ norm of the probability mass function of $X$).   To bridge the two, we can ask for a function $\Con : \zo^n\times \zo^d\rightarrow \zo^m$ such that for every random variable $X$ of min-entropy $k\leq \kmax$, it holds that $\Con(X,U_d)$ is $\eps$-close to having min-entropy at least $k+a$.  Such a function is a necessarily a $(\Kmax,(1-\eps)A)$ vertex expander (where $\Kmax=2^{\kmax}$ and $A=2^a$), and in fact if $a=d$, the converse holds as well~\cite{TaShmaUmZu07}.  A general abstraction of {\em randomness conductors} that encompasses all of these notions was given in \cite{CapalboReVaWi02}, and a zig-zag product for conductors was developed and used to obtain constant-degree bipartite expanders with expansion $(1-\eps)\cdot D$.

The $\ell_2$-to-$\ell_1$ switch from requiring that \Renyi\ entropy increases to only requiring that the output distribution is $\eps$-close in total variation distance to having higher entropy is crucial for enabling these results. Indeed, it is impossible to derive expansion greater than $D/2$ from spectral expansion alone~\cite{Kahale95}.  Already in Wigderson's earlier work with Zuckerman~\cite{WigdersonZu99}, randomness extractors were used to construct balanced bipartite vertex expanders of non-constant degree that are impossible to derive from spectral expansion.

\subsubsection{Multi-source Extractors and Ramsey Graphs}

In the previous section, we argued that seeded extractors (Definition~\ref{def:extractor}) suffice for simulating randomized algorithms with a single sample from a weak random source because we can enumerate over all possible seeds in polynomial time.  However, this trick does not work for a number of other applications of randomness, such as in cryptography, distributed computing, and Monte Carlo simulation, where it is not clear how to combine the results from enumeration.  Thus, it is natural to ask whether we can extract almost-uniform bits given {\em only} access to weak sources of randomness, i.e. with no uniformly random seed. 

For example, we could consider extracting randomness from a small number of independent $k$-sources, a problem first studied by
Chor and Goldreich~\cite{ChorGo88}. 
That is we want a function $\Ext : (\zo^n)^c \rightarrow \zo^m$ such that for all independent random variables $X_1,X_2,\ldots,X_c$ where each $X_i$ is a $k$-source, 
$\Ext(X_1,X_2,\ldots,X_c)$ is $\eps$-close to $U_m$.  Or we could weaken the requirement to that of a disperser, where we only require that the output has support size at least $(1-\eps)\cdot 2^m.$

In addition to their motivation for obtaining high-quality randomness, extractors for $c=2$ independent sources are of interest because of connections to communication complexity and to Ramsey theory.  In particular, a disperser for 2 independent $k$-sources of length $n$ with output length $m=1$ is {\em equivalent} to a {\em bipartite Ramsey graph} --- a bipartite graph with $N$ vertices on each side that contains no $K\times K$ bipartite clique or $K\times K$ bipartite independent set (for $N=2^n$ and $K=2^k$): connect left vertex $x$ and right vertex $y$ iff $\Disp(x,y)=1$.
Giving explicit constructions of Ramsey graphs
that approach $K=O(\log N)$ bound given by the Probabilistic Method~\cite{Erdos47} is a long-standing open problem
posed by \Erdos~\cite{Erdos69}.

Chor and Goldreich~\cite{ChorGo88} gave extractors for 2 independent sources of min-entropy rate $\delta$ (i.e. $k$-sources on $\zo^n$ with $k=\delta n$) when $\delta>1/2$, and there was no improvement in this bound for nearly 2 decades. 
Substantial progress began again in Wigderson's work with Barak and Impagliazzo~\cite{BarakImWi06}, who used new results in arithmetic combinatorics to construct extractors for a constant number of independent sources of min-entropy rate $\delta$ for an arbitrarily small constant $\delta>0$.  Specifically, they used the Sum--Product Theorem over finite fields of Bourgain, Katz, and Tao~\cite{BourgainKaTa04}; this theorem says that for $p$ prime and every subset $A\subseteq \F_p$  whose size is not too close to $p$, either the set $A+A$ of pairwise sums or the set $A\cdot A$ of pairwise products is of size significantly larger than $|A|$. Using this theorem and other results in additive number theory, Barak, Impagliazzo, and Wigderson show that if $A$, $B$, $C$ are random variables distributed in $\F_p$ with min-entropy rate $\delta<.9$, then $A\cdot B+C$ is $\eps$-close to having min-entropy rate $(1+\alpha)\cdot \delta$ for a universal constant $\alpha>0$.  Recursively applying this result reduces the task of extracting from $\poly(1/\delta)$ sources of min-entropy $\delta n$ to extracting from 2 sources of min-entropy rate larger than $1/2$, which allows for applying the Chor--Goldreich extractor~\cite{ChorGo88}.

In subsequent works, Wigderson obtained even better multi-source extractors and dispersers.  With
Barak, Kindler, Shaltiel, and Sudakov~\cite{BarakKiShSuWi10}, Wigderson constructed explicit extractors for 3 sources of min-entropy $k=\delta n$~\cite{BarakKiShSuWi10}.  With Barak, Rao, and Shaltiel~\cite{BarakRaShWi12},
Wigderson constructed dispersers for 2 sources of min-entropy $k=n^{o(1)}$~\cite{BarakRaShWi12}, or equivalently bipartite Ramsey graphs that avoid $K\times K$ cliques and independent sets of size $K=2^{(\log N)^{o(1)}}$.
This latter
result was a major improvement over the previous best explicit construction of Ramsey graphs by Frankl and Wilson~\cite{FranklWi81}, which had $K=2^{\sqrt{n}}$ and only applied to the nonbipartite case.  A long line of subsequent work has continued to improve the parameters of 2-source extractors and dispersers, and very recently Li~\cite{Li23} has achieved 2-source extractors for min-entropy $k=O(\log n)$, which is optimal up to a constant factor, and thus bipartite Ramsey graphs for $K=\polylog(N)$, which is optimal up to the constant in the exponent.

\subsection{Unconditional derandomization} \label{sec:unconditional}

Theorem~\ref{thm:IW} of Wigderson and collaborators gives strong evidence that randomness does not provide a substantial gain in the efficiency of algorithms, but it assumes circuit lower bounds that we are very far from proving. Thus, together with Ajtai~\cite{AjtaiWi89}, Wigderson asked whether  there are large classes of algorithms that we can {\em unconditionally derandomize}, namely without making any unproven complexity assumptions.\footnote{The work of Ajtai and Wigderson~\cite{AjtaiWi89} actually preceded Theorem~\ref{thm:IW}, but was instead motivated by Yao's proof~\cite{Yao82} that $\BPP\subseteq \SUBEXP$ under the assumption that cryptographic pseudorandom generators exist.}  They showed that this is indeed possible, giving an unconditional subexponential-time derandomization of probabilistic constant-depth circuits. After that, unconditional derandomization became
a huge area of research, which is still flourishing. We refer the reader to the survey by Hatami and Hoza~\cite{HatamiHo23} for recent developments in the area.

\subsubsection{Undirected S-T Connectivity}

One subclass of $\BPP$ that has proved amenable to unconditional derandomization is 
$\BPL$, where we restrict the algorithms to use a logarithmic amount of space.  (When we measure the space
complexity of an algorithm, we only count the read-write working memory, and do not count the space needed for the read-only input and write-only output.)

\begin{definition} 
A language $L$ is in $\BPL$ if there exists a randomized algorithm $A$ that always halts,
uses space at most $O(\log n)$ on inputs of length $n$, and satisfies the following for all inputs $x$:
\begin{itemize}
\item  $x \in L \Rightarrow \Pr[\mbox{$A(x)$ accepts}] \geq 2/3$.

\item  $x \not\in L \Rightarrow \Pr[\mbox{$A(x)$ accepts}] \leq 1/3$.
\end{itemize}
\end{definition}

The standard model of a randomized space-bounded machine is one that has access to a coin-tossing
box (rather than an infinite tape of random bits), and thus must explicitly store in its workspace
any random bits it needs to remember.  The requirement that $A$ always halts ensures that its
running time is at most $2^{O(\log n)} = \poly(n)$, because otherwise there would be a loop in its
configuration space.  Thus $\BPL\subseteq \BPP$.

Similarly to the time case (Definition~\ref{def:DTIME}), we can ask what is the smallest deterministic space bound needed to simulate $\BPL$:

\begin{definition}[Deterministic Space Classes]
\label{def:DSPACE}
\vspace{-4ex}
$$
\begin{array}{rcl}
\DSPACE(s(n)) &=& \{L : \mbox{$L$ can be decided deterministically in
space $O(s(n))$}\} \\
\L & = & \DSPACE(\log n) \\
\L^c &=&  \DSPACE(\log^c n) 
\end{array}
$$
\end{definition}
Classic results in complexity theory~\cite{BorodinCoPi83,Jung81}
tell us that $\BPL\subseteq \L^2$; however, this is not really a result about randomized algorithms, since it applies even for the unbounded-error version of $\BPL$ (where inputs in $L$ are accepted with probability greater than 1/2 and inputs not in $L$ with probability at most $1/2$).  Thus the interesting question is whether we can show $\BPL = \L$ (randomization provides only a constant-factor savings in memory), or at least $\BPL \subseteq \L^c$ for a constant $c < 2$.

The potential power of randomization for logspace algorithms was first demonstrated in the late 1970's for the following basic problem:
\begin{compprob}
\USTConn:  Given an undirected graph $G$ and two vertices $s$ and $t$,
is there a path from $s$ to $t$ in $G$?
\end{compprob}
Basic algorithms like breadth-first or depth-first search solve \USTConn\ in linear time, but also take linear space. With randomization we can solve the problem in only logarithmic space:

\begin{theorem}[\cite{AleliunasKaLiLoRa79}] \label{thm:AKLLLR}
\USTConn\ is in $\BPL$.
\end{theorem}

\begin{proofsketch} 
The algorithm simply does a polynomial-length random walk starting at $s$:

\begin{algorithm}[\USTConn\ via Random Walks] \label{alg:AKLLLR}
\alginput{$(G,s,t)$, where $G=(V,E)$ has $n$ vertices.}
\begin{algenumerate}
\item Let $v=s$.
\item Repeat $\poly(n)$ times:
\begin{enumerate}
\item If $v=t$, halt and accept.
\item Else randomly select $v\getsr \{ w : (v,w)\in E\}$.
\end{enumerate}
\item Reject (if we haven't visited $t$ yet).
\end{algenumerate}
\end{algorithm}

Notice that this algorithm only requires space $O(\log n)$, in order to maintain the current vertex $v$ as well as a counter
for the number of steps taken.  Clearly, it never accepts when there isn't a path from $s$ to $t$.  It can be shown that in any connected undirected graph, a random walk of length $\poly(n)$ from one vertex will hit any other vertex with high probability.  Applying this to the connected component containing $s$, it follows that the algorithm accepts with high probability when $s$ and $t$ are connected.  %
\end{proofsketch}

Using Nisan's pseudorandom generator for space-bounded computation~\cite{Nisan92}, Wigderson, together with Nisan and Szemer\'edi~\cite{NisanSzWi92}, proved that \USTConn\ is in $\L^{3/2}$.  Inspired by that result, Saks and Zhou~\cite{SaksZh99} then proved that $\BPL\subseteq \L^{3/2}$, which remains essentially the best derandomization of $\BPL$ to date.\footnote{Recently, Hoza~\cite{Hoza21}
gave a slight improvement, showing that
$\BPL \subseteq \DSPACE(\log^{3/2} n/\sqrt{\log\log n})$.}
Then Wigderson, together with Armoni, Ta-Shma, and Zhou~\cite{ArmoniTaWiZh00}, proved that \USTConn\ is in $\L^{4/3}$.
In 2005, Reingold~\cite{Reingold06} finally resolved the space complexity of \USTConn:

\begin{theorem}[\cite{Reingold06}] \label{thm:Reingold} \USTConn is in $\L$. \end{theorem}  

Reingold's Theorem is based on the following two ideas:
\begin{itemize}
\item \USTConn\ can be solved in logspace on constant-degree expander graphs. More precisely, it is easy on constant-degree
graphs where every connected component is promised to be an expander (i.e. has spectral expansion bounded away from 0): we can try
all paths of length $O(\log N)$ from $s$ in logarithmic space; this works because expanders have logarithmic diameter.

\item  The same operations that Reingold, Vadhan, and Wigderson~\cite{ReingoldVaWi01} used to construct an infinite expander family (described Section~\ref{sec:expanders}) can also be used to
turn {\em any} graph into an expander (in logarithmic space).  There, we started with a constant-sized expander
and used various operations to build larger and larger expanders.  The goal was to increase the size of the
graph (which was accomplished by zig-zag and/or tensoring),
while preserving the degree and the expansion (which was accomplished by zig-zag and squaring).  Here, we want to {\em improve} the expansion (which is accomplished by squaring),
while preserving the degree (as is handled by zig-zag) and ensuring the graph remains of polynomial size (so tensoring is counterproductive and not used).
\end{itemize}

\subsubsection{General Space-Bounded Computation}

Like in the time-bounded case, one of the main approaches to derandomizing $\BPL$ is to construct pseudorandom generators $G : \zo^d\rightarrow
\zo^n$ such that no randomized $(\log n)$-space algorithm can distinguish $G(U_d)$ from $U_n$.   In order to get
derandomizations that are correct on every input $x$, we require pseudorandom generators that fool {\em nonuniform} space-bounded algorithms.  Since randomized space-bounded algorithms get their random bits as a stream of coin tosses, we only need to fool space-bounded distinguishers that read each of their input bits once, in order. Thus, instead of boolean circuits, we want pseudorandom generators for the following class of distinguishers:
\begin{definition}
An {\em ordered branching program} $B$ of {\em width} $w$ and {\em length} $n$ is given by a {\em start state} $s_0\in [w]$, $m$ {\em transition functions} $B_1,\ldots,B_n : [w]\times \zo\rightarrow [w]$, and a set $A\subseteq [w]$ of {\em accept states}.  On an input $x\in \zo^n$, $B$ computes by updating its state via the rule $s_{i}=B_i(s_{i-1},x_i)$ for $i=1,\ldots,n$ and accepting iff $s_n\in A$.  
\end{definition}

The width $w$ of a branching program corresponds to a space bound of $\log w$ bits.
Similarly to Theorem~\ref{thm:PRGderandomize}, a family of generators $G_n : \zo^{d(n)}\rightarrow \zo^n$ that is computable in space $O(d(n))$ and such that $G_n(U_{d(n)})$ cannot be distinguished from $U_n$ by ordered branching of width $w=n$ implies that $\BPL \subseteq \bigcup_c \DSPACE(c\log n+d(n^c))$.  (Enumerating all seeds of length $d(m)$ only requires an additive space increase of $d(m)$.)  In particular, a pseudorandom generator with seed length $d(n)=O(\log^c n)$ immediately implies $\BPL\subseteq \L^c$.

Unfortunately, the best known pseudorandom generator for general space-bounded computation is Nisan's generator~\cite{Nisan92}, whose seed length of $O(\log^2 n)$ does not improve on the bound $\BPL\subseteq \L^2$.  Nevertheless, Saks and Zhou~\cite{SaksZh99} used Nisan's generator as part of a more sophisticated algorithm to obtain their result that $\BPL\subseteq \L^{3/2}$.

Together with Impagliazzo and Nisan~\cite{ImpagliazzoNiWi94}, Wigderson gave an appealing alternative to Nisan's generator that has been the subject of much subsequent research and improved analyses for restricted models of ordered branching programs:

\begin{definition}
    Given a sequence of regular digraphs $\fH=(H_1,\ldots,H_\ell)$ where $\deg(H_i)=d_i$ and $|V(H_i)|=2\prod_{j=1}^{i-1}d_j$, the \textbf{INW generator constructed with $\fH$}, denoted $\INW_\fH$ or $\INW_\ell$ when the family is clear, is the function
    defined recursively where for $x\in \zo$ we have $\INW_0(x)=x$ and for
    $x\in V(H_i)$ and $y\in [d_i]$, we have
    \begin{equation} \label{eqn:INW-recursion}
    \INW_{i}(x,y)=(\INW_{i-1}(x),\INW_{i-1}(H_{i}[x,y])),
    \end{equation}
    where $H_{i}[x,y]$ denotes the $y$'th neighbor of vertex $x$ in the graph $H_{i}$.
    $\INW_i$ thus generates an output of length $2^i$ using a seed of length $\left\lceil \log\left(2\prod_{i=1}^\ell d_i\right)\right\rceil$.
\end{definition}

That is, $\INW_i$ correlates the seeds of $\INW_{i-1}$ used to generate the first $2^{i-1}$ bits and the second $2^{i-1}$ bits as neighbors in the graph $H_i$.
Impagliazzo, Nisan, and Wigderson~\cite{ImpagliazzoNiWi94} proved that an instantiation of this generator
fools logspace algorithms with a seed length of $O(\log^2 m)$.  They did this by analyzing the construction when the graphs $H_i$ are good spectral expanders:
\begin{theorem}[\cite{ImpagliazzoNiWi94}] \label{thm:INW}
If every graph $H_i$ has spectral expansion at least $1-\sigma$, then $\INW_\ell$ $\eps$-fools ordered branching programs of with $w$ and length $n=2^\ell$ with error
at most $\eps=\sigma\cdot nw$.
\end{theorem}
To achieve spectral expansion $1-\sigma$, we can use explicit expanders $H_i$ with degree $d_i=\poly(1/\sigma)$, and hence get seed length 
$$\left\lceil \log\left(2\prod_{i=1}^\ell d_i\right)\right\rceil
= O\left(\log n\cdot \log\left(\frac{1}{\sigma}\right)\right).$$
To achieve error $\eps$ by Theorem~\ref{thm:INW}, we should set $\sigma = \eps/nw$, and thus we get seed length $O(\log n \cdot \log(nw/\eps))$, exactly matching Nisan~\cite{Nisan91} and giving seed length $O(\log^2 n)$ when $w=n$ and $\eps=1/8$ as needed for derandomizing $\BPL$.

To get intuition for Theorem~\ref{thm:INW}, notice that if took the graph $H_i$ to be complete graphs with self-loops, then in Expression~(\ref{eqn:INW-recursion}) for $\INW_i$ we would be using independent seeds for the left half and right half, so the error (distinguishing advantage) of $\INW_{i}$ should be at most twice the error of $\INW_{i-1}$ (since we are using it twice).  Furthermore, since an expander with spectral expansion at least $1-\sigma$ approximates the complete graph to within spectral norm at most $\sigma$, we incur an additional error of at most $\sigma w$ in the $i$'th level of recursion, where we pay a factor of $w$ by summing the error over the $w$ possible states of the branching program at the halfway point.  Thus the error $\eps_i$ for $\INW_i$ can be bounded by the recurrence $\eps_i\leq 2\eps_{i-1}+\sigma w$, which solves to $\eps_\ell \leq (2^{\ell}-1)\cdot \sigma w < \sigma \cdot nw$.

Impagliazzo, Nisan, and Wigderson~\cite{ImpagliazzoNiWi94} actually proved that the INW generator fools a wider class of algorithms than ordered branching programs, called {\em network algorithms}.  Subsequent developments, however, have focused on obtaining improved analyses for more restricted classes of ordered branching programs, namely {\em regular} and {\em permutation} branching programs. An ordered branching program $B$ is a {\em permutation program} if for every $i$ and bit $x_i\in \zo$, the transition function $B_i(\cdot,x_i) : [w]\rightarrow [w]$ is a permutation on the state set.  That is, the transitions are reversible (for any fixed input $x$).  A {\em regular} branching program is more general and just requires that for every state $s_i \in [w]$, there are exactly two pairs $(s_{i-1},x_i)\in [w]\times \zo$ such that $B_i(s_{i-1},x_i)=s_i$.  A more intuitive formulation of regularity comes from thinking of each transition function $B_i$ of the branching program as a bipartite graph with $w$ vertices on each side, where left-vertex $s_{i-1}$ is connected to right-vertices $B_i(s_{i-1},0)$ and $B_i(s_{i-1},1)$; in this viewpoint, a branching program is regular iff all of its associated bipartite graphs are regular.  (They are always 2-leftregular; the additional requirement here is that they are also 2-rightregular.)  One motivation for studying pseudorandomness for regular branching programs is that a general ordered branching program of width $w$ and length $n$ can be simulated by an ordered regular branching program of width $wn$~\cite{ReingoldTrVa06,BogdanovHoPrPy22,LeePyVa23}.

The \USTConn\ problem can be reduced to estimating the acceptance probability of an ordered permutation branching program, and it was shown by \cite{RozenmanVa05} that an instantiation of the INW generator with seed length $O(\log n)$ can be used to derandomize Algorithm~\ref{alg:AKLLLR} on the corresponding graphs and thus give a simpler proof of Reingold's Theorem (Theorem~\ref{thm:Reingold}).  Next, it was shown in \cite{BravermanRaRaYe10} showed that the INW generator fools ordered regular branching programs with seed length $O(\log n \cdot\log\log n + \log n \cdot \log(w/\eps))$.  Note that this seed length is nearly linear rather than quadratic in $\log n$.  In \cite{KouckyNiPu11,De11,Steinke12} it was shown that the INW generator fools ordered permutation branching programs with seed length $O(\poly(w) \cdot \log n\cdot \log(1/\eps)),$ which is $O(\log n)$ for constant $w$ and $\eps$.  Finally, in \cite{HozaPyVa21}, it was shown that the INW generator fools ordered permutation branching programs that have a single accept state with seed length $O(\log n \cdot\log\log n + \log n \cdot \log(1/\eps))$, with no dependence on the width $w$.  In \cite{PyneVa21-WPRG,BogdanovHoPrPy22,ChenLyTaWu23}, the INW generator, with these improved analyses, was also used to construct relaxations of pseudorandom generators (hitting-set generators and weighted pseudorandom generators) for ordered regular and/or permutation branching programs that have an even better dependence on the error parameter $\eps$.

The key to these improved analyses is to show that the error of the INW generator accumulates more slowly for these models of branching programs than given by Theorem~\ref{thm:INW}, for example achieving $\eps=O(\sigma\cdot \log n)$ yields the result of \cite{HozaPyVa21}.  
The error analysis of \cite{HozaPyVa21} builds on \cite{RozenmanVa05} in viewing the composition of the INW generator with an ordered branching program as the result of an iterated graph operation.  Note that if $B$ is an ordered permutation branching program of length $n$ and $G : \zo^d\rightarrow \zo^n$ is any generator, then composing $B$ and $G$ can be viewed as defining a $2^d$-regular bipartite multigraph $B\circ G$ with $w$ vertices on each side, where we connect left-vertex $s\in [w]$ to the final state reached when we run $B$ on each of the outputs of $G$ from start state $s$. The recursive operation (\ref{eqn:INW-recursion}) defining the INW generator amounts to taking a ``product'' of the two bipartite graphs $B_L \circ \INW_{i-1}$ and $B_R\circ \INW_{i-1}$, where $B_L$ and $B_R$ are the first and second halves of a program $B$ of length $2^i$. If the graph $H_i$ is the complete graph with self-loops, then this is a standard graph product operation, where the edges are obtained by first following an edge in $B_L \circ \INW_{i-1}$ and then following an independent edge in $B_R\circ \INW_{i-1}$.  (If the left half and right half are identical, this is simply graph squaring.) When $H_i$ is a sparse expander, then this is a ``derandomized product'' operation that has a similar spirit to the zig-zag product of Wigderson and collaborators~\cite{ReingoldVaWi01}.  Analyzing this repeated derandomized product using notions of approximation from spectral graph theory~\cite{AhmadinejadKeMuPeSiVa20} yields an improved analysis of the INW generator.

\subsubsection{Constant-depth Circuits and Iterated Restrictions}

The first computational model that was studied for unconditional derandomization, in the seminal paper of Ajtai and Wigderson~\cite{AjtaiWi89}, was constant-depth polynomial-size boolean circuits with unbounded fan-in AND and OR gates, also known as $\ACO$. They gave an unconditional construction of a pseudorandom generator with seed length $O(n^\eps)$ fooling $\ACO$, for any constant $\eps>0$. This was improved by Nisan~\cite{Nisan91-AC} to seed length $\polylog(n)$, using a construction that inspired the Nisan--Wigderson generator described in Section~\ref{sec:NW}. (Indeed, Nisan's generator is the special case of the Nisan--Wigderson generator where the hard function $f$ is the parity function.)  Ajtai and Wigderson~\cite{AjtaiWi89} pointed out a compelling algorithmic application of pseudorandom generators for $\ACO$, namely to derandomize the Karp--Luby $\BPP$ algorithm for approximately counting the number of satisfying assignments to a DNF formula (i.e. a depth 2 $\ACO$ circuit that is an OR of ANDs of literals)~\cite{KarpLuMa89}.  There are now nearly polynomial-time deterministic algorithms for this problem, all of which use pseudorandom generators along with other algorithmic techniques~\cite{NisanWi94,LubyVeWi93,LubyVe96}.  

Although Nisan~\cite{Nisan91-AC} dramatically improved upon the seed length of the Ajtai--Wigderson generator, the approach taken by Ajtai and Wigderson---{\em iterated pseudorandom restrictions}---has undergone a revival over the past decade.  The idea of iterated pseudorandom restrictions is to not try to generate all $n$ pseudorandom bits at once, but to use a short seed to select and assign values to a smaller fraction of the bits.  If we use a seed of length $d_0$ to assign a $p$ fraction of the bits, then by iterating, we can use a seed of length $O(d_0\cdot (\log n)/p)$ to assign all the bits. The benefit of this approach is that when analyzing the pseudorandomness of the $pn$ bits generated in each iteration, we can think of the remaining $(1-p)n$ bits as being chosen uniformly at random.  Thus, fooling a test $T : \zo^n\rightarrow \zo$ reduces to fooling a random restriction $\rho$ of $T$ where we select $(1-p)n$ coordinates to restrict pseudorandomly but assign their values uniformly at random.  For many computational models (in particular constant-depth circuits), random restrictions cause substantial simplification, making the restricted function $T|_\rho$ easier to fool. 

Over the past decade, iterated pseudorandom restrictions and variants have been used to obtain improved pseudorandom generators for a variety of computational models. One example is the model of {\em combinatorial rectangles}, which test membership in a set of the form $R_1\times R_2\times \cdots \times R_n \subseteq [m]^n$, which can be viewed as a special case of both ordered branching programs and $\ACO$ formulas.
For this model, Wigderson and collaborators gave the first pseudorandom generator whose seed length is logarithmic in $m$ and $n$ for a subconstant error parameter $\eps$~\cite{ArmoniSaWiZh96}.  The iterated restriction approach of Ajtai and Wigderson was used in \cite{GopalanMeReTrVa12} to achieve a seed length that is nearly logarithmic in all the parameters, i.e. $\tO(\log(mn/\eps))$.  Since then, variants of the iterated restrictions approach have been used to obtain improved generators for constant-depth circuits, {\em arbitrary-order} read-once branching programs, De Morgan formulas, and various restricted versions of these models.  The number of works is too large to list here, so we refer the reader to the excellent survey of Hatami and Hoza~\cite{HatamiHo23}.

\section{Computational Complexity Lower Bounds} \label{sec:lowerbounds}

Proving lower bounds for the resources needed to perform computational tasks, in different computational models, is among the most challenging and most important topics in theoretical computer science.
Let us start by quoting the starting paragraph of Wigderson's recently-published monumental book, {\it Mathematics and Computation: A Theory Revolutionizing Technology and Science}~\cite{wigderson19}:
\begin{quote}
Here is just one tip of the iceberg we'll explore in this book: How much time does it take to find the
prime factors of a 1,000-digit integer? The facts are that (1) we can't even roughly estimate the answer: it
could be less than a second or more than a million years, and (2) practically all electronic commerce and
Internet security systems in existence today rest on the belief that it takes more than a million years!
\end{quote}

This paragraph says it all. While computers have revolutionized our world, the resources required to perform computational tasks are poorly understood.
Developing a mathematical
theory of computation is crucial in our information age, where computers are involved in
essentially every part of our life.

Computational complexity, the study of the amount of resources needed to perform computational tasks, is essential for understanding the power of computation and for developing a theory of computation.
It is also essential in designing efficient communication protocols, secure cryptographic protocols and in understanding human and machine learning.

We present here some of Wigderson's works on computational complexity theory, focusing on computational complexity lower bounds.
We will see that often these works introduced powerful techniques that had substantial impact and many followup works.

\subsection{Boolean Circuit Complexity} \label{sec:circuit-complexity}

Boolean circuits are the standard computational model for computing Boolean functions $f: \{0,1\}^n \rightarrow \{0,1\}$.
Given a Boolean function $f: \{0,1\}^n \rightarrow \{0,1\}$, we ask how many Boolean operations are needed to compute $f$. As the set of allowed Boolean operations, we consider here the set of Boolean logical gates $\{\wedge, \vee, \neg\}$ (also known as De Morgan basis). 

Given $n$ input variables
$x_1,\ldots,x_n \in \{0,1\}$,
a Boolean circuit is a directed acyclic graph as follows:
All nodes are of in-degree 0 or 2.
A node of in-degree 0 (that is, a {\em leaf})
is labelled with either an input variable $x_i$
or its negation $\neg x_i$.
A node of in-degree 2 is labelled with either $\wedge$ or $\vee$
(in the first case the
node is an $\mathsf{AND}$ gate and in the second case an
$\mathsf{OR}$ gate).
A node of out-degree 0
is called an {\em output} node. 
The circuit is called a {\em formula} if the underlying
graph is a (directed) tree.

Each node in the circuit (and in particular each output node)
computes a Boolean function from $\{0,1\}^n$ to $\{0,1\}$ as follows.
A leaf just computes the value of the input variable
or negation of input variable
that labels it. For every non-leaf node~$v$,
if $v$ is an $\mathsf{AND}$ gate it computes the $\mathsf{AND}$ of the functions computed by its two children, and if
$v$ is an $\mathsf{OR}$ gate it computes the $\mathsf{OR}$ of the functions computed by
its two children.
If the circuit has only one output node, the function computed by the circuit is the function computed by the output node.

A Boolean circuit is {\em monotone} if it doesn’t use negation gates. Each node in a monotone Boolean circuit (and in particular each output node) computes a monotone Boolean function from $\{0,1\}^n$ to $\{0,1\}$.

The {\em size} of a circuit is defined to be the number of nodes
in it and the {\em depth} of a circuit is defined to be the length of the longest directed path from a leaf to an
output node in the circuit.
For a circuit $C$, we denote its size by $\mathsf{S}(C)$ and its depth by $\mathsf{D}(C)$.
For a Boolean function $f$, we denote by $\mathsf{S}(f)$ the size of the smallest Boolean circuit for $f$, usually referred to as the circuit size of $f$, and by $\mathsf{D}(f)$ the smallest depth of a Boolean circuit for $f$, usually referred to as the circuit depth of $f$.
For a monotone Boolean function $f$, we refer to the size of the smallest monotone Boolean circuit for $f$, as the monotone circuit size of $f$, and  to the smallest depth of a monotone Boolean circuit for $f$, as the monotone circuit depth of $f$.

We note that often the unbounded-fanin case is also considered, where the in-degree of a node is not limited to be 0 or 2. For example, this is convenient when studying constant-depth circuits. In these cases, the size of the circuit is usually defined as the number of edges in it, rather than the number of nodes.

Proving lower bounds for the size and depth of Boolean circuits
has been a major challenge for many years.
In particular, the biggest challenge is to prove super-polynomial
lower bounds for the size of Boolean circuits and formulas,
for some explicit function. Such bounds would imply lower bounds for essentially all other models of computation. For example, super-polynomial (in $n$) lower bounds on the size of (a family of) circuits that compute a family of functions $\{f_n: \{0,1\}^n \rightarrow \{0,1\}\}_{n \in \mathrm{N}}$ would imply  that that family of functions is not in the complexity class $\mathsf{P}$ (polynomial time). If in addition the family of functions is in $\mathsf{NP}$ (non-deterministic polynomial time), such a result would imply that $\mathsf{P} \neq \mathsf{NP}$.

However, progress on this type of questions has been very limited. The best known lower bounds for the size of Boolean circuits, for an explicit function, are only linear in~$n$~\cite{LR,IM}, and the best known lower bounds for the depth of Boolean circuits, for an explicit function, are only logarithmic in~$n$.

\subsection{Communication Complexity}

Communication complexity, first introduced by Yao~\cite{Yao}, is a central model in complexity theory that studies the amount of communication needed to solve a problem, when the input to the problem is distributed between two (or more) parties. 

In the two-player deterministic model, each of two players gets an input, where the two inputs $x,y$ are chosen from some set of possibilities (known to both players). The players' goal is to solve a communication task that depends on both inputs, such as computing a function $f(x,y)$, where $f:\{0,1\}^n \times \{0,1\}^n \rightarrow \{0,1\}$ is known to both players and $x,y$ are inputs of length $n$ bits.

The players communicate in rounds, where in each round one of the players sends a message to the other player. At the end of the protocol, in the example given above, both players need to know the value of $f(x,y)$. 

The communication complexity of a protocol is the maximal number of bits communicated by the players in the protocol, where the maximum is taken over all possibilities for the inputs. The communication complexity of a communication task is the minimal communication complexity of a protocol that solves
that task.
For a communication protocol $P$, we denote its communication complexity by $\mathsf{CC}(P)$. 
For a communication task $G$, we denote by $\mathsf{CC}(G)$ the smallest communication complexity of a (deterministic) protocol that solves $G$.
The probabilistic case, where the players are allowed to  use a public random string and are allowed to err with some fixed small probability smaller than $\tfrac{1}{2}$ is often studied as well.
For a communication task $G$, we denote by $\mathsf{CC}_{\epsilon}(G)$ the smallest communication complexity of a (probabilistic) protocol that solves $G$ correctly with probability at least $1-\epsilon$ on every input.

As an example, we give the problem of Set-Intersection, or Set-Disjointness, a central problem in communication complexity. In this problem, each of two players gets a vector in $\{0,1\}^n$ and their goal is to determine whether there exists a coordinate $i \in [n]$ where they both have~$1$.
This simple problem inspired a lot of progress in communication complexity.
It has been known for a long time that the probabilistic communication complexity of Set-Intersection is $\Omega(n)$~\cite{KS,Razb,BYKS,BM, BGPW}. The lower bound is trivially tight, up to the multiplicative constant.

\subsection{Karchmer-Wigderson Games}

Karchmer and Wigderson gave a striking connection between the depth of Boolean circuits and communication complexity. They showed that for every Boolean function $f: \{0,1\}^n \rightarrow \{0,1\}$, there is a simple and intuitive communication complexity game $G_f$, such that, the smallest depth of a Boolean circuit for $f$ is exactly equal to the deterministic communication complexity of $G_f$. Moreover, if $f$ is monotone, there is also a communication complexity game $M_f$, such that, the smallest depth of a monotone Boolean circuit for $f$ (that is, a Boolean circuit for $f$ that doesn't use negations) is exactly equal to the deterministic communication complexity of~$M_f$. In particular, this reduces the problem of proving lower bounds for the depth of Boolean circuits, a problem that seems hard to understand or analyze, to a problem in communication complexity that seems much more intuitive and easier to work with~\cite{KW}. 

\begin{definition}\label{def-KW}~\cite{KW} {\bf (KW Games, $\pmb{G_f}$):}
For every function $f: \{0, 1\}^n \rightarrow \{0, 1\}$,
define the communication game $G_f$ as follows:
Player~1 gets $x \in \{0, 1\}^n$, such that, $f(x) = 1$.
Player~2 gets $y \in \{0, 1\}^n$, such that, $f(y) = 0$.
The goal of the two players is to find a coordinate $i \in [n]$, such that, $x_i \neq y_i$ (note that there is at least one such $i$ since $f(x) \neq f(y)$).
\end{definition}

\begin{definition}\label{mon-KW1}~\cite{KW}  {\bf (KW Games, $\pmb{M_f}$):}
For every monotone function $f: \{0, 1\}^n \rightarrow \{0, 1\}$,
define the communication game $M_f$ as follows:
Player~1 gets $x \in \{0, 1\}^n$, such that, $f(x) = 1$.
Player~2 gets $y \in \{0, 1\}^n$, such that, $f(y) = 0$.
The goal of the two players is to find a coordinate $i \in [n]$, such that, $x_i =1$ and $y_i = 0$ (note that there is at least one such $i$ since $f(x) > f(y)$, and hence since $f$ is monotone, $x \not \leq y$).
\end{definition}

Recall that we denote deterministic communication complexity by $\mathsf{CC}$  and circuit depth by $\mathsf{D}$. In particular,
for a function $f: \{0, 1\}^n \rightarrow \{0, 1\}$, we denote by $\mathsf{D}(f)$ the smallest depth of a Boolean circuit for $f$. We denote by $\mathsf{CC}(G_f)$ the deterministic communication complexity of the game~$G_f$, and if $f$ is monotone, we denote by $\mathsf{CC}(M_f)$ the deterministic communication complexity of the game~$M_f$. 

\begin{theorem}\label{KW}~\cite{KW}
For every $f: \{0, 1\}^n \rightarrow \{0, 1\}$,
$\mathsf{CC}(G_f) = \mathsf{D}(f) $.
\end{theorem}

\begin{proof}
Let $z_1,\ldots,z_n \in \{0,1\}$ be the $n$ input variables for $f$ and recall that we denote by $x,y$ the inputs for the game $G_f$. 

{\bf Proving $\pmb{\mathsf{CC}(G_f) \leq \mathsf{D}(f)}$: } 
Let $C$ be any Boolean circuit for $f$. We will construct a communication protocol for the game~$G_f$, with communication complexity $\mathsf{D}(C)$. The construction is by induction on~ $\mathsf{D}(C)$.

{\bf Base case:}
$\mathsf{D}(C) = 0$. In this case, $f(z_1,\ldots,z_n)$ is simply the function $z_i$ or $\neg z_i$, for some $i$.
Therefore, there is no need for communication,
since $i$ is a coordinate in which $x$ and $y$
always differ.
That is, the two players can give the answer $i$, for any input pair
$(x,y)$. This is a protocol for $G_f$, with communication complexity $0$.

{\bf Induction step:}
Consider the top gate of $C$.
Assume first that the top gate is an $\mathsf{AND}$ gate and hence $C = C_1 \wedge C_2$, where $C_1,C_2$ are the two sub-circuits representing the two children of the top gate of $C$.
Thus,
$
\mathsf{D}(C_1), \mathsf{D}(C_2) \leq \mathsf{D}(C) - 1.
$
Denote by $f_1$ and~$f_2$ the functions computed by $C_1$ and $C_2$
respectively. Thus $f = f_1 \wedge f_2$.
By the inductive hypothesis,
$
\mathsf{CC}(G_{f_1}), \mathsf{CC}(G_{f_2}) \leq \mathsf{D}(C) - 1
$.
We know that $f(x) = 1$ and $f(y) =0$. Therefore, we know that
$f_1(x),f_2(x)$ are both equal to $1$ and at least one of $f_1(y)$ or $f_2(y)$ is equal to $0$.
Let us present the protocol for $G_f$.
In the first step of the protocol,
Player~2 sends a value in $\{1,2\}$, indicating which of the functions
$f_1$ or $f_2$ is equal to $0$ on $y$ (or an arbitrary value in $\{1,2\}$ if both are equal to $0$).
Assume that Player~2 sends $1$.
In this case, we have $f_1(x) = 1$ and
$f_1(y) = 0$.
Hence, to solve the game $G_f$, the players can apply a protocol for $G_{f_1}$.
By the inductive
hypothesis, there is such a protocol
with communication complexity
$\mathsf{CC}(G_{f_1}) \leq \mathsf{D}(C)-1$.
In the same way,
if Player~2 sends $2$ the players can use the protocol for $G_{f_2}$.
The players used only one additional bit of communication.
Hence, we can conclude that
$$\mathsf{CC}(G_f)  \leq  1 + \max\{\mathsf{CC}(G_{f_1}), \mathsf{CC}(G_{f_2})\}  
 \leq  1 + (\mathsf{D}(C) - 1)  =  \mathsf{D}(C).$$
We assumed that $C = C_1 \wedge C_2$. The other case,
$C = C_1 \vee C_2$, is proved in the same way,
except that Player~1 is the one who sends the first
bit, indicating whether $f_1(x) = 1$ or $f_2(x) = 1$.

Since the construction is valid for every circuit $C$ for $f$, and in particular for the one with smallest depth, we can conclude that $\mathsf{CC}(G_f) \leq \mathsf{D}(f)$.

{\bf Proving $\pmb{\mathsf{CC}(G_f) \geq \mathsf{D}(f)}$: }
For this proof, we define a more general communication game. For any two
disjoint sets: $A, B \subseteq \{0, 1\}^n$,
denote by $G_{A, B}$ the following game:
Player~1 gets $x \in A$.
Player~2 gets $y \in B$.
The goal of the two players is to find a coordinate $i$, such that, $x_i \neq y_i$.
Note that  $G_f$ is the same as $G_{f^{-1}(1), f^{-1}(0)}$.

We will prove the following claim: 
If $\mathsf{CC}(G_{A, B}) = d$ then there
is a function $g: \{0, 1\}^n \rightarrow \{0, 1\}$,
such that:
$g(x) = 1$, for every $x \in A$;
$g(y) = 0$, for every $y \in B$; and 
$\mathsf{D}(g) \leq d$.
That is, the function $g$ separates $A$ from $B$, and $\mathsf{D}(g) \leq d$.
Note that
for the game $G_f = G_{f^{-1}(1), f^{-1}(0)}$,
the function $g$ must be the function $f$ itself.
Hence, we obtain that
$\mathsf{D}(f) \leq \mathsf{CC}(G_f)$, as required.
The proof of the claim is
by induction on $d = \mathsf{CC}(G_{A, B})$.

{\bf Base case:}
$d = 0$. That is, the two players know the answer without any communication.
Hence, there is a coordinate $i$, such that, for every $x \in A$
and every $y \in B$, we have $x_i \neq y_i$.
Thus, either
the function $g(z) = z_i$ or the function
$g(z) = \neg z_i$ satisfies the requirements of the claim
(depending on whether for every $x \in A$ we have $x_i=1$, or,
for every $x \in A$ we have $x_i=0$).

{\bf Induction step:}
We have a protocol of communication complexity $d$ for the game~$G_{A, B}$.
Assume first that Player~1 sends the first bit in the protocol.
That bit partitions the set $A$ into two disjoint sets
$A = A_0 \cup A_1$ (where $A_0$ is the set of all inputs $x$ where Player~1 sends 0 and $A_1$ is the set of all inputs $x$ where Player~1 sends~1).
If the first bit sent by Player~1 is 0, the rest of the protocol is a protocol
for the game~$G_{A_0, B}$.
If the first bit sent by Player~1 is 1, the rest of the protocol is a protocol
for the game~$G_{A_1, B}$.
Hence, for both games, $G_{A_0, B}$ and $G_{A_1, B}$,
we have protocols with communication complexity
at most $d - 1$.
By the inductive hypothesis, we have two functions $g_0$ and $g_1$
that satisfy:
$g_0(x) = 1$, for every $x \in A_0$;
$g_1(x) = 1$, for every $x \in A_1$; 
$g_0(y) = g_1(y) = 0$, for every $y \in B$; and
$\mathsf{D}(g_0), \mathsf{D}(g_1) \leq d - 1$.
We define $g = g_0 \vee g_1$. Thus:
For every $x \in A$, we have
$g(x) = g_0(x) \vee g_1(x) = 1$;
For every $y \in B$, we have
$g(y) = g_0(y) \vee g_1(y) = 0$; and
$\mathsf{D}(g) \leq 1 + \max\{\mathsf{D}(g_0), \mathsf{D}(g_1)\} \leq d$.
That is, $g$ satisfies the requirements.

If Player~2 sends the first bit,
$B$ is partitioned into two disjoint sets, $B = B_0 \cup B_1$, and
as before, the rest of the protocol is a protocol for the
games $G_{A, B_0}$ and $G_{A, B_1}$ (depending on the bit that was sent). By the
inductive hypothesis, we
have two functions, $g_0,g_1$, corresponding to the two games,
$G_{A, B_0}$ and $G_{A, B_1}$, such that:
$g_0(x) = g_1(x) = 1$, for every  $x \in A$;
$g_0(y) =0$, for every  $y \in B_0$;
$g_1(y) = 0$, for every  $y \in B_1$.
We define $g = g_0 \wedge g_1$.
Thus:
For every $x \in A$, we have
$g(x) = g_0(x) \wedge g_1(x) = 1$;
For every $y \in B$, we have
$g(y) = g_0(y) \wedge g_1(y) = 0$; and
$\mathsf{D}(g) \leq 1 + \max\{\mathsf{D}(g_0), \mathsf{D}(g_1)\} \leq d$.
\end{proof}

For a monotone Boolean function $f: \{0, 1\}^n \rightarrow \{0, 1\}$, denote by $\mathsf{MD}(f)$ the smallest depth of a monotone Boolean circuit for $f$.

\begin{theorem}\label{mon-KW}~\cite{KW}
For every monotone $f: \{0, 1\}^n \rightarrow \{0, 1\}$,
$\mathsf{CC}(M_f) = \mathsf{MD}(f) $.
\end{theorem}

\begin{proof}
Similar to the proof of Theorem~\ref{KW}.
\end{proof}

\noindent
{\bf Example ($\pmb{k}$-Clique):}
Take the Boolean function $f$ to be the $(n/2)$-Clique
function in simple graphs with $n$ vertices. That is, the input for $f$ is a simple graph with $n$ vertices and the output is 1 if and only if the graph contains a clique of size at least~$n/2$.
The games $G_f$ and $M_f$ are defined as follows: In both games,
Player~1 gets a graph $x$ (with~$n$ vertices) that contains
a clique of size at least $n/2$ and Player~2
gets a graph $y$ (with~$n$ vertices) that doesn't contain
a clique of size at least $n/2$. 
The goal of the two players in the game $M_f$ is to find
an edge in the graph  $x$ that is not an edge  in the graph $y$.
The goal of the two players in the game $G_f$ is to find
an edge in the graph~$x$ that is not an edge  in the graph~$y$ or an edge in the graph~$y$ that is not an edge  in the graph~$x$. 

\noindent
Theorem~\ref{KW} shows that the communication
complexity of the game $G_f$ is exactly equal to
the circuit depth of the $(n/2)$-Clique
function. In particular, one can try to prove a lower bound for
the circuit depth of the $(n/2)$-Clique function, by proving a lower bound for
the communication complexity of the game $G_f$.
Note that no lower bound better than
$\Omega(\log n)$ has ever been proved for the circuit depth of an explicit
Boolean function and such a bound would be a major breakthrough.

\noindent
Theorem~\ref{mon-KW} shows that the communication
complexity of the game $M_f$ is exactly equal to
the monotone circuit depth of the $(n/2)$-Clique
function. Moreover, it turned out that one can use this connection to prove a lower bound for
the monotone circuit depth of the $(n/2)$-Clique function, by proving a lower bound for
the communication complexity of the game~$M_f$~
\cite{RW}.
\vspace{4mm}

We will now present an alternative equivalent way to define the game $M_f$, in terms of the minterms and maxterms of the monotone Boolean function $f$.
Every monotone Boolean function can be characterized by the set of its minterms and the set of its maxterms.
\begin{definition} {\bf (Minterm, Maxterm):}
Let $f:\{0, 1\}^n \rightarrow \{0, 1\}$ be a monotone Boolean function.
A {\em minterm} of $f$ is an input $x \in \{0, 1\}^n$, such that, $f(x) = 1$ and for every input $x' < x$, we have $f(x') = 0$.
A {\em maxterm} of $f$ is an input $y \in \{0, 1\}^n$, such that, $f(y) = 0$ and for every input $y' > y$, we have
$f(y') = 1$.
\end{definition}

\begin{definition}\label{mon-KW2}~\cite{KW} {\bf (KW Games, $\pmb{M_f}$):}
For every monotone function $f: \{0, 1\}^n \rightarrow \{0, 1\}$,
define the communication game $M_f$ as follows:
Player~1 gets $x \in \{0, 1\}^n$, such that, $x$ is a minterm of $f$.
Player~2 gets $y \in \{0, 1\}^n$, such that, $y$ is a maxterm of $f$.
The goal of the two players is to find a coordinate $i \in [n]$, such that, $x_i =1$ and $y_i = 0$ (note that there is at least one such $i$ since $f(x) > f(y)$, and hence since $f$ is monotone, $x \not \leq y$).
\end{definition}

We have defined the game $M_f$ in two different ways, once in Definition~\ref{mon-KW1} and once in Definition~\ref{mon-KW2}. While the two definitions do not give the exact same game, the two games are equivalent, so we denote both of them by $M_f$. 
To see the equivalence, let $M^{'}_f$ be the game from Definition~\ref{mon-KW1} and let $M^{''}_f$ be the game from Definition~\ref{mon-KW2}. First, note that $M^{''}_f$ is a restriction of the game $M^{'}_f$ to a subset of inputs, so any protocol for $M^{'}_f$ is also a protocol for $M^{''}_f$. On the other hand, the players can use a protocol for $M^{''}_f$ to solve $M^{'}_f$ as follows: Given an input $x$ such that $f(x)=1$, Player~1 can find a minterm $x'$ of $f$ such that $x' \leq x$. In the same way, given an input $y$ such that $f(y)=0$, Player~2 can find a maxterm $y'$ of $f$ such that $y' \geq y$. The players can now apply the protocol for $M^{''}_f$ on inputs $x',y'$ to find a coordinate $i$ such that $x'_i =1$ and $y'_i=0$. Since $x_i \geq x'_i$ and  $y_i \leq y'_i$, we also have $x_i =1$ and $y_i=0$.

\vspace{4mm}
\noindent
{\bf Example ($\pmb{s}$-$\pmb{t}$-Connectivity):}
Take the Boolean function $f$ to be the $s$-$t$-Connectivity
function in simple graphs with $n$ vertices. That is, the input for $f$ is a simple graph with $n$ vertices, two of which are labeled as $s$ and $t$, and the output is 1 if and only if the graph contains a path connecting $s$ and $t$. Obviously, $f$ is a monotone function, since adding edges cannot disconnect
an existing path from $s$ to $t$.

\noindent
A minterm of $f$ is a graph that contains a path from $s$ to $t$, and no additional edges. That is, a minterm is just a path from $s$ to $t$ (that does not intersect itself).
A maxterm of $f$ is a graph $G$, such that,
the set of vertices of $G$ can be partitioned into two disjoint
sets $S$ and $T$, with
$s \in S$ and $t \in T$ and such that
$G$ contains all edges inside $S$ and inside $T$, but no edge between $S$ and $T$. We think of a maxterm as a partition of the set of vertices into two sets ($S$ and $T$), or as a two-coloring of the vertices by the colors $0$ and $1$ (where $S$ is colored $0$ and $T$ is colored  $1$).

\noindent
The game $M_f$ is defined as follows: Given $n$ vertices, two of which are labeled by $s$ and $t$,
Player~1 gets a path from $s$ to $t$ and Player~2 gets
a coloring of the $n$ vertices by the colors $\{0,1\}$,
such that, $s$ is colored 0 and $t$ is colored 1.
The goal of the two players is to find an edge $(u,v)$ on the path,
such that, $u$ is colored 0 and $v$ is colored~1 (or vice versa).

\noindent
Theorem~\ref{mon-KW} shows that the communication
complexity of the game $M_f$ is exactly equal to
the monotone circuit depth of the $s$-$t$-Connectivity
function. Moreover, it turned out that one can use this connection to prove a lower bound for
the monotone circuit depth of the $s$-$t$-Connectivity function, by proving a lower bound for
the communication complexity of the game~$M_f$~
\cite{KW}.
\vspace{4mm}

Unlike the case of general Boolean circuits, where progress in proving lower bounds for explicit Boolean functions has been very limited, there has been a long and very successful line of works that establish strong
lower bounds for the monotone circuit size and for the monotone circuit depth
of many explicit functions, starting from Razborov's celebrated  super-polynomial lower bounds for the size of monotone Boolean circuits~\cite{Razb1, Razb2}. 

Since their introduction, KW games have had a huge impact on the study of monotone circuit depth and beyond, and have been further studied in numerous works.
Already in their original paper, Karchmer and Wigderson used KW games to prove a tight lower bound of $\Omega(\log^2 n)$ for the monotone circuit depth of the $\pmb{s}$-$\pmb{t}$-Connectivity function in graphs with $n$ vertices~\cite{KW}. In particular, this result gave the first super-polynomial separation between monotone circuit size and monotone formula size and separated the monotone versions of the complexity classes $\mathsf{NC_1}$ and $\mathsf{NC_2}$. We present this result in Section~\ref{sec:KW-LB}.

Raz and Wigderson used KW games to prove tight lower bounds of $\Omega(n)$ for the monotone circuit depth of the clique and matching functions in graphs with $n$ vertices~\cite{RW}. We present this result in Section~\ref{sec:RW-LB}.

Karchmer, Raz and Wigderson used KW games to outline an approach for proving super-logarithmic lower bounds for the depth of general Boolean circuits~\cite{KRW}. We present this result in Section~\ref{sec:KRW-LB}.

Raz and McKenzie used KW games to separate the monotone versions of the complexity classes $\mathsf{NC}$ and $\mathsf{P}$, as well as  $\mathsf{NC_i}$ and $\mathsf{NC_{i+1}}$ for every $i$~\cite{RM}. That paper also introduced a general technique for proving lower bounds for communication complexity, a technique that was later on named by G\"{o}\"{o}s, Pitassi and Watson, the {\em lifting method}.  G\"{o}\"{o}s, Pitassi and Watson initiated the study of the lifting method as a general technique for proving separation results in communication complexity~\cite{GPW}, followed by a long line of recent works.

For a long time, KW games have been used mainly to study circuit depth, rather than circuit size. Nevertheless, a recent paper by Garg, G{\"{o}}{\"{o}}s, Kamath and Sokolov shows how to use (an extension of) KW games to prove lower bounds for  monotone circuit size~\cite{GGKS}, using Razborov's DAG-like communication protocols~\cite{Razb3}.

\subsection{Lower Bounds for the Monotone Depth of ST-Connectivity} \label{sec:KW-LB}

We will now present Karchmer and Wigderson's proof that any monotone circuit for the $s$-$t$-Connectivity
function in graphs with $n$ vertices is of depth $\Omega(\log^2 n)$~\cite{KW}.
We deviate from Karchmer and Wigderson's original presentation in various places.%

Recall that the game $M_f$ for the $s$-$t$-Connectivity
function  is defined as follows: Given $n$ vertices, two of which are labeled by $s$ and~$t$,
Player~1 gets a path from $s$ to~$t$ and Player~2 gets
a coloring of the $n$ vertices by the colors $\{0,1\}$,
such that, $s$ is colored 0 and $t$ is colored 1.
The goal of the two players is to find an edge $(u,v)$ on the path,
such that, $u$ is colored 0 and $v$ is colored~1 (or vice versa). By Theorem~\ref{mon-KW}, the communication
complexity of this game is exactly equal to
the monotone circuit depth of the $s$-$t$-Connectivity
function.

It is helpful to first see an upper bound for the communication complexity of the game.
A simple protocol for this game is as follows:
In the first round,
Player~1  sends the name (number) of the middle vertex in the
path and Player~2 replies with its
color. If the
color of the middle vertex is 0 then the players
continue with the second half of the path, and if
the color is 1 then the players continue with the first
half of the path.
The players continue to perform a binary search,
until they are
left with a path of length 1. This path will be an edge $(u,v)$,
such that, $u$ is colored 0 and $v$ is colored 1.
In each round of the protocol, the players
communicate $O(\log n)$ bits (the number of the vertex and its color).
Since in each step the path is shortened by a factor of 2, the number of rounds
will be $O(\log n)$.
Altogether, the communication complexity of the protocol is
$O(\log^2 n)$.
Hence, by Theorem~\ref{mon-KW}, the monotone circuit depth
of $s$-$t$-Connectivity is $O(\log^2n)$. Next, we present the lower bound.

\begin{theorem}\label{KW-LB}~\cite{KW}
The monotone circuit depth
of $s$-$t$-Connectivity is $\Omega(\log^2n)$.
\end{theorem}

\begin{proof} 
For the proof of the lower bound, we will modify the communication game $M_f$ for the $s$-$t$-Connectivity
function, and present a variant of the game that we refer to as $STCON(\ell, n)$. In this game, there are two parameters, $n$ and $\ell \leq n^{0.1}$. We assume without loss of generality that $\ell$ is a power of~2. We have $\ell \cdot n$ vertices arranged in $\ell$ layers with $n$ vertices in each layer, and two additional vertices $s$ and~$t$. We assume that the layers are numbered $(1,\ldots,\ell)$ and the vertices in each layer are numbered $(1,\ldots,n)$. Player~1 gets a path of length $\ell + 1$ from $s$ to $t$ that passes through each of the $\ell$ layers exactly once, in their order. That is, the path starts at~$s$, goes to a vertex in the first layer then a vertex in the second layer and so on, and finally goes from the last layer to the vertex $t$. Such a path can be presented as $x \in [n]^{\ell}$, specifying the number of the vertex that the path reaches in each layer. 
Player~2 gets
a coloring of the $\ell \cdot n +2$ vertices by the colors $\{0,1\}$,
such that, $s$ is colored 0 and $t$ is colored 1. Such a coloring can be presented as  $y \in \{0,1\}^{\ell \cdot n}$, specifying the color of each vertex in each layer.
The goal of the two players is to find an edge $(u,v)$ on the path,
such that, $u$ is colored 0 and $v$ is colored~1 (or vice versa).

We will show a lower bound of $\Omega(\log \ell \cdot \log n)$ for the communication complexity of  $STCON(\ell, n)$. 
Since $STCON(\ell, n)$ is a restriction to a subset of inputs of the game $M_f$ (for the $s$-$t$-Connectivity
function with $\ell \cdot n+2$ vertices), such a bound implies a lower bound of $\Omega(\log^2 n)$ for the
monotone circuit depth of the $s$-$t$-Connectivity
function in graphs with $n$ vertices. 
Next, we give the proof for 
$$\mathsf{CC}(STCON(\ell, n)) = \Omega(\log \ell  \cdot \log n).$$

Let $X = [n]^{\ell}$ and  
$Y = \{0,1\}^{\ell \cdot n}$. For a subset $A\subseteq X$, we define its density as $\alpha =\tfrac{|A|}{|X|}$ and for a subset $B\subseteq Y$, we define its density as 
$\beta = \tfrac{|B|}{|Y|}$.
Recall that in the game $STCON(\ell, n)$, the input for Player~1 is viewed as $x \in X$
and the input for Player~2 is viewed as $y \in Y$.

We will consider restrictions of the game $STCON(\ell, n)$ to subsets of inputs $A\subseteq X$ and $B\subseteq Y$ and define the game $STCON(\ell, n, A, B)$ to be the same as $STCON(\ell, n)$, except that the input for Player~1 is $x \in A$
and the input for Player~2 is $y \in B$.
We denote by $\mathsf{C}(\ell, n, \alpha,\beta)$ the minimal communication complexity of a game $STCON(\ell, n, A, B)$ with a set $A \subseteq X$ of density $\alpha$ and a set $B \subseteq Y$ of density $\beta$.

Fixing $n$ to be a (sufficiently large) integer, and fixing $t \overset{\operatorname{def}}{=} \tfrac{1}{2n^{0.1}}$, we will show that for every $\ell \leq n^{0.1}$, every $\alpha \geq t $ and  $\beta \geq 0$,
$$\mathsf{C}(\ell, n, \alpha,\beta) \geq 
c \cdot \log \ell  \cdot \log n + \log(\alpha) + \log(\beta),$$
where $c > 0$ is a (sufficiently small universal) constant and the logarithm is base 2.
Hence,
$\mathsf{CC}(STCON(\ell, n)) = \mathsf{C}(\ell, n, 1,1) = \Omega(\log \ell  \cdot \log n)$.

The proof for $\mathsf{C}(\ell, n, \alpha,\beta) \geq 
c \cdot \log \ell  \cdot \log n + \log(\alpha) + \log(\beta)$ is by induction over $\ell,\alpha,\beta$, in this order (and note that since $n$ is fixed there is a finite number of possibilities for $\ell,\alpha,\beta$, so the induction is sound). We will consider two cases: 
$\alpha \geq 2t $ and $2t > \alpha \geq t $.

{\bf Case I: $\pmb{\alpha \geq 2t}$: }
Let $A \subseteq X$ be a subset of density $\alpha$ and $B \subseteq Y$ be a subset of density $\beta$.
Consider any protocol $P$ for the game $STCON(\ell, n, A, B)$ and let $d$ be the communication complexity of the protocol. Since $\alpha \geq 2t$, none of the edges of the path $x$ is fixed and hence $d>0$.
We will prove that $$d \geq 
c \cdot \log \ell  \cdot \log n + \log(\alpha) + \log(\beta).$$

Assume first that Player~1 sends the first bit in the protocol $P$.
That bit partitions the set $A$ into two disjoint sets
$A = A_0 \cup A_1$ (where $A_0$ is the set of all inputs $x$ where Player~1 sends 0 and $A_1$ is the set of all inputs $x$ where Player~1 sends~1).
If the first bit sent by Player~1 is 0, the rest of the protocol is a protocol
for the game~$STCON(\ell, n, A_0, B)$.
If the first bit sent by Player~1 is 1, the rest of the protocol is a protocol
for the game~$STCON(\ell, n, A_1, B)$.
Hence, for both games, $STCON(\ell, n, A_0, B)$ and $STCON(\ell, n, A_1, B)$,
we have protocols with communication complexity
at most $d - 1$. Let $\alpha_0$ be the density of $A_0$ and $\alpha_1$ be the density of $A_1$. Note that $\alpha_0 + \alpha_1 = \alpha$ and hence at least one of $\alpha_0, \alpha_1$ is larger than or equal to $\alpha /2$.
Hence $\mathsf{C}(\ell, n, \alpha/2,\beta) \leq d-1$. 
By the inductive hypothesis,
$d -1 \geq 
c \cdot \log \ell  \cdot \log n + \log(\alpha/2) + \log(\beta) $, that is
$$d \geq 
c \cdot \log \ell  \cdot \log n + \log(\alpha) + \log(\beta).$$
The case where Player~2 sends the first bit in the protocol $P$ is similar.

{\bf Case II: $\pmb{2t > \alpha \geq t }$: }
Let $A \subseteq X$ be a subset of density $\alpha$ and $B \subseteq Y$ be a subset of density $\beta$.
Consider any protocol $P$ for the game $STCON(\ell, n, A, B)$ and let $d$ be the communication complexity of the protocol. 
We will prove that $$d \geq 
c \cdot \log \ell  \cdot \log n + \log(\alpha) + \log(\beta).$$
Note that we can assume without loss of generality that 
$\log(\beta) \geq - \log^2n$, as otherwise the right hand side of the inequality is smaller than 0 (if $c < 1$).

Every path $x \in A$ can be written as $x=(x_L,x_R)$, where $x_L \in [n]^{\ell/2}$ is the left-hand half of the path $x$ (the first $\ell/2$ coordinates of $x$) and $x_R \in [n]^{\ell/2}$ is the right-hand half of the path $x$ (the last $\ell/2$ coordinates of $x$). We say that $x_L \in [n]^{\ell/2}$ is significant if there exist at least $\tfrac{\alpha}{4}\cdot n^{\ell/2}$ extensions $x_R \in [n]^{\ell/2}$, such that, $(x_L,x_R)\in A$. Let $A_L \subseteq [n]^{\ell/2}$ be the set of significant paths $x_L$. We say that $x_R \in [n]^{\ell/2}$ is significant if there exist at least $\tfrac{\alpha}{4}\cdot n^{\ell/2}$ extensions $x_L \in [n]^{\ell/2}$, such that, $(x_L,x_R)\in A$. Let $A_R \subseteq [n]^{\ell/2}$ be the set of significant paths $x_R$.
Let $\alpha_L$ be the density of $A_L$ in $[n]^{\ell/2}$, that is, $\alpha_L =\tfrac{|A_L|}{n^{\ell/2}}$. Let $\alpha_R$ be the density of $A_R$ in $[n]^{\ell/2}$, that is, $\alpha_R =\tfrac{|A_R|}{n^{\ell/2}}$. Since for every $(x_L,x_R)\in A$, either $x_L$ is not significant or $x_R$ is not significant or both are significant, we have 
$$\alpha \cdot n^{\ell} =|A| \leq n^{\ell/2} \cdot \tfrac{\alpha}{4} \cdot n^{\ell/2} + \tfrac{\alpha}{4} \cdot n^{\ell/2} \cdot n^{\ell/2} + \alpha_L \cdot n^{\ell/2}\cdot \alpha_R \cdot n^{\ell/2},$$
that is 
$$\tfrac{\alpha}{2} \leq \alpha_L \cdot \alpha_R.$$
Thus, $\alpha_L \geq \sqrt{\tfrac{\alpha}{2}}$ or $\alpha_R \geq \sqrt{\tfrac{\alpha}{2}}$. Without loss of generality 
$$\alpha_L \geq \sqrt{\tfrac{\alpha}{2}}.$$

Every coloring $y \in B$ can be written as $y=(y_L,y_R)$, 
where $y_L \in \{0,1\}^{(\ell/2) \cdot n}$ is the left-hand half of the coloring $y$ (the first $(\ell/2) \cdot n$ coordinates of $y$) and $y_R \in \{0,1\}^{(\ell/2) \cdot n}$ is the right-hand half of the coloring $y$ (the last $(\ell/2) \cdot n$ coordinates of $y$). We say that $y_L \in \{0,1\}^{(\ell/2) \cdot n}$ is significant if there exist at least $\tfrac{\beta}{2}\cdot 2^{\ell \cdot n /2}$ extensions $y_R \in \{0,1\}^{(\ell/2) \cdot n}$, such that, $(y_L,y_R)\in B$. Let $B_L \subseteq \{0,1\}^{(\ell/2) \cdot n}$ be the set of significant colorings $y_L$. 
Let $\beta_L$ be the density of $B_L$ in $\{0,1\}^{(\ell/2) \cdot n}$, that is, $\beta_L =\tfrac{|B_L|}{2^{\ell \cdot n /2}}$. Since for every $(y_L,y_R)\in B$, either $y_L$ is not significant or significant, we have 
$$\beta \cdot 2^{\ell \cdot n } =|B| \leq 2^{\ell \cdot n /2} \cdot \tfrac{\beta}{2} \cdot 2^{\ell \cdot n /2}  + \beta_L \cdot 2^{\ell \cdot n /2} \cdot 2^{\ell \cdot n /2},$$
that is,
$$ \beta_L \geq \tfrac{\beta}{2} .$$

Let $T$ be a subset of vertices (to be determined later) in the last $\ell / 2$ layers, that is, layers $\tfrac{\ell}{2}+1,\ldots,\ell$.
Given $T$, we define the set $A'_L \subseteq A_L$ to be the set of all $x_L \in A_L$ such that there exists an extension $x_R \in [n]^{\ell/2}$, such that, $(x_L,x_R)\in A$ and all vertices of the path $x_R$ are in $T$.
Given $T$, we define the set $B'_L \subseteq B_L$ to be the set of all $y_L \in B_L$ such that there exists an extension $y_R \in \{0,1\}^{(\ell/2) \cdot n}$, such that, $(y_L,y_R)\in B$ and all vertices of $T$ are colored 1 by the coloring $y_R$.

We claim that the protocol $P$ can be used to solve the communication game $STCON(\ell/2, n, A'_L, B'_L)$ and hence $\mathsf{CC}(STCON(\ell/2, n, A'_L, B'_L)) \leq d$. This can be done as follows. 
Given $T$ and an input $x_L \in A'_L$, Player~1 finds an extension $x_R \in [n]^{\ell/2}$, such that, $(x_L,x_R)\in A$ and all vertices of the path $x_R$ are in $T$.
Given $T$ and an input $y_L \in B'_L$, Player~2 finds an extension $y_R \in \{0,1\}^{(\ell/2) \cdot n}$, such that, $(y_L,y_R)\in B$ and all vertices of $T$ are colored 1 by the coloring $y_R$. The players run the protocol $P$ on inputs $x=(x_L,x_R)$, $y=(y_L,y_R)$. Since all vertices of the path $x_R$ are in $T$, they are all colored 1 by the coloring $y_R$. Hence, the edge $(u,v)$ returned by the communication protocol $P$ must satisfy $u=s$ or $u$ is in the first $\ell / 2$ layers, that is, layers $1,\ldots,\tfrac{\ell}{2}$. Thus, it is a valid answer for the game $STCON(\ell/2, n)$ on inputs $x_L,y_L$.

It remains to show that there exists a subset of vertices $T$, as above, such that $A'_L, B'_L$ are large, say $|A'_L| \geq  |A_L| / 2$ and $|B'_L| \geq  |B_L| / 2$. Assume first that there exists such a set $T$. Then, since $\alpha_L \geq  \tfrac{\sqrt{\alpha}}{2}$ and $\beta_L \geq \tfrac{\beta}{2} $, we get
$$
\mathsf{C}\bigl(\tfrac{\ell}{2}, n, \tfrac{\sqrt{\alpha}}{4},\tfrac{\beta}{4}\bigr) \leq 
\mathsf{CC}\bigl(STCON(\ell/2, n, A'_L, B'_L)\bigr) \leq d.
$$
Hence, by the inductive hypothesis,
\begin{align*}
d 
&\geq 
c \cdot \log \left( \tfrac{\ell}{2} \right)  \cdot \log n 
+ \log\left(\tfrac{\sqrt{\alpha}}{4}\right) + \log\left(\tfrac{\beta}{4}\right)
\\
&=
c \cdot \log \ell  \cdot \log n - c \cdot \log n
+ \tfrac{1}{2} \cdot \log\left(\alpha \right) + \log\left(\beta \right) -4
\\
&=
c \cdot \log \ell  \cdot \log n  + \log\left(\alpha \right) + \log\left(\beta \right)
- c \cdot \log n
- \tfrac{1}{2} \cdot \log\left(\alpha \right)  - 4
\\
&\geq
c \cdot \log \ell  \cdot \log n  + \log\left(\alpha \right) + \log\left(\beta \right),
\end{align*}
where the last inequality is because 
$- c \cdot \log n
- \tfrac{1}{2} \cdot \log\left(\alpha \right)  - 4 \geq 0$, which is true for a sufficiently small constant $c$, since $\alpha < 2t = \tfrac{1}{n^{0.1}}$ (by the premise of Case II) and since $n$ is sufficiently large.

Thus, it remains to argue that there exists a subset of vertices $T$, as above, such that $|A'_L| \geq  |A_L| / 2$ and $|B'_L| \geq  |B_L| / 2$.
Recall that
$T$ is a subset of vertices in the last $\ell / 2$ layers, that is, layers $\tfrac{\ell}{2}+1,\ldots,\ell$.
We will choose $T$ randomly as follows. Take $n^{0.2}$ random paths $x_R \in [n]^{\ell/2}$ (viewed as paths on the last $\ell / 2$ layers) and let $T$ be the union of the sets of vertices on all these paths. Equivalently, the restriction of $T$ to each layer (from the last $\ell / 2$ layers), is generated by taking $n^{0.2}$ random vertices (with repetitions). Note also that $T$ can be extended to a set $T' \supset T$ of size, say, 
$2n^{0.2} \cdot (\ell / 2)$, such that the distribution of $T'$ is exponentially close to the distribution of a random set of size $2n^{0.2} \cdot (\ell / 2)$ of vertices in the last $\ell / 2$ layers.

Recall that the set $A'_L \subseteq A_L$ is the set of all $x_L \in A_L$ such that there exists an extension $x_R \in [n]^{\ell/2}$, such that, $(x_L,x_R)\in A$ and all vertices of the path $x_R$ are in~$T$. Recall that for every $x_L \in A_L$ there exist at least $\tfrac{\alpha}{4}\cdot n^{\ell/2}$ extensions $x_R \in [n]^{\ell/2}$, such that, $(x_L,x_R)\in A$. Therefore, since $\tfrac{\alpha}{4} \geq \tfrac{1}{8n^{0.1}}$, for each $x_L \in A_L$ with probability exponentially close to 1, one of these extensions was chosen among the $n^{0.2}$ random paths that were chosen to generate $T$. Thus, with probability very close to 1 almost every $x_L \in A_L$ is also in $A'_L$ and in particular $|A'_L| \geq  |A_L| / 2$.

Recall that the set $B'_L \subseteq B_L$ is the set of all $y_L \in B_L$ such that there exists an extension $y_R \in \{0,1\}^{(\ell/2) \cdot n}$, such that, $(y_L,y_R)\in B$ and all vertices of $T$ are colored~1 by the coloring $y_R$. Recall that for every $y_L \in B_L$
there exist at least $\tfrac{\beta}{2}\cdot 2^{\ell \cdot n /2}$ extensions $y_R \in \{0,1\}^{(\ell/2) \cdot n}$, such that, $(y_L,y_R)\in B$, and recall that we assumed (without loss of generality) that $\beta \geq 2^{- \log^2n}$.
For each extension $y_R$, we consider the set $T_{y_R}$ of all vertices (in the last $\ell / 2$ layers) that $y_R$ colors~1. For every $y_L \in B_L$, we consider the family of sets
$F_{y_L} = \{T_{y_R}\}_{ y_R: (y_L,y_R)\in B}$.
Thus, for every $y_L \in B_L$, we have that 
$|F_{y_L}| \geq 2^{- \log^2n -1}\cdot 2^{\ell \cdot n /2}$. 
By Kruskal–Katona theorem~\cite{kruskal63, Katona87,Lovasz93} (or alternatively by information theoretic arguments), such a large family of sets is guaranteed to contain,  with probability close to~1, a set $T_{y_R}$ that contains the random set $T'$  (where the probability is over the choice of $T'$). Note that if $T_{y_R}$ contains $T'$, the coloring $y_R$ colors all vertices in $T'$ (and hence all vertices in $T$) by~1.
Thus, with probability close to~1 almost every $y_L \in B_L$ is also in $B'_L$ and in particular $|B'_L| \geq  |B_L| / 2$.
\end{proof}

\subsection{Lower Bounds for the Monotone Depth of Clique and Matching} \label{sec:RW-LB}

Next, we present Raz and Wigderson's lower bound of $\Omega(n)$ for the monotone circuit depth of the clique and matching functions in graphs with $n$ vertices~\cite{RW}. The proof establishes a lower bound of $\Omega(n)$ for the communication complexity of the corresponding KW games, by a direct reduction to known lower bounds in communication complexity, namely the lower bound of $\Omega(n)$ for the probabilistic communication complexity of Set-Disjointness~\cite{KS,Razb,BYKS,BM, BGPW}. This, in turn, further demonstrates the power of KW games, as well as the power of reductions from Set-Disjointness as a major tool for proving lower bounds in communication complexity and other computational models.

Recall that in the problem of Set-Intersection, or Set-Disjointness, each of two players gets a vector in $\{0,1\}^n$ and their goal is to determine whether there exists a coordinate $i \in [n]$ where they both have~$1$.
Recall that
for a communication task $G$, we denote by $\mathsf{CC}_{\epsilon}(G)$ the smallest communication complexity of a probabilistic protocol that solves $G$ correctly with probability at least $1-\epsilon$ on every input.
\begin{theorem}\label{thm:LB-Disjointness}~\cite{KS}  For any constant $\epsilon >0$, 
$\mathsf{CC}_{\epsilon}(Disjointness) \geq \Omega(n)$.
\end{theorem}

We will consider the following communication game, denoted $M_1$:
\begin{definition}
{\bf (Communication game $M_1$):}
Let $n=3k$ and let $V$ be a set of $n$ vertices.
Player~1 gets a $k$-matching $x$ on (a subset of) the set of vertices~$V$, that is,
$k$ edges (with vertices in $V$) that don't touch each other.
Player~2 gets a set~$y$ of $k-1$ vertices in $V$. 
The goal of the two players is to find an edge in $x$ that does not touch any 
of the vertices in $y$. (By the pigeonhole principle there must be at least one such edge). 
\end{definition}
We will prove that the deterministic communication complexity of 
$M_1$ is $\Omega(n)$,
$$\mathsf{CC}(M_1) \geq \Omega(n).$$

This bound implies lower bounds for the monotone circuit depth of several functions. We give a few examples:

\begin{theorem}\label{Thm:LB-Match}~\cite{RW} 
Let $n=3k$.
Let $Match$ be 
the (monotone) Boolean function that gets as an input a graph with $n$ vertices
and outputs 1 if and only if the graph contains a $k$-matching
(and outputs 0 otherwise). 
The monotone circuit depth of $Match$ is $\Omega(n)$.
\end{theorem}
\begin{proof}
Consider an input $(x,y)$ for the game $M_1$.
The $k$-matching $x$ is a minterm of the function $Match$.
The set $y$ of $k-1$ vertices can be viewed as a maxterm of
the function $Match$, by considering a graph that contains all possible edges with at least one vertex in $y$.
Any protocol $P$ for the monotone KW game of the function $Match$
can be applied on $(x,y)$ to get an edge in $x$ that doesn't touch $y$.
That is, any protocol
$P$ for the monotone KW game of the function $Match$ can be applied also as a protocol for $M_1$. Since $\mathsf{CC}(M_1) \geq \Omega(n)$, the communication complexity of $P$ is~$\Omega(n)$.
Hence, by Theorem~\ref{mon-KW}, the monotone circuit depth of $Match$ is $\Omega(n)$.
\end{proof}

\begin{theorem}\label{Thm:LB-PM}~\cite{RW} 
Let $PM$ be 
the (monotone) Boolean function that gets as an input a graph with $n$ vertices
and outputs 1 if and only if the graph contains a perfect matching
(and outputs 0 otherwise). 
The monotone circuit depth of $PM$ is $\Omega(n)$.
\end{theorem}

\begin{proof}
Follows by a standard reduction from $Match$ to $PM$: Given an input graph $Z$ for the function $Match$, where the number of vertices in $Z$ is $n=3k$, construct a graph $Z'$ by adding $k$ vertices to $Z$ and connecting them to all other vertices. Then, there exists a perfect matching in $Z'$ if and only if there exists a matching of size $k$ in $Z$.
\end{proof}

\begin{theorem}\label{Thm:LB-Clique}~\cite{RW} 
Let $n=3k$.
Let $Clique$ be 
the (monotone) Boolean function that gets as an input a graph with $n$ vertices
and outputs 1 if and only if the graph contains a clique of size $2k+1$
(and outputs 0 otherwise). 
The monotone circuit depth of $Clique$ is $\Omega(n)$.
\end{theorem}
\begin{proof}
Consider an input $(x,y)$ for the game $M_1$.
Given the set $y$ of $k-1$ vertices, consider the graph $y'$ that contains all edges that do not touch $y$ (that is, a clique in the complement of $y$). Since $y'$ is a clique of size $2k+1$, the function $Clique$ outputs 1 on $y'$.
Given the $k$-matching $x$, let the graph $x'$ be the complement of $x$, that is, the graph that contains all edges except the matching $x$. 
The function $Clique$ outputs 0 on $x'$.
Any protocol $P$ for the monotone KW game of the function $Clique$
can be applied on $(y',x')$ to get an edge in $y'$ that is not an edge in $x'$, that is an edge of $x$ that doesn't touch $y$.
Thus, any protocol
$P$ for the monotone KW game of the function $Clique$ can be applied also as a protocol for $M_1$. Since $\mathsf{CC}(M_1) \geq \Omega(n)$, the communication complexity of $P$ is~$\Omega(n)$.
Hence, by Theorem~\ref{mon-KW}, the monotone circuit depth of $Clique$ is $\Omega(n)$.
\end{proof}

Using similar arguments, one can establish lower bounds for the monotone depth of several other functions, such as, matching and perfect matching in bipartite graphs and clique functions with different sizes of cliques. 

It remains to prove the lower bound for the deterministic communication complexity of $M_1$.
\begin{theorem}\label{thm:RWmain}~\cite{RW}
$\mathsf{CC}(M_1) \geq \Omega(n).$
\end{theorem}

\begin{proof}
Let $n=3k$ and let $V$ be a set of $n$ vertices.
Consider the following communication game, denoted $M_2$:
Player~1 gets a $k$-matching $x$ on (a subset of) the set of vertices~$V$.
Player~2 gets a set~$y$ of $k$ vertices in $V$. 
The goal of the two players is to output 1 
if there is an edge in $x$ that does not touch any of the 
vertices in $y$, and output 0 otherwise, that is, if every edge in $x$ touches a vertex in $y$.

We will first prove that
for any constant $\epsilon > 0$,
\begin{equation} \label{eq:bd1}
\mathsf{CC}(M_1) \geq \Omega(\mathsf{CC}_{\epsilon}(M_2)).
\end{equation}

Assume that we have a deterministic communication protocol $P_1$ for the communication game
$M_1$. We will use $P_1$ to construct a probabilistic communication protocol $P_2$ for the communication game $M_2$,
with the same communication complexity as $P_1$ (up to an additive constant).

First note that, using a common random string,  we can assume that the protocol $P_1$ is a zero-error probabilistic protocol, such that, for every input $(x,y)$ for the game $M_1$, the protocol $P_1$ outputs each correct answer with the exact same probability
(that is, if for the input $(x,y)$ there are several correct answers
the protocol outputs each of them with the same probability).
This can be assumed, since, using the common random string, the players can 
randomly permute the vertices in $V$ before applying the protocol $P_1$.

Let $(x,y)$ be an input for the game $M_2$.
Player~2 gets the set $y$ of $k$ vertices, and will 
randomly choose a vertex $v \in y$ and remove it.
Now, Player~2 is left with a set $y'$ of $k-1$ vertices.
The two players can now apply the protocol $P_1$ (for $M_1$) on the input
$(x,y')$ and obtain as an output an edge $e \in x$ that doesn't 
touch any of the vertices in $y'$.  The players now check if the removed vertex $v$ is on the edge $e$.
If the vertex $v$ is not on the edge $e$, the protocol $P_2$ (for $M_2$) will outputs 1
(as $e$ is an edge that doesn't touch any vertex in $y$). 
If the vertex $v$ is on the edge $e$, the 
protocol $P_2$ will output 0 (that is, $P_2$ assumes that there is no edge
in $x$ that doesn't touch $y$, as such an edge was not found by $P_1$).

Note that if $P_2$ outputs 1
there can be no error (as $e$ does not touch  $v$  
or any  other vertex in $y$, 
since the protocol $P_1$ is always correct).
On the other hand, if $P_2$ outputs~0, an error is possible, as there might be a different edge $e'$ in $x$ that doesn't touch any 
of the vertices in $y$, and yet 
the protocol $P_1$ outputs the edge $e$ that does touch $v$. 
However, since the edges do not touch each other, there is at most 
one edge $e\in x$ that touches $v$.
Since we assume that the 
protocol $P_1$ outputs each possible correct answer with the exact same 
probability, and since an error occurs only if $e$ was the output of $P_1$
(and not any of the possible edges $e'$), 
the probability for an error is at 
most $1/2$ (for any input $(x,y)$). (The probability of error may be smaller if there are several edges $e'$ that do not touch any vertex in $y$).

To further reduce the probability of error to any constant $\epsilon$, one can
repeat the protocol $P_2$ a constant number of times.
This concludes the proof for Equation~\ref{eq:bd1}.

Next, we consider the following communication complexity game, denoted $3Dist$  (3-Distinctness): Let $n=3k$. Player~1 and Player~2 get inputs
$x,y \in \{a,b,c\}^k$, respectively. That is, each player gets a string 
of $k=n/3$ letters from $\{a,b,c\}$. The goal is to decide  
whether there is a coordinate $i$, such that $x_i = y_i$.

We will prove that
for any constant $\epsilon > 0$,
\begin{equation} \label{eq:bd2}
\mathsf{CC}_{\epsilon}(M_2) \geq \mathsf{CC}_{\epsilon}(3Dist).
\end{equation}

Assume that we have a probabilistic communication protocol $P_2$ for the communication game
$M_2$. We will use $P_2$ to construct a probabilistic communication protocol $P_3$ for the communication game $3Dist$,
with the same communication complexity and the same error as $P_2$. 

Let $(x,y)$ be an input for the game $3Dist$. Thus $x,y \in \{a,b,c\}^k$.
For each coordinate in $\{1,\ldots,k\}$, we construct a triangle (with different vertices for each triangle) 
and label its 3 vertices by $a,b,c$.
We label each edge of each triangle by the letter that labels the vertex that
it does not touch (that is, the vertex opposite to it).

The players convert their inputs to inputs for the game $M_2$
in the following way:
Player~1 
interprets her $k$ coordinates as the
corresponding $k$ edges in the $k$ triangles (one edge for each coordinate).
That is, each $x_i$ is interpreted as the corresponding edge
in the $i^{th}$ triangle.
Denote the set of these edges by $x'$. 
Player~2 interprets her
$k$ coordinates as the corresponding $k$ vertices in the $k$ triangles. 
That is, each $y_i$ is interpreted as the corresponding vertex 
in the $i^{th}$ triangle.
Denote the set of these vertices by $y'$.
Obviously, there is an edge in $x'$ that doesn't touch $y'$
if and only if there is a coordinate $i$, such that, $x_i=y_i$.
Thus, the players can use the protocol $P_2$ on input $(x',y')$ and declare the answer. This gives a probabilistic communication protocol $P_3$ for $3Dist$ with the same communication complexity and the same error as $P_2$.
This concludes the proof for Equation~\ref{eq:bd2}.

Finally, we will prove that for any constant $\epsilon > 0$,
\begin{equation} \label{eq:bd3}
\mathsf{CC}_{\epsilon}(3Dist) \geq \Omega(n).
\end{equation}
This will follow by a reduction from the
Set-Disjointness problem and by the lower bound for the probabilistic communication complexity of Set-Disjointness (Theorem~\ref{thm:LB-Disjointness}).

Assume that we have a probabilistic communication protocol $P_3$ 
for $3Dist$. 
We will show how to use this protocol to solve the Set-Disjointness
problem. 
Given an  
input pair $(x,y)$ for the Set Disjointness problem, such that $x,y \in \{0,1\}^k$,
the two players will generate inputs $x',y'$ for $3Dist$
as follows.
To generate $x'$,
Player~1 starts from $x$ and translates 0 to $b$ and 1 to $a$.
That is, for every $i$, if $x_i=0$ then $x'_i=b$ and if 
$x_i=1$ then $x'_i=a$.
To generate $y'$,
Player~2 starts from $y$ and translates 0 to $c$ and 1 to $a$.
That is, for every $i$, if $y_i=0$ then $y'_i=c$ and if
$y_i=1$ then $y'_i=a$.
Obviously, $x_i=y_i =1$ if and only if $x'_i = y'_i$. Hence, the two players
can apply the protocol $P_3$  on $(x',y')$ and declare the answer. 
This gives a probabilistic communication protocol for Set-Disjointness, with the same communication complexity and the same error as $P_3$.
By Theorem~\ref{thm:LB-Disjointness}, the communication complexity of the protocol is $\Omega(n)$. This concludes the proof for Equation~\ref{eq:bd3}.

By Equation~\ref{eq:bd1}, Equation~\ref{eq:bd2} and Equation~\ref{eq:bd3}, we get $\mathsf{CC}(M_1) \geq \Omega(n)$.
\end{proof}

\subsection{KRW Conjecture} \label{sec:KRW-LB}

Karchmer, Raz and Wigderson suggested an approach for proving super-logarithmic lower bounds for general 
Boolean circuit depth~\cite{KRW}. We will briefly outline this approach here. 

Let $n$ be an integer and assume for simplicity that $\log \log n$ is also an integer (where the logarithm is base~2). Let $k = \log n$. Let $f: \{0,1\}^k \rightarrow \{0,1\}$ be a random Boolean function. Since it's not hard to prove (by a standard counting argument) that a random Boolean function has large circuit depth (with high probability), we can assume that, say, $\mathsf{D}(f) \geq \tfrac{k}{2}$ (where $\mathsf{D}$ denotes circuit depth).

For two Boolean functions, $h: \{0,1\}^r \rightarrow \{0,1\}$ and 
$g: \{0,1\}^m \rightarrow \{0,1\}$, define their composition 
$h \circ g: \{0,1\}^{rm} \rightarrow \{0,1\}$ by 
$$
h \circ g (\vec{x_1},\ldots,\vec{x_r}) = h(g(\vec{x_1}),\ldots,g(\vec{x_r})),
$$
where 
$\vec{x_1},\ldots,\vec{x_r} \in \{0,1\}^m$.
Define $f^{(d)}$ to be the composition of $f$ with itself $d$ times.

KRW conjectured that for a random function $f: \{0,1\}^k \rightarrow \{0,1\}$ and any function $g: \{0,1\}^m \rightarrow \{0,1\}$,
$$
\mathsf{D}(f\circ g) \geq
\epsilon \cdot \mathsf{D}(f) + \mathsf{D}(g)
$$
(with high probability over the choice of $f$),
for some constant $\epsilon > 0$.
There are also various variants of this conjecture.

Assuming that the conjecture holds, we get 
$$
\mathsf{D}(f^{(d)}) \geq \epsilon \cdot \mathsf{D}(f) \geq 
\tfrac{\epsilon}{2} \cdot d\cdot k,
$$
(with high probability over the choice of $f$). Taking $d= k / \log k$, we get a function $f^{(d)} : \{0,1\}^n \rightarrow \{0,1\}$ of super-logarithmic depth.
The function $f^{(d)}$ is not explicit, as $f$ is a random function, but since $f$ depends on only $k=\log n$ input variables, its truth table of size $n$ can be given as $n$ additional input variables, so that $f^{(d)}$ is explicitly given.

To prove the conjecture, KRW suggested to use Karchmer-Wigderson games. 
Given $f: \{0,1\}^k \rightarrow \{0,1\}$ and  $g: \{0,1\}^m \rightarrow \{0,1\}$, the KW game corresponding to $f\circ g$ is as follows: \\
Player~1 gets $\vec{x_1},\ldots,\vec{x_k} \in \{0,1\}^m$, such that,
$$
f(g(\vec{x_1}),\ldots,g(\vec{x_k})) =1.
$$
Player~2 gets $\vec{y_1},\ldots,\vec{y_k} \in \{0,1\}^m$, such that,
$$
f(g(\vec{y_1}),\ldots,g(\vec{y_k})) =0.
$$
The goal of the two players is to find $(i,j)$ such that
$\vec{x_{i j}} \neq \vec{y_{i j}}$.

To see the intuition behind the conjecture,  assume that the inputs $(\vec{x_1},\ldots,\vec{x_k})$ for Player~1 and 
$(\vec{y_1},\ldots,\vec{y_k})$ for Player~2 satisfy that for every 
$i \in \{1,\ldots,k\}$, if $g(\vec{x_i}) = g(\vec{y_i})$ then $\vec{x_i} = \vec{y_i}$. Then, an answer $(i, j)$ for the KW game corresponding  to $f\circ g$ gives an answer $i$ for the KW game corresponding to $f$ 
(with input 
$((g(\vec{x_1}),\ldots,g(\vec{x_k})), (g(\vec{y_1}),\ldots,g(\vec{y_k})))$) 
and an answer $j$ for an instance of the KW game corresponding to $g$ (namely, the KW game corresponding to the $i^{th}$ coordinate, that is, the game played with input $(\vec{x_i}, \vec{y_i})$).

Finally, we note that while the conjecture is still wide open and seems hard to prove, some steps towards proving the conjecture have been done in several papers, including the works by Edmonds, Impagliazzo, Rudich and Sgall~\cite{EIRS01}, H{\aa}stad and Wigderson~\cite{HW90}, Gavinsky, Meir, Weinstein and Wigderson~\cite{GMWW17}, Dinur and Meir~\cite{DM18}, Meir~\cite{Meir23} (to name a few).

\subsection{Communication Complexity of Set-Disjointness}

The Set-Disjointness problem that was already mentioned before is a central problem in communication complexity, with numerous applications.
We have already seen how strong lower bounds for the monotone depth of Boolean functions (Theorem~\ref{Thm:LB-Match}, Theorem~\ref{Thm:LB-PM} and Theorem~\ref{Thm:LB-Clique}) follow from known lower bounds for the probabilistic communication complexity of Set-Disjointness (Theorem~\ref{thm:LB-Disjointness}).
The Set-Disjointness problem can be described as follows (equivalently to our previous description).
Each of two players gets a subset of $[n]$ and their goal is to determine whether the two subsets intersect.

Since the Set-Disjointness problem is so central, and because of the many applications, it's very interesting to also study variants of this problem. 
H{\aa}stad and Wigderson~\cite{HW07} studied the perhaps  most natural variant of the problem, the Set-Disjointness problem with sets of a fixed size $k$.
That is, given $n$ and $k \leq n/2$, Player~1 gets a subset $x \subset [n]$ of size $k$, Player~2 gets a subset $y \subset [n]$ of size $k$, and their goal is to determine whether the two subsets $x,y$ intersect. We denote this communication game by $D^n_{k}$. In the deterministic case, it is not hard to prove that for every $k \leq n/2$,
$\mathsf{CC}(D^n_{k}) = \Theta(\log \binom{n}{k})$~\cite{HW07}.

H{\aa}stad and Wigderson proved that in the probabilistic case, the communication complexity of $D^n_{k}$ is in fact $O(k)$. (This bound is tight when $k < c n$ for any constant $c < 1/2$).

\begin{theorem}\label{thm:HW}~\cite{HW07}  
For any $n$ and $k \leq n/2$ and constant $\epsilon >0$, 
$\mathsf{CC}_{\epsilon}(D^n_{k}) = O(k).$
\end{theorem}

\begin{proof} {\bf (Sketch)}
Player~1 gets a subset $x \subset [n]$ of size $k$ and Player~2 gets a subset $y \subset [n]$ of size $k$.
The players run a communication protocol that, assuming that $x,y$ are disjoint, has communication complexity $O(k)$ and at the end of the protocol both players know (with high probability) two disjoint subsets $S,T \subset [n]$, such that, $x \subseteq S$ and $y \subseteq T$. The sets $S,T$ can be viewed as a proof for the disjointness of $x,y$. If after $ck$ bits of communication (when $c$ is a sufficiently large constant), 
the protocol fails, it follows that (with high probability) $x,y$ are not disjoint.

The protocol works in $O(\log k)$ steps. First define,
$$N_{0}=[n],\;\; S_{0}=\emptyset,\;\; T_{0}=\emptyset,\;\; x_{0}=x,\;\; y_{0}=y.$$
After each Step~$i$, the players will have subsets $N_{i},S_{i},T_{i},x_{i},y_{i} \subseteq [n]$, where 
$N_{i}\cup S_{i}\cup T_{i}$ is a partition of $[n]$ and 
$$x \cap T_{i} = \emptyset,\;\; y \cap S_{i} = \emptyset,\;\; x_i = x \cap N_i,\;\; y_i = y \cap N_i.$$ 
Moreover, $$S_{i-1} \subseteq S_i,\;\; T_{i-1} \subseteq T_i.$$
Intuitively, after each Step~$i$, the players have already restricted the possible intersection of $x$ and $y$ to the set $N_i$ and it remains to check if $x_i$ and $y_i$ intersect.

Each Step~$i$ is done as follows. Assume without loss of generality that 
$|x_{i-1}| \leq |y_{i-1}|$. (Otherwise we switch the rolls of the players in Step~$i$). Let $k_{i-1} = |x_{i-1}|$.
The players interpret the public random string as a sequence of random subsets 
$Z_1,Z_2,\ldots \subseteq N_{i-1}$. Player~1 examines the first $2^{ck_{i-1}}$ sets in this sequence (where $c$ is a sufficiently large constant) and sends the index $j$ of the first set $Z_j$, such that, $x_{i-1} \subseteq Z_j$ (if such a set exists). Since $c$ is sufficiently large, such a set $Z_j$ exists with high probability, as the probability that a random set $Z_j$  satisfies $x_{i-1} \subseteq Z_j$ is $2^{-k_{i-1}}$.
Since the random string is public, both players now know $Z_j$ and update
$$N_{i}=Z_j,\;\; S_{i}=S_{i-1},\;\; T_{i}= T_{i-1} \cup N_{i-1} \setminus Z_j,\;\; x_{i} = x_{i-1},\;\; y_{i} = y_{i-1} \cap Z_j.$$
Note that all the required properties from the sets $N_{i},S_{i},T_{i},x_{i},y_{i}$ are satisfied. That is, 
$N_{i}\cup S_{i}\cup T_{i}$ is a partition of $[n]$ and 
$$x \cap T_{i} = \emptyset,\;\; y \cap S_{i} = \emptyset,\;\; x_i = x \cap N_i,\;\; y_i = y \cap N_i,\;\; S_{i-1} \subseteq S_i,\;\; T_{i-1} \subseteq T_i.$$ 
Note also that if $x,y$ are disjoint then $|y_{i}|$ is equal to $|y_{i-1}|/2$ in expectation and as long as $|y_{i-1}|$ is sufficiently large (say, larger than a sufficiently large constant), $|y_{i}|\leq 0.6 \cdot |y_{i-1}|$ with high probability and hence $|y_{i}|+|x_{i}| \leq 0.8 \cdot (|y_{i-1}| + |x_{i-1}|)$ with high probability (where high probability here means 1 minus  probability exponentially small in $|y_{i-1}| + |x_{i-1}|$). 

If $x,y$ are disjoint then, after repeating this protocol for $O(\log k)$ steps, we get that with high probability the final $x_i,y_i$ are both empty. This is because the sum of their sizes keeps decreasing by a constant factor until it's smaller than a (sufficiently large) constant and then it keeps decreasing by at least 1, with a constant probability in each step. When $x_i,y_i$ are both empty, the protocol stops, and we have two disjoint subsets $S,T \subset [n]$, such that, $x \subseteq S$ and $y \subseteq T$ (where $S,T$ are the final $S_i,T_i$).
The communication complexity is $O(k)$, since in each step the communication complexity is $O(|y_{i-1}| + |x_{i-1}|)$ and thus converges to $O(k)$, as 
$O(|y_{i-1}| + |x_{i-1}|)$ keeps decreasing by a constant factor until it's constant.
\end{proof}

The protocol in the proof of Theorem~\ref{thm:HW} uses an exponential amount of randomness. Nevertheless, the amount of randomness can always be reduced to $O(\log n)$ by a general theorem of Newman~\cite{New91}.

\subsection{Quantum versus Classical Communication Complexity}

Buhrman, Cleve and Wigderson were the first to study communication complexity advantages of quantum communication protocols over classical ones~\cite{BCW98}. A quantum communication protocol is a protocol where the players can send quantum states, rather than just classical bits, and the communication complexity of the protocol is defined to be the total number of qubits sent by the protocol, that is, the sum of the lengths (in qubits) of all the quantum states that are sent by the protocol~\cite{Yao93}.

Buhrman, Cleve and Wigderson proved a general theorem that shows that any quantum algorithm with small query complexity implies quantum communication protocols for related problems, with small communication complexity. 
Given a (total or partial) function $f: \{0, 1\}^n \rightarrow \{0, 1\}$, one can define the following two communication complexity problems: Given two inputs, $x,y \in \{0, 1\}^n$, where Player~1 gets~$x$ and Player~2 gets~$y$, the goal of the two players is to compute 
$f(x \wedge y)$, or $f(x \oplus y)$ 
(where $x \wedge y$ and $x \oplus y$ denote a coordinate by coordinate application of $\wedge$ and $\oplus$). Buhrman, Cleve and Wigderson proved that if there is a quantum algorithm for computing the function $f(z)$, with $k$ quantum queries to the input
$z=(z_1,\ldots,z_n)$, then there are quantum communication complexity protocols for computing  $f(x \wedge y)$ and $f(x \oplus y)$, with communication complexity $O(k\cdot \log n)$~\cite{BCW98}.

This general theorem gives a method for translating quantum query complexity upper bounds into quantum communication complexity upper bounds, as well as translating quantum communication complexity lower bounds into quantum query complexity lower bounds. Using this general theorem, Buhrman, Cleve and Wigderson obtained interesting consequences in both directions. They used known lower bounds for quantum communication complexity to obtain new lower bounds for quantum query complexity. They also used the general theorem to establish an exponential separation between zero-error quantum communication complexity, and classical deterministic communication complexity, that is, they gave a communication task that can be solved by a zero-error quantum communication complexity protocol with small communication complexity, and such that any classical deterministic communication complexity protocol for that task, requires exponentially larger communication complexity.

Perhaps the most striking consequence of Buhrman, Cleve and Wigderson's general theorem is that it proves
that the Set-Disjointness problem can be solved by a quantum protocol with communication complexity $O(\sqrt{n} \cdot \log n)$~\cite{BCW98}. 

\begin{theorem}\label{thm:BCW}~\cite{BCW98} 
The quantum communication complexity of Set-Disjointness is $O(\sqrt{n} \cdot \log n)$.
\end{theorem}

The proof of Theorem~\ref{thm:BCW}
follows from the general theorem by using Grover's algorithm for computing the $\mathsf{OR}$ of $n$ input variables, using only $O(\sqrt{n})$ quantum queries to the input~\cite{Grover96}. 

Theorem~\ref{thm:BCW} stands in contrast to Theorem~\ref{thm:LB-Disjointness} that states that the classical probabilistic communication complexity of Set-Disjointness is $\Omega(n)$. Theorem~\ref{thm:BCW} established a quadratic gap between quantum and classical probabilistic communication complexity and was followed by a long line of works that further studied the relative power of quantum and classical communication protocols.
We note that this  quadratic separation remained essentially the largest known gap between quantum and classical probabilistic communication complexity of total functions, for almost two decades. A line of recent works improved that gap to an almost cubic gap~\cite{ABK16, ABBGJKLS16, Tal20, BS21, SSW21}. Proving a super-polynomial gap  between quantum and classical probabilistic communication complexity of total functions remains a fascinating and long-standing open problem in communication complexity. For partial functions (promise problems), exponential gaps between quantum and classical probabilistic communication complexity were established by a long line of works~\cite{Raz99, BYJK08, GKKRW08, KR11, Gavinsky20, GRT19}.
Finally, we note that Theorem~\ref{thm:BCW} was proved to be essentially tight by Razborov~\cite{Razborov_2003}.

\subsection{Partial Derivatives in Arithmetic Circuit Complexity} \label{sec:partialderivatives} 

Arithmetic circuits are the standard computational model for arithmetic computations, such as computing the determinant or the permanent of a matrix
or the product of two matrices.
Given a field $\mathbb{F}$ and an $n$-variate polynomial $P(x_1,\ldots,x_n)$ over~$\mathbb{F}$, 
we ask how many $+, \times$ operations over $\mathbb{F}$ are needed to compute $P$.

An arithmetic circuit over $\mathbb{F}$,
with input variables
$x_1,\ldots,x_n \in \mathbb{F}$,
is a directed acyclic graph as follows:
Every node of in-degree 0 (that is, a {\em leaf})
is labelled with either an input variable
or a field element or a product of an input variable and a field element.
Every node of in-degree larger than 0 is labelled with either $+$ or $\times$
(in the first case the node is a {\em sum } gate and in the second case a
{\em product} gate).
A node of out-degree 0 is called an {\em output} node. 
The circuit is called a {\em formula} if the underlying
graph is a (directed) tree.

Each node in the circuit (and in particular each output node)
computes a polynomial in the ring of polynomials $\mathbb{F}[x_1,\ldots,x_n]$ as follows. A leaf just computes the value of the input variable, or field element, or product of input variable and  field element, that labels it. For every non-leaf node $v$, if $v$ is a sum gate it computes the sum of the polynomials computed by its children, and if $v$ is a product gate it computes the product of the polynomials computed by its  children.
If the circuit has only one output node, the polynomial computed by the circuit is the polynomial computed by the output node.

The {\em size} of a circuit is defined to be the number of wires (edges)
in it and the {\em depth} of a circuit is defined to be the length of the longest directed path from a leaf to an
output node in the circuit.

Proving lower bounds for the size of arithmetic circuits
has been a major challenge for many years.
Super-linear lower bounds for the size of general arithmetic circuits were proven in the seminal works of Strassen~\cite{strassen73} and Baur and Strassen~\cite{BS83}.
Their method, however, only gives lower bounds of up to $\Omega (n\log d)$, where $n$ is the number of input variables and $d$ is the degree of the computed polynomial.
In particular, if the degree $d=d(n)$ is polynomial in $n$ this gives lower
bounds of at most $\Omega (n\log n)$.
Lower bounds for various restricted classes of arithmetic circuits have also been studied in many works.

In 1997, Nisan and Wigderson suggested a general approach for obtaining lower bounds for restricted classes of arithmetic circuits~\cite{NW97}.
The approache is based on measuring the dimension of the vector space
spanned by all partial derivatives of the polynomials computed at the nodes of the circuit. 
(Partial derivatives were previously used to obtain lower bounds for arithmetic circuits in the works of Smolensky~\cite{Smolensky90} and Nisan~\cite{Nisan91}).

For an $n$-variate polynomial $f(x_1,\ldots,x_n)$, let $D(f)$ denote the set of all partial derivatives, of all orders, of $f$ (including $f$ itself as the partial derivative of order~0), and let $\mathsf{Dim}(f)$ denote the
dimension of the vector space spanned by $D(f)$. 
The main idea is to bound the growth of $\mathsf{Dim}(f)$ from the leaves to the outputs of the circuit and hence show that for an output of the circuit, 
$\mathsf{Dim}(f)$ is bounded. Thus, the circuit cannot compute a polynomial $P(x_1,\ldots,x_n)$ with a larger $\mathsf{Dim}(P)$.
The following simple formulas are easily proved and are useful for bounding the growth of $\mathsf{Dim}(f)$,
\begin{align*}
\mathsf{Dim}(h+g) &\leq \mathsf{Dim}(h) + \mathsf{Dim}(g),\\
\mathsf{Dim}(h \times g) &\leq \mathsf{Dim}(h) \cdot \mathsf{Dim}(g).
\end{align*}

Nisan and Wigderson used this approach to prove several lower bounds, including exponential lower bounds for the  size of depth-3 homogeneous circuits (where an homogeneous circuit is a circuit where all nodes in the circuit compute homogeneous polynomials), and exponential lower bounds 
for the  size of constant-depth set-multilinear circuits (where a set-multilinear circuit is a circuit where the set of variables 
$\{x_1,\ldots,x_n\}$
is partitioned into $d$ subsets $X_1,\ldots,X_d$, such that, for every node $v$ in the circuit,  each monomial in the polynomial computed by the node $v$ contains at most one variable from each subset $X_i$)~\cite{NW97}.

The partial-derivatives method of Nisan and Wigderson has been very influential on later works. Many subsequent works used this approach as a starting point and further built on these ideas to obtain lower bounds for additional classes of arithmetic circuits. In particular, these ideas have been very important in the study of  multilinear circuits 
(for example,~\cite{RazS05, RazM09, RazM06, RY09, RSY08, DMPY12, ChillaraL019, AlonKV20}),
constant-depth homogeneous circuits (for example,~\cite{KumarS13b, KayalLSS17, KumarS17})
and bounded-depth arithmetic circuits (for example,~\cite{ShpilkaW01, Kayal12, GKKS14, FournierLMS15, Limaye0T21, AGK0T22}).

\subsection{Resolution Made Simple}

{\em Resolution} is a proof system (technically, refutation system) for refuting unsatisfiable CNF formulas, that is, unsatisfiable Boolean formulas in conjunctive normal forms.

Given Boolean variables
$x_1,\ldots,x_n \in \{0,1\}$, a {\em literal} is either a variable,
$x_i$, or a negation of a variable,  $\neg x_i$.
A clause in these variables is an $\mathsf{OR}$ of literals, that is, $\bigvee_{i=1}^k z_i$, for some $k$, where each $z_i$ is a literal.
The {\em Resolution rule} says that if $C$ and
$D$ are two clauses and $x_i$ is a variable then any assignment
that satisfies both clauses, $C \vee x_i$ and $D \vee \neg
x_i$, also satisfies the clause $C \vee D$. 
Thus, from $C \vee x_i$ and $D \vee \neg x_i$, one can deduce $C \vee D$.

A Resolution refutation for a set of clauses $F$
(equivalently, for a CNF formula~$F$) proves that the clauses in $F$ are not simultaneously satisfiable.
For a set of clauses~$F$,
a Resolution refutation 
is a sequence of clauses $C_1,C_2,\ldots,C_s$,
such that: (1) Each clause $C_j$ is either a clause in $F$ or obtained by the Resolution rule from two previous clauses in the sequence, and (2) The last
clause, $C_s$, is the empty clause (and is hence unsatisfiable).
The {\em size}, or {\em length}, of a Resolution refutation
is the number of clauses in it.

It is
well known that Resolution is a sound and complete propositional
proof system, that is, a CNF formula $F$ is unsatisfiable if and only if
there exists a Resolution refutation for $F$. We think of a
refutation for an unsatisfiable formula $F$ also as a proof for
the tautology $ \neg F$.
Hence, Resolution refutations are also called Resolution proofs.

Resolution is one of the most widely studied propositional
proof systems. 
Lower bounds for the size of Resolution proofs for 
many propositional tautologies have been proved, starting from 
Haken's celebrated exponential lower bounds for the propositional pigeonhole principle~\cite{Haken85}.

Ben-Sasson and Wigderson suggested a general approach for proving lower
bounds for the size of Resolution proofs, an approach that
generalized, unified and simplified essentially all previously known lower bounds for Resolution, was used to obtain many additional lower bounds, and
ultimately gave a deeper understanding of Resolution as a proof system~\cite{Ben-SassonW01}. 

The approach focuses on the {\em width} of a resolution proof. 
The width of a resolution proof is defined to be the number of literals in the largest clause of the proof. 
Ben-Sasson and Wigderson argued that Resolution is best studied when the focus is on the width. Their key theorem relates the smallest length of a Resolution proof to the smallest width of a Resolution proof. Informally, the theorem states that if a set of clauses $F$ has a short Resolution refutation then it also has a Resolution refutation with small width. The proof is based on a proof by Clegg, Edmonds and Imagliazzo, who gave similar relations (between size of a proof and degree of a proof) for algebraic proof systems~\cite{CleggEI96}.

\begin{theorem}\label{thm:BW}~\cite{Ben-SassonW01} 
Let $F$ be a an unsatisfiable CNF formula. Let $w_0$ be the size of the largest clause in~$F$. Let $w$ be the minimal width of a Resolution refutation for~$F$.
Let $s$ be the minimal size of a Resolution refutation for $F$. Then,
$$ w \leq w_0 + O\left(\sqrt{n \log s}\right).$$
\end{theorem}
In particular, Theorem~\ref{thm:BW} shows that one can obtain lower bounds for the size of Resolution proofs by proving lower bounds for the width of Resolution proofs (which, in many cases, is easier to analyze).

The size of the clauses of a Resolution proof was implicit in previous works and  played a major roll in previous lower bounds.
Previous lower bounds for the size of Resolution proofs were usually proved in two steps as follows. In the first step, the entire proof was hit by a random restriction of the variables (that is, some of the variables were randomly set to~0, some were randomly set to~1 and some were left untouched), in order to hit and eliminate all large clauses of the proof (assuming for a contradiction that the proof is short). The second step proved that large clauses must exist in any Resolution refutation for the restriction of the unsatisfiable formula under the random restriction from the first step (and hence the proof must be long).
The approach of Ben-Sasson and Wigderson simplified essentially all previous proofs, as the random restriction was no longer needed and one could focus on proving lower bounds on the width of Resolution refutations for the original unsatisfiable formula, rather than for a random restriction of it.

\newcommand{\defeq}{:=}
\renewcommand{\Cap}{{\mathrm{Cap}}}
\newcommand{\Tr}{\mathrm{Tr}}
\newcommand{\BL}{\mathrm{BL}}
\newcommand{\Sn}{{S}}

\section{Complexity, Optimization, and, Symmetries}\label{sec:optimization}

This section presents an overview of work by Wigderson and his co-authors on  optimization methods to come up with efficient algorithms for various algorithmic problems in computational complexity theory, mathematics, and physics \cite{LSW98,GargGOW16,GargGOW17,ALOW17,Allen-ZhuGLOW18, BurgisserGOWW18, BurgisserFGOWW18,BurgisserFGOWW19}.
A common theme in all these works is the realization that the relevant algorithmic tasks can be formulated as optimization problems over  algebraic  groups that also have an analytic structure.
A representative optimization problem is to find a minimum norm vector in the orbit of a given $\mathrm{GL}_n(\mathbb{C})$ action   on a vector space.
This viewpoint led Wigderson and his co-authors to deploy tools from invariant theory, representation theory, and optimization to develop a quantitative theory of optimization over Riemannian manifolds that arise from continuous symmetries of noncommutative groups.

The starting point is  the work  \cite{LSW98}  that analyzes the convergence of a matrix scaling algorithm to compute an approximation to the permanent (Section \ref{sec:permanent}).
This corresponds to the commutative setting where the symmetries corresponded to diagonal subgroups (tori) of a matrix group.
The role of symmetries in the analysis of the algorithm, however, was not quite explicit.

In \cite{Gurvits04}, Gurvits  extended the results of \cite{LSW98} to the noncommutative setting of ``operators''.
In particular, he studied Edmonds' singularity problem \cite{Edmonds} and, motivated by \cite{LSW98}, he presented a (deterministic)  ``operator scaling'' algorithm for it.
However, he fell short of presenting convergence bounds for this algorithm.
Section \ref{sec:operator} presents the work  \cite{GargGOW16} that gives convergence bounds for Gurvits' operator scaling algorithm.
This paper makes  the first contact of scaling algorithms to invariant theory.
It also demonstrates the applicability of computational  problems over group orbits and scaling techniques far beyond complexity theory: to mathematics and physics.
Section \ref{sec:ncit} presents a result from   \cite{GargGOW16} that gives a deterministic polynomial time algorithm for the noncommutative version of Edmonds' singularity problem. 
Section \ref{sec:BL}  gives an outline of a result from \cite{GargGOW17} that shows how operator scaling  can be used to efficiently compute  Brascamp-Lieb constants important in mathematics.

Section \ref{sec:geodesic} visits the paper \cite{Allen-ZhuGLOW18} which starts with the realization that the problem of finding a minimum-norm vector over an orbit  is a geodesically convex optimization problem over a Riemannian manifold.
Subsequently, \cite{Allen-ZhuGLOW18}  extend the theory of second-order methods in convex optimization to the setting of geodesically convex optimization and give an algorithm whose running time depends logarithmically on the error in the approximation.
The focus here is on introducing geodesic convexity and showing how the capacity of an operator can be captured by a geodesically convex optimization problem.

Finally,  Section \ref{sec:nullcone}, presents results from \cite{BurgisserFGOWW19}.
Here,  the general norm minimization problem is introduced and  various variants of it studied by \cite{BurgisserFGOWW19} are presented.
These problems unify and generalize  prior works in this line.
Of particular importance is the connection to noncommutative duality in invariant theory which extends linear programming duality and allows one to give conditions on when an optimization problem is feasible. 
This gives rise to other connections such as  moment maps (analog of Euclidean gradients) and a precise notion of geodesic convexity.
This paper culminates with the definition and convergence bounds for first-order and second-order algorithms for various  optimization problems over noncommutative matrix groups.
The convergence bounds are based on novel parameters related to the group action via a synthesis of algebra and analysis. 
This paper also gives a host of new analytic algorithms for various  problems important in invariant theory and complexity theory.

\subsection{Permanent and matrix scaling}\label{sec:permanent}

Let $A\in \R^{n\times n}$ be a square matrix with entries $A_{i,j}$ for $1\leq i,j \leq n$. The permanent of $A$ is defined as:
$$\mathrm{Per}(A)\defeq \sum_{\sigma \in S_n} \prod_{i=1}^n A_{i,\sigma(i)},$$
where $S_n$ is the set of all permutations over $n$ symbols, i.e., the set of bijections $\sigma: \{1,2,\ldots, n\} \to \{1,2,\ldots, n\}$.
The permanent makes its appearance in various branches of science and mathematics and algorithms to compute it are sought after. 
For instance, permanents of $0,1$-valued matrices are intimately connected to perfect matchings in bipartite graphs.
Consider a bipartite graph $G=(L,R,E)$ where $L,R$ is the bipartition of the vertex set of $G$ and $E$ is the set of edges of $G$.
Assume  $|L|=|R|=n$ and define an $n \times n $ matrix $A$ (adjacency matrix of $G$) whose $(i,j)$th entry is $1$ is an edge between the $i$th vertex of $L$ and the $j$th vertex of $R$.
It follows from the definition that the permanent of $A$ is equal to the number of perfect matchings in $G$.

The computational complexity of the permanent has been extensively studied in theoretical computer science.
 Valiant~\cite{Valiant79} proved that it is unlikely that there is an efficient algorithm that computes the permanent of a nonnegative matrix -- even when the matrix has only $0,1$ entries (the problem is $\mathsf{\#P}-$complete).   
This result, under standard assumptions in complexity theory, rules out an efficient algorithm to compute the permanent of a nonnegative matrix and raises the question of finding approximations to it.
Checking if $\mathrm{Per}(A)$ of a  nonnegative matrix is zero or not, however, is in $\mathsf{P}$ since it reduces to checking if the associated bipartite graph has a perfect matching or not.

\subsubsection{Doubly stochastic matrices and their permanents}

A special class of nonnegative matrices is doubly-stochastic matrices whose row sums and column sums are all equal to one.
\begin{definition}\label{def:DS} {\bf(Doubly stochastic matrix)}
An $n \times n$ matrix $A$ is said to be doubly stochastic if it is nonnegative and 
its rows and columns sum up to one:
For each $i$, $\sum_{j=1}^n A_{i,j}=1$ and for each $j$, $\sum_{i=1}^n A_{i,j}=1$.    
\end{definition}
\noindent
If a nonnegative matrix $A$ is an adjacency matrix of a graph $G$ each of whose vertices has degree $d \geq 1$, then the matrix $\frac{1}{d}A$ is doubly stochastic.
 The set of all doubly stochastic matrices is convex and, in fact, a polytope -- the Birkhoff polytope \cite{birk:46}.
   The well-known Birkhoff-von Neumann theorem states that the Birkhoff polytope is a convex hull of $n \times n$ permutation matrices.

The matrix with all entries $\frac{1}{n}$ is doubly stochastic.
Its permanent is $\frac{n!}{n^n}$.
van der Waerden conjectured that the permanent of any $ n \times n$ doubly-stochastic matrix must be at least $\frac{n!}{n^n}$.
Interestingly, this lower bound does not depend on the entries of $A$ as long as it is doubly stochastic.
Egorychev~\cite{Egorychev81} and Falikman~\cite{Falikman81} proved the van der Waerden conjecture.

\begin{theorem}\label{thm:EF}{\bf (Permanent of doubly-stochastic matrices \cite{Egorychev81,Falikman81})}
For any $n \times n$ doubly-stochastic matrix $A$, $\mathrm{Per}(A) \geq \frac{n!}{n^n}$.
\end{theorem}
\noindent
On the other hand, just for a   row-stochastic  matrix $A$, it trivially holds that $$\mathrm{Per}(A) \leq \prod_{i=1}^n\sum_{j=1}^n A_{i,j}=1.$$
Since $\frac{n!}{n^n} \geq e^{-n}$, for a doubly-stochastic matrix, the permanent is between $e^{-n}$ and $1$.
Hence, if  $A$ is doubly stochastic, then we can  output $1$ and this is an $e^{n}$ approximation to its permanent.

\subsubsection{Matrix scaling}
The starting point  of the work Linial, Samorodnitsky, and Wigderson \cite{LSW98} (see also the journal version of this paper \cite{LSW98}) 
is the observation that the
permanent (of any matrix) has certain {\em symmetries}: For positive vectors $x,y \in \mathbb{R}_{>0}$, if we define $B:=XAY$ where $X$ and $Y$ are diagonal matrices corresponding to vectors $x$ and $y$ respectively, then we can write down the permanent of $B$ exactly:
\begin{equation}\label{eq:scaling}
\mathrm{Per}(B)=\left(\prod_{i=1}^n x_i  \right) \mathrm{Per}(A)\left(\prod_{j=1}^n y_j  \right).
\end{equation}
This operation of left and right multiplying $A$ with diagonal matrices is referred to as (matrix) {\em scaling}.
Thus, in the case $A$ is not doubly stochastic (something that can be  efficiently checked), one can try to find a  scaling $(x,y)$ of $A$ such that $B$ is doubly stochastic.
If so, one can output 
$$\frac{1}{\left(\prod_{i=1}^n x_i  \right) \left(\prod_{j=1}^n y_j  \right)}$$ 
as an approximation for $\mathrm{Per}(A)$.
From the discussion in the previous section, such an algorithm would be an $e^{n}$  approximation to the permanent.

An approach to finding such a scaling is to do the following iteratively: Find a vector $x$ that ensures that all the rows of the scaled $A$ sum up to one, and then pick a $y$ that ensures the same for the columns.
This matrix scaling algorithm was suggested by Sinkhorn \cite{Sinkhorn64}.
Franklin and Lorenz \cite{FL89} analyzed the convergence rate of Sinkhorn's scaling algorithm.
They showed that, when a doubly-stochastic scaling of $A$ exists, Sinkhorn's algorithm outputs a matrix $B$ that is $\eps$ away (in $\ell_\infty$-distance) from being doubly stochastic and, to do so, it takes a polynomial number of iterations in  the number of bits needed to represent the input matrix $A$ and  $\frac{1}{\eps}$.
Kalantari and Khachiyan \cite{KK96}   gave a convex-optimization-based algorithm to check if $A$ can be scaled to a doubly-stochastic matrix and, if it can be, then to find an $\eps$ approximation to it.
The running time of their algorithm is polynomial in the number of bits needed to represent the input matrix $A$ and  $\log \frac{1}{\eps}$; thus, giving a deterministic polynomial time algorithm that approximates the permanent of a nonnegative $A$ to within a multiplicative factor of $e^n$.

The question that \cite{LSW98} studied is if the number of iterations can be made independent of the number of bits needed to represent $A$.
Such an algorithm, whose number of iterations does not depend on the entries of $A$, is referred to as a {\em strongly polynomial} time algorithm.
At its core, this turns out to be related to the following mathematical question:
{If  $\mathrm{Per}(A)>0$, then how small can it get as a function of the entries of $A$?}
Theorem \ref{thm:EF} \cite{Egorychev81,Falikman81} implies that, if $A$ is doubly stochastic, then this cannot get below $e^{-n}$.

\medskip
\noindent
{\bf Preprocessing step.} The idea in 
\cite{LSW98} is to augment Sinkhorn's scaling algorithm with a {\em preprocessing} step that, in the beginning, scales the columns of $A$ to ensure that the permanent of the new matrix is lower bounded by $n^{-n}.$ 
They do so by first efficiently finding a permutation $\sigma \in S_n$  that maximizes $\prod_{i=1}^nA_{i,\sigma(i)}$.
They then show that there is a positive diagonal matrix $Y$ such that $B=AY$ and, for all 
$1 \leq i,j \leq n$, $B_{i,\sigma(i)} \geq B_{i,j}$.
This ensures that if we normalize the rows of $B$ such that each of them sums up to one, the permanent of the resulting matrix is at least $\frac{1}{n^n}$.

\medskip
\noindent
{\bf Potential function and measuring progress.} 
To analyze the progress in Sinkhorn's scaling algorithm, \cite{LSW98} consider the permanent itself as the potential function.
If $A_t$ is the matrix at the beginning of the $t$th iteration of Sinkhorn's scaling algorithm, they show that, as long as $A_t$ is {\em far} from being doubly stochastic,
\begin{equation}\label{eq:per_improve}   \mathrm{Per}(A_{t+1}) \gtrsim \left( 1+ \frac{1}{n}\right)\mathrm{Per}(A_t).
\end{equation}
They use the following potential function that measures the distance of a matrix $B$ from being doubly stochastic:
\begin{equation}\label{eq:spotential} \mathrm{ds}(B):=\|R(B)-I\|_F^2 + \|C(B)-I\|_F^2.
\end{equation}
Here $R(B),C(B)$ are  diagonal matrices whose $(i,i)$th entries are the sum of the $i$th row and $i$th column respectively.

To gain some intuition why \eqref{eq:per_improve} is true, first note that if we have positive numbers $c_1,\ldots,c_n$ that sum up to $1$ and are more than $\delta$ distance from all one vector ($\|1-c\|_2^2 \approx \delta$) then 
$\prod_{i=1}^n c_i \lesssim 1-\frac{\delta}{2}$.
Hence, if we have a matrix $B$ that is row stochastic and we scale its columns to $1$, i.e., consider $B C^{-1}$, where $C$ is the diagonal matrix corresponding to the column sums of $B$, then 
$$\mathrm{Per}(BC^{-1})=\frac{\mathrm{Per}(B) }{\prod_{i=1}^nc_i} \gtrsim \mathrm{Per}(B) \cdot (1+\delta).$$
Thus, as long as $\delta \geq  \frac{1}{n}$, the permanent increases by a multiplicative factor of $1+\frac{1}{n}$.

\medskip
\noindent
{\bf Termination condition.}
If after $t$ iterations, $\mathrm{ds}(A_t) \geq \frac{1}{n}$, then $$\mathrm{Per}(A_{t+1}) \gtrsim \left(1+\frac{1}{n} \right)^t\mathrm{Per}(A_1).$$
Since the permanent of a row-stochastic matrix is upper bounded by $1$, and $\mathrm{Per}(A_1) \geq \frac{1}{n^n}$ due to the preprocessing step, the above cannot continue for more than about $n^2$ iterations. 
Thus, after roughly $n^2$ iterations, $\mathrm{ds}(A_{t}) < \frac{1}{n}$, and $A_{t}$ is close to a doubly stochastic matrix.
Finally, \cite{LSW98} prove an approximate version of Theorem \ref{thm:EF} 
and lower bound the permanent of approximately doubly-stochastic matrices. 
Roughly speaking, they show that if $B$ is row stochastic and $\mathrm{ds}(B)< \frac{1}{n}$, then $\mathrm{Per}(B)>\frac{1}{e^{n(1+o(1))}}$.
Thus, we can output the matrix produced after about $n^2$ iterations.
This completes the sketch of the proof of the following theorem.

\begin{theorem}\label{thm:lsw}{\bf (Approximating permanent via matrix scaling \cite{LSW98})}
There is an algorithm that, given an $n \times n$ nonnegative matrix $A$, computes a number $Z$ such that 
$\mathrm{Per}(A) \leq Z \leq e^{n(1+o(1))} \cdot \mathrm{Per}(A)$ using $\tilde{O}(n^5)$ elementary operations.
\end{theorem}

\noindent
Subsequent to the work of  \cite{LSW98}, Jerrum, Sinclair, and Vigoda~\cite{JSV04}, building upon a long line of work, showed that the Markov Chain Monte Carlo framework can be deployed to obtain a randomized algorithm to estimate the permanent of any nonnegative matrix to within a factor of $1+\eps$ in time that is polynomial in the bit-lengths of $A$ and $\frac{1}{\eps}$.
As for deterministic algorithms, 
in a follow-up work, Gurvits and Samorodnitsky \cite{GS02} show how  scalings can be viewed as  solutions to certain  convex programs -- leading to convex programming relaxations for the permanent and better deterministic approximations; see Section \ref{sec:poly_cap} and \cite{StraszakV17Permanent} for a discussion.
This line of work on deterministic approximation algorithms has recently been generalized to a class of general  counting and optimization problems; see  \cite{StraszakV17,AnariG17}.

\subsection{Noncommutative  singularity testing and operator scaling}\label{sec:operator}

Edmonds \cite{Edmonds} considered the following generalization of checking whether the permanent of a nonnegative matrix is zero or not:
Given an $m$-tuple of $n \times n$ complex matrices $A_1,\ldots,A_m$, is there a singular matrix in their linear space (over $\mathbb{C}$) or not?
This {\em singularity} problem is equivalent to deciding if the polynomial 
$$p_{A_1,\ldots,A_m}(x_1,\ldots,x_m):=\det\left(x_1A_1+\cdots + x_mA_m \right)$$
is identically zero or not.
$p_{A_1,\ldots,A_m}$ is a homogeneous polynomial of degree $n$ and can be efficiently evaluated at any given point.
To see how deciding if a bipartite graph has a perfect matching is a special case of Edmonds' singularity problem, we let $A_i$ be the matrix which has a $1$ only at the entry corresponding to the $i$th edge in the associated graph and $0$ elsewhere; see 
\cite{LovaszSingular}.

There is a simple and efficient randomized algorithm to test this:  
Pick independent and random values for each of the variables $x_1,\ldots,x_m$ from the set $\{1,2,\ldots, 2n\}$ and output the value of  $p_{A_1,\ldots,A_m}$ for this input. 
It can be shown that if $p_{A_1,\ldots,A_m}$ is not identically zero then, with probability at least $\frac{1}{2}$, this algorithm outputs a nonzero value.
By repeating an appropriate number of times, this probability can be amplified to any number less than $1$.
This problem is an instance of the {\em Polynomial Identity Testing}  (PIT) problem where one is given a polynomial and the goal is to check if it is identically zero or not.
The randomized algorithm mentioned above works for PIT as well.
While for some special cases of PIT deterministic algorithms are known (e.g., the deterministic primality testing algorithm of Agrawal, Kayal, and Saxena \cite{AgrawalKaSa04}),
the problem of coming up with an efficient deterministic algorithm for PIT remains open.
We mention that Edmonds' singularity problem is almost the same as the fully general  PIT problem due to a result of Valiant \cite{Valiant79} that establishes the ``universality'' of the determinant.
\cite{KabanetsI04} proved that derandomizing PIT implies arithmetic circuit lower bounds for the complexity class {\sf NEXP}; tying the goal of derandomizing PIT to one of the central goals of theoretical computer science: that of proving circuit lower bounds.

Gurvits \cite{Gurvits04} considered a version of Edmonds' singularity problem and  reformulated it in terms of {\em completely positive operators} that take positive definite matrices to positive definite matrices.
Subsequently, he generalized the matrix scaling algorithm of Linial, Samorodnitsky, and Wigderson \cite{LSW98} to {\em operator scaling} for this problem.
He introduced a potential function -- {\em capacity} -- that can track the progress of the operator scaling algorithm and used it to give deterministic polynomial time algorithms for Edmonds' singularity problem for various special cases (Section \ref{sec:cap}).
However, he could not prove a bound on the number of iterations of his operator scaling in general. 
The main result of the paper by Garg, Gurvits, Oliviera, and Wigderson \cite{GargGOW16} is a bound on the number of iterations of Gurvits' operator scaling algorithm.
The key ingredient in their analysis is a lower bound on the   capacity of a completely positive operator (Section \ref{sec:cap_lower}).
This implies that Gurvits' operator scaling algorithm can also approximate the capacity of a completely positive operator to any accuracy in polynomial time.
Moreover, \cite{GargGOW16} show that this algorithm implies a  deterministic polynomial time algorithm for testing a {\em noncommutative} version of Edmonds' problem (Section \ref{sec:ncit}).
Here, prior to the work of \cite{GargGOW16}, the best algorithms (whether randomized or deterministic) required an exponential time algorithm \cite{IQS17}. 
In a companion paper Garg, Gurvits, Oliviera, and Wigderson \cite{GargGOW17} show the application of this operator scaling machinery to the various computational problems involving the Brascamp-Lieb inequalities (Section \ref{sec:BL}).

\subsubsection{Completely positive operator and its capacity}\label{sec:cap}
Let $M_n(\mathbb{C})$ denote the set of $n\times n$ matrices with complex entries.
Let $\mathrm{GL}_n(\mathbb{C})$ denote the degree $n$ general linear group of $n \times n$ invertible matrices over $\mathbb{C}$.
Let $\mathrm{SL}_n(\mathbb{C})$ denote the degree $n$ special linear group of $n \times n$  matrices over $\mathbb{C}$ with determinant $1$.
Both of the above are groups with respect to ordinary matrix multiplication.
Let ${H}_n(\mathbb{C})$ denote the set $n \times n$ Hermitian matrices.
Let $S^n_+$ denote the set of $n \times n$ complex  positive semi-definite (PSD) matrices and 
let $S^n_{++}$ denote the set of $n \times n$ complex  positive definite (PD) matrices.
For two matrices $X,Y$ their tensor product is denoted by $X \otimes Y$.

\begin{definition}\label{def:PT}{\bf (Completely positive operator)} For positive integers $n_1 \geq n_2$, an operator $T:M_{n_1}(\mathbb{C}) \to M_{n_2}(\mathbb{C})$ is said to be completely positive if there are $n_2 \times n_1$ complex matrices $A_1,\ldots,A_m$ such that, for $X \in S^{n_1}_{++}$,  $T(X)=\sum_{i=1}^m A_iXA_i^\dagger$.
The dual of $T$ is denoted by $T^*$ and is such that $T^*(Y)=\sum_{i=1}^m A_i^\dagger YA_i$ for $Y \in S^{n_2}_{++}$.
\end{definition}
\noindent
If $n_1=n_2=n$, we say that $T$ is a {\em square} operator. 
\begin{definition}\label{def:DSPT}{\bf (Doubly-stochastic completely positive operator)}
A completely positive operator $T:M_{n_1}(\mathbb{C}) \to M_{n_2}(\mathbb{C})$ is said to be doubly stochastic if $T\left(\frac{n_2}{n_1}I_{n_1}\right)=I_{n_2}$ and $T^*\left(I_{n_2}\right)=I_{n_1}$.
\end{definition}
 \cite{Gurvits04} introduced the following notion of capacity for completely positive operators.
\begin{definition}\label{def:cap}{\bf (Capacity of a completely positive operator \cite{Gurvits04})} For a completely positive operator $T:M_{n_1}(\mathbb{C}) \to M_{n_2}(\mathbb{C})$, its capacity is defined as
$$ \Cap(T):= \inf\left\{\frac{\det\left(\frac{n_2}{n_1}T(X)\right)}{\det(X)^{\frac{n_2}{n_1}}}: X \succ 0\right\}.$$
\end{definition}
We focus on the square case and return to the rectangular (nonsquare) case in Section \ref{sec:BL}.
In the square case ($n_1=n_2=n$),
$$ 
 \Cap(T):=\inf\{\det(T(X)): X \succ 0, \ \det(X)=1\}. 
$$
A square operator is said to be {\em rank decreasing} if there is an $X \succeq 0$ such that $\mathrm{rank}(T(X)) < \mathrm{rank}(X)$.
Operators that are not rank decreasing are referred to as {\em rank nondecreasing}. 
The analog of this property in the matrix case (for a nonnegative matrix $A$) is as follows: For every nonnegative vector $x$, the number of coordinates of the vector $Ax$ that are positive is at least the number of coordinates of $x$ that are positive.
This is just Hall's condition and implies that the permanent of $A$ is positive.
\cite{Gurvits04} proved that, for a completely positive operator,
 $\Cap(T)>0$ if and only if  $T$ is rank nondecreasing.
\cite{GargGOW16} give other conditions that are equivalent for a completely positive operator to be rank nondecreasing.
One such condition that is relevant to proving a lower bound on the capacity is that there exist $d \times d$ matrices $F_1,\ldots,F_m$ for some $d$ such that the polynomial \begin{equation}\label{eq:cap_det}\det(F_1\otimes A_1+ \cdots  + F_m\otimes A_m) \neq 0.
\end{equation}
Similar to the notion of distance to a matrix to being doubly stochastic in Definition \ref{def:DS}, consider the following distance of a completely positive operator from being doubly stochastic:
\begin{equation}\label{eq:dso}
 \mathrm{ds}_{O}(T):=\Tr((T(I)-I)^2)+ \Tr((T^*(I)-I)^2).
 \end{equation}
An analog of Equation \eqref{eq:scaling} that captures the symmetries of the operator setting is as follows:
Let $T$ be a completely positive square operator defined by $A_1,\ldots,A_m$ and $B,C \in \mathrm{GL}_n(\mathbb{C})$.
Then, if we define $T_{B,C}$ to be the operator defined by $BA_1C,\ldots,BA_mC$, then 
\begin{equation}\label{eq:operatorscaling}
    \Cap(T_{B,C})= |\det(B)|^2 \cdot   \Cap(T) \cdot  |\det(C)|^2.
\end{equation}
If $B,C \in \mathrm{SL}_n(\mathbb{C})$, then 
 $ \Cap(T_{B,C})=\Cap(T)$.
This is true because the capacity is defined in terms of determinants and, hence, the symmetries of the determinant arise.
It is worth noting that the polynomials in (the l.h.s. of) Equation \eqref{eq:cap_det} are invariant when $B,C$ have determinant $1$. 
In fact, these polynomials linearly span the space of all such invariant polynomials.

Moreover, suppose $T$ is a completely positive operator specified by $A_1,\ldots,A_m$ and either $\sum_{i=1}^m A_i A_i^\dagger =I$ (row-stochastic) or $\sum_{i=1}^m A_i^\dagger A_i =I$  (column-stochastic) then it follows from  the AM-GM inequality that 
\begin{equation}\label{eq:cap_upper}
\Cap(T) \leq \det(T(I)) \leq \left(\frac{\Tr (T(I))}{n} \right)^n=1.
\end{equation}

\subsubsection{Operator scaling}\label{sec:cap_lower}
We present a sketch of the operator scaling algorithm and its analysis.
Suppose $T$ is a  completely positive operator  specified by $m$ $n \times n$ matrices $A_1,\ldots,A_m$, where each entry of each matrix is an  integer bounded in absolute value by $M$.
Our goal is to decide if $\Cap(T)>0$ or not. 
Or equivalently, to decide if $T$ is rank nondecreasing.

An operator scaling of $T$ is given by positive matrices $B,C$ such that 
the operator $T_{B,C}$ defined by $B^{\frac{1}{2}}A_1C^{\frac{1}{2}},\ldots,B^{\frac{1}{2}}A_mC^{\frac{1}{2}}$ is doubly stochastic.
The left normalization (or scaling) of $T$, denoted by $T_L$, is defined as 
$$T_L(X):=T(I)^{-\frac{1}{2}}T(X) T(I)^{-\frac{1}{2}}$$
and the right normalization (or scaling) of $T$ is defined as $$T_R(X):=T(T^*(I)^{-\frac{1}{2}}X T^*(I)^{-\frac{1}{2}}).$$
It follows  that $T_L(I)=I$ and $T_R^*(I)=I$.

Gurvits' operator scaling algorithm \cite{Gurvits04} follows the same outline as the matrix scaling algorithm analyzed in \cite{LSW98}.
It first checks if both $T(I)$ and $T^*(I)$ are  nonsingular. 
If not, then $T$ is rank decreasing and the algorithm stops.
Else, it keeps performing left and right normalizations on $T$ until the distance to double stochasticity is below $\frac{1}{n}$.
If the operator $T$ is rank decreasing, then one can argue that the left and right normalizations cannot make it rank nondecreasing. 
Thus, the algorithm will always output rank decreasing in this case.

\cite{GargGOW16} prove that if $T$ is rank nondecreasing, then after a small-enough number of iterations $t$, the operator $T_t$ is such that $\mathrm{ds}_{O}(T_t) < \frac{1}{n}$.
This is analogous to the matrix case: They show that every iteration such where $\mathrm{ds}_{O}(T_t)>\frac{1}{n}$, 
$$ \Cap(T_{t+1}) \gtrsim \left(1+\frac{1}{n} \right) \Cap (T_t).$$
Since there is an upper bound of $1$ on the capacity of a row-stochastic  operator,  it remains to lower bound  $\Cap(T_1)$ when $\Cap(T)>0$.
We note that for the matrix case, we used permanent as a measure, but could have also used an appropriate notion of capacity as defined in Section \ref{sec:poly_cap}.

The main technical contribution of \cite{GargGOW16} is a lower bound on the capacity of a right-normalized completely positive operator. 
Let $T_A$ be a completely positive operator an operator specified by integer-valued matrices $A_1,\ldots,A_m$ each of whose entry is bounded in absolute value by $M$. 
Let $T$ be the right normalization of $T_A$.
Then, it follows that  $$\Cap(T)=\frac{\Cap(T_A)}{\det(T^*(I))}.$$
Thus, to lower bound $\Cap(T)$, it is sufficient to lower bound $\Cap(T_A)$ and upper bound $\det(T^*_A(I))$.
The latter follows from an upper bound on 
$$\Tr(T^*_A(I))= \sum_{i=1}^m \Tr \left(A_i^\dagger A_i\right) \leq M^2mn^2.$$
Thus, by the AM-GM inequality 
\begin{equation}\label{eq:det_upper}\det(T^*(I))\leq \left(\frac{\Tr(T^*_A(I))}{n}\right)^n \leq (Mmn)^n.
\end{equation}
The original proof of  a lower bound on the capacity of a nondecreasing completely positive operator $T_A$ relied on degree bounds in invariant theory; we return to it in the next section.
Here we mention their proof based on Alon's Combinatorial Nullstellensatz \cite{alon_1999}; see also \cite{Vishnoi_Null}. 
Alon's result states that if $p(z_1,\ldots,z_\ell)$ is a nonzero polynomial (over $\mathbb{C}$) with the degree of $z_i$ is $d_i$,
then there are nonnegative integers $(a_1,\ldots,a_\ell)$  such that $\sum_{i=1}^\ell a_i \leq d$ and $a_i \leq d_i$ such that $p(a_1,\ldots,a_\ell)\neq 0$.

From Equation \eqref{eq:cap_det}, we know that if  $T$ is rank nondecreasing, then there exist $d \times d$ matrices $F_1,\ldots,F_m$ for some $d$ such that the polynomial $\det(F_1\otimes A_1+ \cdots  F_m\otimes A_m) \neq 0$.
Thus, the (ordinary) polynomial $\det(X_1\otimes A_1+ \cdots  X_m\otimes A_m)$ (in the variables corresponding to entries of matrices $X_1,\ldots,X_m$) is nonzero.
Thus, Alon's result implies that there exist
integer-valued matrices $D_1,\ldots,D_m$  such that 
$\det(D_1\otimes A_1+ \cdots  D_m\otimes A_m) \neq 0$ and,  importantly, the sum of the square of all the entries of all the matrices is bounded by $n^2d$.

Let $X \succ 0$ and define $C_i:=T_A(X)^{-\frac{1}{2}}A_iX^{\frac{1}{2}}$.
Thus, $\sum_{i=1}^m C_iC_i^\dagger=I$ and, hence $\Tr\left( \sum_{i=1}^m C_iC_i^\dagger\right)=n$.
Now, let $Y:=D_1\otimes C_1 + \cdots + D_m \otimes C_m$.
Then, on the one hand, by the AM-GM inequality, $$\det(YY^\dagger) \leq \left(\frac{\Tr(YY^\dagger)}{nd} \right)^{nd} \leq \left(\frac{n^3d}{nd}\right)^{nd}= n^{2dn},$$ where one uses the bound on the sum of the square of entries of $D_i$s.
On the other hand, 
$$ \det(YY^\dagger) = |\det(Y)|^2 \geq | \det(D_1 \otimes A_1 + \cdots + D_m \otimes A_m)|^2 \det(X)^d \cdot \det(T_A(X))^{-d}.$$
Since all entries of $D_1 \otimes A_1 + \cdots + D_m \otimes A_m$ are integers and its determinant is nonzero, $|\det(D_1 \otimes A_1 + \cdots + D_m \otimes A_m)|\geq 1$, implying 
\begin{equation}\label{det_lower}
\det(T_A(X)) \geq (\det(YY^\dagger))^{-\frac{1}{d}} \geq n^{-\frac{2dn}{d}}= \frac{1}{n^{2n}}. 
\end{equation}
Thus, combining Equations \eqref{eq:det_upper} and \eqref{det_lower}, we obtain the following theorem.
\begin{theorem}\label{thm:cap_lower}{\bf (Lower bound on the capacity of a rank nondecreasing  operator \cite{GargGOW16})}
Let $T$ be the right-normalized version of a rank nondecreasing and  completely positive operator given by $A_1,\ldots, A_m$, where each $A_i$ is an $n\ \times n$ integer matrix with each entry bounded in  absolute value  by $M$.
Then, $\Cap(T) \geq \frac{1}{(Mmn^3)^n}$.
\end{theorem}

\noindent
As discussed above, this implies the following theorem to check if a completely positive operator is rank nondecreasing or, equivalently, if its capacity is positive.
\begin{theorem}\label{thm:ranknondecreasing}{\bf (Checking if  a completely positive  operator is rank nondecreasing \cite{GargGOW16})}
There is an algorithm that, given a completely positive operator $T$ 
given by $A_1,\ldots, A_m$, where each $A_i$ is an $n\ \times n$ integer matrix with each entry bounded in  absolute value  by $M$, decides if $T$ is rank nondecreasing or not in time polynomial in $n,m$, and $\log M$.
\end{theorem}
While we did bound the number of iterations needed by Gurvits' operator scaling algorithm for the above theorem, we omitted a discussion on ensuring that the bit complexity of the numbers that arise in the execution of the algorithm remain polynomially bounded in the input bit length; see \cite{GargGOW16} for details.

\cite{GargGOW16} also show how an adaptation of Gurvits' operator scaling algorithm can be used to obtain an approximation of the operator capacity. 
We omit the algorithm and the proof.

\begin{theorem}\label{thm:capacity}{\bf (Approximating the capacity of an  operator \cite{GargGOW16})}
There is an algorithm that, given a completely positive operator $T$ on dimension $n$, and described by $b$ bits, outputs a $1+\eps$ multiplicative approximation to $\Cap(T)$ in time polynomial in $n,b,\frac{1}{\eps}$.
\end{theorem}

\noindent
In a subsequent work, B\"urgisser,  Garg, Oliveira,  Walter, and Wigderson \cite{BurgisserGOWW18}
present a  generalization of operator scaling to {\em tensor scaling}; we omit the details.
We note that, unlike the matrix and operator scaling case, to test scalability, it is not sufficient
to take $\eps$ which is polynomially small. 
Currently, there is no known polynomial time algorithm for testing the scalability of tensors.

In another follow-up work, B\"urgisser, Franks, Garg, Oliveira,  Walter, and Wigderson \cite{BurgisserFGOWW18} study the {\em nonuniform} version of scaling where one is given prescribed {\em marginals} and an input matrix/operator/tensor, and the goal is to decide if we can scale
the input to have the prescribed marginals?
For instance, instead of scaling a nonnegative matrix so that the row sums and column sums are all one, one may ask to find a scaling to a specified row sum vector $r$ and a column sum vector $c$.
 In the matrix scaling case, the
theory of nonuniform scaling is not much different from the theory of uniform scaling. 
However in
the operator and tensor scaling settings, the nonuniformity presents additional challenges; see \cite{BurgisserFGOWW18}.

\subsubsection{Noncommutative singularity and identity testing}\label{sec:ncit}
Let $A_1,\ldots,A_m \in M_n(\mathbb{C})$ and consider $x_1,\ldots,x_m$ to be noncommutative variables.
The algorithmic problem, which is a noncommutative version of Edmonds' singularity problem, is to check if $L:=\sum_{i=1}^m x_iA_i$ is invertible (nonsingular) over the skew-field (also known as division ring or field of fractions) of $x_1,\ldots,x_n$.
This notion of nonsingularity is nontrivial to define and 
there are several equivalent ways to define it. 
Perhaps the simplest is if there is a way of ``plugging in'' matrix for each $x_i$ to get an invertible matrix, i.e., do there exist $d \times d$ matrices $B_1, \ldots, B_m$ (for some $d$) s.t. $\sum_{i=1}^m B_i \otimes A_i$ is invertible.

The connection between the noncommutative singularity problem and the capacity of a completely positive operator is as follows:
Consider the completely positive operator $L(X):=\sum_{i=1}^m A_i X A_i^\dagger$
defined by the matrices $A_1,\ldots,A_m$ input to the noncommutative singularity problem.
Then,  $\sum_{i=1}^m x_iA_i$ is singular over the skew-field if and only if there is an $X \succ 0$ such that $\mathrm{rank}(L(X)) < \mathrm{rank}(X)$, i.e., the completely positive operator $L$ is rank decreasing.
Thus, from Theorem \ref{thm:ranknondecreasing}, it immediately follows that the problem of checking noncommutative singularity is in {\sf P}.
\begin{theorem}\label{thm:ncpit}{\bf (Noncommutative singularity testing \cite{GargGOW16})}
There is a deterministic algorithm, that given $m$ $n \times n$ matrices $A_1,\ldots,A_m$ whose entries need at most $b$ bits to represent, decides in time $\mathrm{poly}(n,m, b)$ if the matrix $L = \sum_{i=1}^m x_iA_i$ is invertible
over the free skew field.
\end{theorem}

\noindent
Polynomial identity testing, in the commutative setting, 
 captures the polynomial and rational function identity test for formulas \cite{Valiant79}.
The same is not true in the noncommutative setting.
However, Cohn \cite{Cohn} proved that there is an efficient algorithm that converts every arithmetic formula
$\phi(x)$ in noncommuting variables  of size $s$ to a symbolic matrix $L_\phi$ of size $\mathrm{poly}(s)$, such that the
rational expression computed by $\phi$ is identically zero if and only if $L_\phi$ is singular.
Theorem \ref{thm:ncpit} implies that 
there is a deterministic algorithm, that,  for any noncommutative formula over $\mathbb{Q}$ of size $s$ and bit complexity $b$,
determines in $\mathrm{poly}(s, b)$ steps if it is identically zero.
Thus, the noncommutative rational identity testing problem is in {\sf P}; see also the works of  \cite{IvanyosQS18,HamadaHirai} for different proofs of this result. 
Note that Theorem \ref{thm:ncpit} requires access to the matrices $A_1,\ldots,A_m$.
The problem of proving an analogous result when we have only black-box access to $\sum_{i=1}^m x_iA_i$ remains open.
We note that, in the noncommutative setting, inversions are nontrivial to handle than in the commutative setting where we can push them out and eliminate them. 
Indeed, an efficient deterministic algorithm to check if a noncommutative formula {\em without} inversions is identically zero was known; see Raz and Shpilka \cite{RazS05}.

\subsubsection{Brascamp-Lieb constants}
\label{sec:BL}

Let $n$, $m$, and $(n_j)_{j\in[m]}$ be positive integers and $p:=(p_j)_{j\in[m]}$ be nonnegative real numbers.
Let $B:=(B_j)_{j\in[m]}$ be an $m$-tuple of linear transformations where $B_j$ is a surjective linear transformation from $\R^{n}$ to $\R^{n_j}$.
The corresponding  Brascamp-Lieb {\em datum} is denoted by $(B,p)$.
The Brascamp-Lieb inequality states that for each Brascamp-Lieb datum $(B,p)$ there exists a constant $C(B,p)$ (not necessarily finite) such that for any selection of real-valued, nonnegative, Lebesgue measurable functions $f_j$ where $f_j:\R^{n_j}\rightarrow \R$,
\begin{equation}
\int_{x\in\R^n} \left(\prod_{j\in[m]} f_j(B_j x)^{p_j}\right)dx
\leq
C(B,p) \prod_{j\in[m]}\left(\int_{x\in \R^{n_j}} f_j(x) dx \right)^{p_j}.
\label{eq:BLinequality}
\end{equation}
The smallest constant that satisfies \eqref{eq:BLinequality} for any choice of $f:=(f_j)_{j\in[m]}$ satisfying the properties mentioned above is called the Brascamp-Lieb {\em constant} and we denote it by $\BL(B,p)$. 
Brascamp-Lieb inequalities generalize  many inequalities used in analysis and all of mathematics, such as the H\"older inequality and Loomis-Whitney; see the paper by Brascamp and Lieb \cite{brascamp1976best}. 

A Brascamp-Lieb datum $(B,p)$ is called \emph{feasible} if $\BL(B,p)$ is finite, otherwise, it is called \emph{infeasible}.
Bennett,  Carbery,  Christ,  and Tao \cite{bennett2008brascamp} proved that the constant $ \mathrm{BL}(B,p)$ is nonzero whenever $p$ belongs to the  set $P_B \subseteq \R^m$ defined as follows:
$$\textstyle P_B := \left\{p\in \R^m_{\geq 0}: \sum_{j=1}^m p_j \dim(B_j U) \geq \dim(U), \mbox{ for every lin. subspace }U\subseteq \R^n\right\}.$$
Note that the above definition has infinitely many linear constraints on $p$ as $V$ varies over different subspaces of $\R^n$.
However, there are only finitely many different linear restrictions as $\dim(B_j V)$ can only take integer values from $[n_j]$.
Consequently, $P_{B}$ is a convex set and, in particular, a polytope. 
Examples of Brascamp-Lieb polytopes include matroid basis polytopes and linear matroid intersection polytopes; see \cite{GargGOW17}.

A Brascamp-Lieb inequality is nontrivial only when $(B,p)$ is a feasible Brascamp-Lieb datum.
Therefore, it is of interest to  characterize feasible Brascamp-Lieb data and compute the corresponding Brascamp-Lieb constant.
Towards this, Lieb~\cite{lieb1990gaussian} showed that one needs to consider only Gaussian functions as inputs for \eqref{eq:BLinequality}.
This result suggests the following characterization of the Brascamp-Lieb constant as an optimization problem.
\begin{theorem}\label{thm:lieb}{\bf (Gaussian maximizers \cite{lieb1990gaussian})}
Let $(B,p)$ be a Brascamp-Lieb  datum with $B_j\in\R^{n_j\times n}$ for each $j\in[m]$.
Then,
\begin{equation}
\frac{1}{\BL(B,p)^2}=\inf\left\{\frac{\det\left({\sum_{j=1}^m p_j B_j^\top Y_j B_j}\right)}{\prod_{j=1}^m \det(Y_j)^{p_j}}: Y_j\in \R^{n_j \times n_j}, Y_j\succ 0, j=1,2, \ldots, m\right\}.
\label{eq:BLconstant}
\end{equation}
\label{thm:BLconstant}
\end{theorem}

\noindent

\noindent
One of the computational questions concerning the Brascamp-Lieb inequality is: Given a Brascamp-Lieb datum $(B,p)$, can we compute $\BL(B,p)$ in time that is polynomial in the number of bits required to represent the datum?
Since computing $\BL(B,p)$ exactly may not be possible due to the fact that this number may not be rational even if the datum $(B,p)$ is, one seeks an arbitrarily good approximation.
Formally, given the entries of $B$ and $p$ in binary, and an $\eps>0$, compute a number $Z$ such that 
\begin{equation*}
 \BL(B,p) \leq Z \leq  (1+\eps) \; \BL(B,p)
\end{equation*}
in time that is polynomial in the combined bit lengths of $B$ and $p$ and $\log \frac{1}{\eps}.$

There are a few obstacles to this problem: (a) Checking if a given Brascamp-Lieb datum is feasible is not known to be in {\sf P}. 
(b)  The formulation of the Brascamp-Lieb constant by Lieb~\cite{lieb1990gaussian} as in \eqref{eq:BLconstant} is neither concave nor logconcave in the usual sense.
Thus,  techniques developed in the context of linear and convex optimization do not seem to be directly applicable.

Garg, Gurvits, Oliviera, and Wigderson \cite{GargGOW17} gave an algorithm to compute the Brascamp-Lieb constant in polynomial time when the vector $p$ is rational and given in unary. 
More precisely, the running time of their algorithm to compute $\BL(B,p)$ up to multiplicative error $1+\epsilon$ has a polynomial dependency to $\epsilon^{-1}$ and the {\em magnitude} of the denominators in the components of $p$ rather than the number of bits required to represent them.
They also presented algorithms with similar running times for checking if a Brascamp-Lieb datum is feasible, or if a given point is approximately in the Brascamp-Lieb polytope. 
The key idea in \cite{GargGOW17}  is to use Lieb's characterization (Theorem  \ref{thm:lieb}) to reduce the problem of computing $\BL(B,p)$ to the problem of computing the capacity of a completely positive operator.
We note that the special case  when the matrices are of rank $1$; i.e., $B_j\in \R^{1\times n}$ for every $j=1,2, \ldots, m$ was studied in \cite{StraszakV19}.
By interpreting Brascamp-Lieb constants in the rank-$1$ regime as solutions to certain entropy-maximization problems, \cite{SinghV14,StraszakV19} showed that they can be computed, up to a multiplicative precision $\eps>0$, in time polynomial in $m$ and $\log \frac{1}{\eps}$.

\medskip
\noindent
{\bf The reduction.}
Let $p_j=\frac{c_j}{c}$ for integers $(c_j)_{j\in[m]}$ and $c$.
\cite{GargGOW17} construct a completely positive operator $T_{B,p}$ such that $\Cap(T_{B,p})=\frac{1}{\mathrm{BL}(B,p)^2}$.
Let $m':=\sum_{j=1}^m c_j$ and consider a mapping $\sigma:[m']\to [m]$ which maps all those $i$ to $j$ that satisfy
$$ \sum_{k<j} c_k < i \leq \sum_{k \leq j}c_k.$$
Let $M_{ij}$ be an $n_{\sigma(i)} \times n$ matrix that is zero if $\sigma(i) \neq j$ and $B_{\gamma(i)}$ if $\gamma(i)=j$.
Now, for $\ell \in [m']$ define $A_\ell$ to be the block matrix whose rows are $M_{i\ell}$ for $i \in [m']$. 
$T_{B,p}$ is now a rectangular completely positive operator from ${M}_{nc}(\mathbb{C}) \to {M}_{n}(\mathbb{C})$ that maps a positive definite $X$ to $\sum_{i \in [m']} A_i^\dagger X A_i$.

Recall the capacity of a nonsquare completely positive operator (Definition \ref{def:cap}): $$\Cap(T_{B,p}):=\inf\left\{ \left( \frac{\det(T_{B,p}(X))}{c}\right): X \succ 0, \ \det(X)^{\frac{1}{c}}=1 \right\}.$$
Given the block form of each $A_i$, it follows that
$$ T_{B,p}(X)= \sum_{i=1}^{m'} B_{\sigma(i)}^\dagger X_i B_{\sigma(i)},$$
where $X_i$ is an appropriate submatrix of $X$.
Thus, it follows from the basic properties of the determinant that we can write
$$\Cap(T_{B,p})=\inf\left\{ \det \left( \frac{\sum_{i=1}^{m'}B_{\sigma(i)}^\dagger X_i B_{\sigma(i)}}{c}\right): X_i \succ 0, \ \prod_{i=1}^{m'}\det(X_i)=1 \right\}.$$
Replace $\sum_{i:\sigma(i)=j}X_i$ by $c_jY_j$ to obtain 
$$\Cap(T_{B,p})=\inf\left\{ \det \left( \frac{\sum_{j=1}^{m}c_jB_{j}^\dagger Y_j B_{j}}{c}\right): Y_j \succ 0, \ \prod_{j=1}^{m}\det(Y_j)^{c_j}=1\right\}=\frac{1}{\mathrm{BL}(B,p)^2}$$
via Theorem \ref{thm:lieb}.
To ensure we can use the algorithm developed for capacity, we also need to also prove that $T_{B,p}$ is rank nondecreasing. 
Towards this, first, we need to extend the notion of rank nondecreasing to nonsquare operators and then show that it satisfies this property; see \cite{GargGOW17} for the details.

Note that this construction does not lead to an optimization  problem whose dimension is polynomial in the input bit length as the size of the constructed operator in the operator scaling problem depends exponentially on the bit lengths of the entries of $p$. 
From the geodesic convexity of capacity (discussed in Section \ref{sec:geodesic}), it follows that the Brascamp-Lieb constant is also a solution to a geodesically convex optimization problem.
A succinct geodesically convex formulation was provided in \cite{SraVY18}.

\subsubsection{Polynomial capacity}\label{sec:poly_cap}
A basic version of the capacity of polynomials was considered in a paper by Gurvits and Samorodnitsky \cite{GS02} and then generalized to operators (Definition \ref{def:cap}) by \cite{Gurvits06}.
Subsequently, Gurvits  defined a notion of  capacity for hyperbolic polynomials in \cite{Gurvits06} and used it to prove a generalization of van der Waerden conjecture by Bapat  \cite{BAPAT1989107} for mixed discriminants.
In this section, we present this notion of polynomial capacity just for the setting of the permanent.
For an $n \times n$ nonnegative matrix $A$, consider the polynomial
$$f_A(x_1,\ldots,x_n) :=  \prod_{i=1}^n \sum_{j=1}^n A_{i,j} x_j.$$ 
 \cite{Gurvits06} considered the following notion of capacity:
\begin{equation}\label{eq:matrix_cap} \Cap(f_A):= \inf \left\{{f_A(x_1,\ldots,x_n)}: x_i>0, \ {\prod_{i=1}^n x_i} =1  \right\}.
\end{equation}
It is easily checked that if $A$ is stochastic then $0 \leq \Cap(f_A) \leq 1$, and $\Cap(f_A)=1$ if and only if  $A$ is doubly stochastic.
The main result of \cite{Gurvits06} when specialized for the above polynomial implied that 
$$ \mathrm{Per}(A) \geq \left( \frac{n!}{n^n} \right) \Cap(f_A),$$
giving an alternate proof of the van der Waerden conjecture (Theorem \ref{thm:EF}).
Thus, $\Cap(f_A)$ is an $e^{-n}$ approximation to $\mathrm{Per}(A)$.
One can also replace the permanent potential function with the capacity in the proof of \cite{LSW98} presented in Section \ref{sec:permanent}.
Moreover, after introducing new variables $y_i = \log x_i$ and replacing the objective with $\log f_A(x_1,\ldots,x_n)$, one obtains a convex program that can be solved efficiently; see  \cite{Gurvits06,SinghV14,StraszakV17Permanent,StraszakV19}.
This gives an alternate proof of Theorem \ref{thm:lsw}.
As discussed in previous sections, \cite{Gurvits06} arrived 
at this notion of capacity while trying to extend the work of Linial, Samorodnitsky, and Wigderson \cite{LSW98} to Edmonds' singularity problem. 
The proof technique in \cite{Gurvits06}  relied on the location of the roots of the polynomial under consideration ($f_A$ in the case of permanent).
This viewpoint itself has had far-reaching consequences in theoretical computer science and mathematics; see \cite{Vishnoi-Survey}.

\subsection{Capacity and geodesic convex optimization}\label{sec:geodesic}
In the most general setting, an optimization problem takes the form
$$\inf_{x\in K} f(x),$$
for some set $K$ and some function $f:K\to\R$.\footnote{Part of this section draws from \cite{VishnoiGeodesic}}
When $K \subseteq \R^d$, we can talk about the convexity of $K$ and $f$. 
$K$ is said to be convex if any ``straight line'' joining two points in $K$ is entirely contained in $K$, and $f$ is said to be convex if, on any such straight line, the average value of $f$ at the endpoints is at least the value of $f$ at the mid-point of the line.
When $f$ is ``smooth'' enough, there are equivalent definitions of convexity in terms of the standard differential structure in $\R^d$: the gradient or the Hessian of $f$.
Thus, convexity can also be viewed as a property arising from the interaction of the function and how we differentiate in $\R^n$; e.g., the Hessian of $f$ at every point in $K$ should be positive semi-definite. 
When both $K$ and $f$ are convex, the optimization problem is called a convex optimization problem.
The fact that the convexity of $f$ implies that any local minimum of $f$ in $K$ is also a global minimum, along with the fact that computing gradients and Hessians is typically easy in Euclidean spaces, makes it well-suited for developing first-order algorithms such as gradient descent and second-order algorithms such as interior point methods.
Analyzing the convergence of these methods boils down to understanding how well-behaved  derivatives of the function are, and there is a well-developed theory of algorithms for convex optimization  see \cite{boyd2004convex,nesterov2004introductory,vishnoi_book}.

Several optimization problems, however, are nonconvex.
An important example is that of the capacity of a completely positive operator (Definition \ref{def:cap}) 
\begin{equation}\label{eq:cap_geo}
 \Cap(T):=\inf\{\det(T(X)): X \succ 0, \ \det(X)=1\}=\inf\left\{\frac{\det(T(X))}{\det(X)}: X \succ 0\right\},
 \end{equation}
which is nonconvex as the objective function is nonconvex. 
However, \cite{GargGOW16} observed a curious property of the capacity: Consider the following Lagrangian of this optimization problem:
$$f(X,\lambda):= \log \det (T(X)) + \lambda \cdot \log \det X,$$
where $\lambda$ is the multiplier for the constraint.
Then, any $X$ for which $\nabla_X f(X,\lambda)=0$ is an optimal solution 
to Equation \eqref{eq:cap_geo}.
B\"urgisser,  Garg, Oliveira,  Walter, and Wigderson \cite{BurgisserGOWW18} mention that the capacity optimization problem, while nonconvex, is {\em geodesically convex}.
While the domain of positive definite matrices is convex in the ordinary sense, the key to showing that  capacity optimization is geodesically convex is to view this space as a manifold and redefine what it means to be a straight line by introducing a {\em metric}.

This redefinition of a straight line entails the introduction of a different differential structure.
Roughly speaking, a manifold is a topological space that locally looks like Euclidean space.
''Differentiable manifolds'' are a special class of manifolds that come with a differential structure that allows one to do calculus over them.
Straight lines on differential manifolds are called ``geodesics'', and  a set that has the property that a geodesic joining any two points in it is entirely contained in the set is called geodesically convex (with respect to the given differential structure).
A function that has this property that  its average value at the end points of a geodesic is at least the value of $f$ at the mid-point of the geodesic is called geodesically convex (with respect to the given differential structure). 
And, when $K$ and $f$ are both geodesically convex, the optimization problem  is called a geodesically convex optimization problem.
Geodesically convex functions also have key properties similar to convex functions such as the fact that a local minimum is also a global minimum.

Allen-Zhu,  Garg, Li, Oliveira, and Wigderson \cite{Allen-ZhuGLOW18} develop first-order and second-order methods for a class of geodesically  convex optimization problems that include capacity.
In this section, we first introduce the basics of geodesic convexity (Section \ref{sec:geodesic_cvx}),  show that the capacity optimization problem in Equation \eqref{eq:cap_geo} is geodesically convex (Section \ref{sec:cap_geodesic}), and give a high-level view of the algorithms in \cite{Allen-ZhuGLOW18} (Section \ref{sec:2ndorder}). 
We do not develop a theory of geodesic convexity here but give the minimal details to ensure that we can argue that the capacity function in \eqref{eq:cap_geo} is geodesically convex; see \cite{Udr94,VishnoiGeodesic} for a thorough treatment  on geodesic convexity.

\subsubsection{The Riemannian geometry of positive definite matrices and geodesic convexity}\label{sec:geodesic_cvx}
For simplicity, here we consider the case of real symmetric matrices and symmetric positive definite matrices.
Let $\Sn^n$ denote the space of all 
 $n \times n$ real symmetric matrices and  
let $\Sn^n_{++}$ denote the space of all $n \times n$ symmetric positive definite matrices.
$\Sn^n_{++}$ is a smooth manifold and the 
 tangent space at $P\in\Sn^n_{++}$ is $\R^{\frac{n(n+1)}{2}}$ which is homeomorphic to $\Sn^n$ for each $P\in \Sn^n_{++}$.
We consider the metric  induced by the Hessian of the  function: $-\log\det(P)$ for a positive definite matrix $P$.
This function is convex and 
the  metric is 
$$g_P(U, W):= \Tr[P^{-1}UP^{-1} W]$$
 for $P\in\Sn^n_{++}$ and $U,W\in\Sn^n$.
$g_P$ is clearly  symmetric, bilinear, and positive definite.
It is also nondegenerate as $\Tr[P^{-1}UP^{-1} W]=0$ for every $W$ implies 
$$\Tr[P^{-1}UP^{-1} U]=\Tr[P^{-\frac{1}{2}}UP^{-\frac{1}{2}} P^{-\frac{1}{2}}UP^{-\frac{1}{2}}]=0
$$
 or equivalently $P^{-\frac{1}{2}}UP^{-\frac{1}{2}}=0$.
Since $P$ is a nonsingular matrix, $P^{-\frac{1}{2}}UP^{-\frac{1}{2}}=0$ is equivalent to $U=0$.
Next, we observe that $\Sn^n_{++}$ with $g$ is a Riemannian manifold.
This follows from the observation that 
$g_P$ varies smoothly with $P$.

Since the metric tensor allows us to measure distances on a Riemannian manifold,  there is an alternative, and sometimes useful, way of defining geodesics on it: as length-minimizing curves.
Before we can define a geodesic in this manner, we need to define  the length of a curve on a Riemannian manifold.
This gives rise to a notion of distance between two points as the  minimum length of a  curve that joins these points.
Using the metric tensor we can measure the instantaneous length of a given curve.
Integrating along the vector field induced by its derivative, we can measure the length of the curve.
And, we can then define the shortest curve -- geodesic -- that connects two points.

It is well-known that the geodesic with respect to the Hessian of the log-determinant metric that joins $P$ to $Q$ on $\Sn^n_{++}$ can be parameterized as follows (see \cite{bhatia2009positive}):
\begin{equation}\label{eq:pd-geodesics}
\rho(t) := P^{\frac{1}{2}} (P^{-\frac{1}{2}} Q P^{-\frac{1}{2}})^t P^{\frac{1}{2}}.
\end{equation}
Thus, $\rho(0)=P$ and $\rho(1)=Q$.

In general, let $(M, g)$ be a Riemannian manifold.
A set $K\subseteq M$ is said to be  geodesically  convex with respect to $g$, if for any $p,q\in K$, any geodesic $\rho_{pq}$ that joins $p$ to $q$ lies entirely in $K$. 
It follows from Equation \eqref{eq:pd-geodesics} that $\Sn^n_{++}$ is a geodesically convex set with respect to the metric defined above. 

\begin{definition}[\bf Geodesically convex function]
Let $(M, g)$ be a Riemannian manifold and $K\subseteq M$ be a geodesically convex set with respect to $g$.
A function $f:K\to\R$ is said to be a geodesically convex function with respect to $g$ if for any $p,q\in K$, and for any geodesic $\rho{pq}:[0,1]\to K$ that joins $p$ to $q$,
$$
\forall t\in[0,1]\;\;f(\gamma_{pq}(t))\leq (1-t) f(p) + t f(q).
$$
\end{definition}
$\log\det(X)$  is geodesically both convex and concave on $\Sn^n_{++}$ with respect to the metric $g_X(U, V):=\Tr[X^{-1} U X^{-1} V]$.
To see this, let $X, Y\in\Sn^n_{++}$ and $t\in[0,1]$. 
Then, the geodesic joining $X$ to $Y$ is 
$$\rho(t)=X^{\frac{1}{2}}(X^{-\frac{1}{2}}YX^{-\frac{1}{2}})^t X^{\frac{1}{2}}.$$ 
Thus, 
$$
\log\det(\rho(t)) = \log\det(X^{\frac{1}{2}} (X^{-\frac{1}{2}}Y X^{-\frac{1}{2}})^t X^{\frac{1}{2}}) = (1-t)\log\det(X) + t\log\det(Y).
$$
Therefore, $\log\det(X)$ is a geodesically linear function over the positive definite cone with respect to the metric $g$.

\subsubsection{Geodesic convexity of capacity}\label{sec:cap_geodesic}

We now show that the capacity of a completely positive operator $T$ is a geodesically convex optimization problem.
First, we show that 
$T(X)$ is ``geodesically convex''.
In other words, for any geodesic, $\rho:[0,1]\to\Sn^n_{++}$,
\begin{equation}\label{eq:operator_geo}
\forall t\in[0,1],\;\;T(\rho(t))\preceq (1-t)T(\rho(0)) + t T(\rho(1)).
\end{equation}
Write $T(X):=\sum_{i=1}^m A_i X A_i^\top$ for some $n\times n$ matrices $A_i$.
Consider the geodesic $\rho(t):=P^{\frac{1}{2}}\exp(tQ)P^{\frac{1}{2}}$ for $P\in\Sn^n_{++}$ and $Q\in\Sn^n$.
The second derivative of $T$ along $\rho$ is
$$
\frac{d^2 T(\rho(t))}{d t^2} =\sum_{i=1}^m A_i P^{\frac{1}{2}} Q\exp(tQ)QP^{\frac{1}{2}} A_i^\top 
					= T(P^{\frac{1}{2}} Q\exp(tQ)QP^{\frac{1}{2}}).
$$
Since $P^{\frac{1}{2}} Q\exp(tQ)QP^{\frac{1}{2}}$ is positive definite for any $t\in[0,1]$, $T(P^{\frac{1}{2}} Q\exp(tQ)QP^{\frac{1}{2}})$ is also positive definite as $T$ is a strictly positive operator.
Consequently, $\frac{d^2}{dt^2} T(\rho(t))$ is positive definite, and~\eqref{eq:operator_geo} holds.

Now, we argue that $\log\det(T(X))$ is also geodesically convex.
We need to show that the Hessian of $\log\det(T(X))$ is positive semi-definite along any geodesic.
Let us consider the geodesic 
$\rho(t):=P^{\frac{1}{2}}\exp(tQ)P^{\frac{1}{2}}$  for $P\in\Sn^n_{++}$ and $Q\in\Sn$,
and let $h(t):=\log\det(T(\rho(t)))$.
The  second derivative of $\log\det(T(X))$ along $\rho$ is:
\begin{align*}
\frac{d^2  h(t)}{d t^2}
	=& \Tr\left[-T(\rho(t))^{-1}\frac{d}{dt} T(\rho(t))T(\rho(t))^{-1}\frac{d}{dt} T(\rho(t)) + T(\rho(t))^{-1}\frac{d^2}{dt^2}T(\rho(t))\right].
\end{align*}
Thus, we need to verify that 
$\left.\frac{d^2 h(t)}{d t^2}\right\rvert_{t=0}\geq0$.
In other words, we need to show that 
$$\Tr\left[T(P)^{-1}\left(T(P^{\frac{1}{2}}Q^2 P^{\frac{1}{2}}) -T(P^{\frac{1}{2}}QP^{\frac{1}{2}})T(P)^{-1} T(P^{\frac{1}{2}}QP^{\frac{1}{2}})\right)\right]\geq 0.
$$
In particular, if we show that
$$
T(P^{\frac{1}{2}}Q^2 P^{\frac{1}{2}}) \succeq T(P^{\frac{1}{2}}QP^{\frac{1}{2}})T(P)^{-1} T(P^{\frac{1}{2}}QP^{\frac{1}{2}}),
$$
then we are done.
Let us define another strictly positive linear operator
$$
T'(X):=T(P)^{-\frac{1}{2}} T(P^{\frac{1}{2}}XP^{\frac{1}{2}})T(P)^{-\frac{1}{2}}.
$$
If $T'(X^2)\succeq T'(X)^2$, then by picking $X=Q$ we arrive at the conclusion.
This inequality is an instance of  Kadison's inequality, see~\cite{bhatia2009positive} for more details.
Therefore, $\log\det(X)$ is a geodesically convex function.
We can now conclude that $\log\det T(X) - \log\det X$ is geodesically convex as $\log\det X$ is geodesically linear.

\begin{theorem}\label{thm:cap_geodesic}{\bf (Geodesic convexity of capacity \cite{Kubo1979,sra2015conic})}
Let $T(X)$ be a completely positive linear operator.
Then, 
$\frac{\det(T(X)}{\det(X)}$ is geodesically convex on $\Sn^n_{++}$ with respect to the metric $g_X(U, W):=\Tr[X^{-1} U X^{-1} W]$.

\end{theorem}

\subsubsection{Computing the capacity via geodesically convex optimization}\label{sec:2ndorder}

As discussed in Section \ref{sec:poly_cap}, for polynomial capacity, one
can make an appropriate change of variables and make the polynomial capacity optimization problem  convex with respect to the Euclidean
metric. 
This allows for the deployment of  standard convex optimization techniques to obtain algorithms that run in time polynomial in $n, \log M, \log \frac{1}{\eps})$; see \cite{KLRS08,CMTV17,ALOW17,StraszakV19}.
The main result of Allen-Zhu, Garg, Li, Oliviera, and Wigderson \cite{Allen-ZhuGLOW18} is an algorithm which $\eps$-approximates capacity and runs in time polynomial in
$n, m, \log M$ and $\log\frac{1}{\eps}$,
where $M$ denotes the largest magnitude of an entry
of $A_i$. 
Thus, it improves upon the result of \cite{GargGOW16} presented in Section \ref{sec:operator}
which runs in time polynomial in $n, m, \log M, \frac{1}{\eps}$.
The algorithm of \cite{Allen-ZhuGLOW18} finds an $X_\eps \succ 0$ such that 
$$\log \det (T(X_\eps)) -\log\det(X_\eps)   \leq \log \Cap (T) + \eps.$$
Their algorithm is a geodesic generalization
of the ``box-constrained'' Newton's method  introduced in \cite{CMTV17,ALOW17}. 
In
each iteration, their algorithm expands the objective into its second-order Taylor expansion and then solves it via Euclidean convex optimization; see \cite{boyd2004convex, nesterov2004introductory,vishnoi_book} for  Newton's method in Euclidean space.
Their algorithm is a general
second-order method and applies to any geodesically convex problem (over the space of positive definite matrices) that
satisfies a particular ``robustness'' property. 
This robustness property asserts  that the function behaves like a quadratic function
in every ``small'' neighborhood with respect to the metric, 
it is  weaker than self-concordance,
and it was introduced in the Euclidean space in \cite{CMTV17,ALOW17}.

Roughly speaking, their algorithm starts with an $X_0=I$ and computes  $X_{t+1}$ from $X_t$ by solving a constrained Euclidean convex quadratic minimization problem as follows: 
For a symmetric matrix $H$, let $f^t(H):=F(X_t^{\frac{1}{2}} e^{H}X_t^{\frac{1}{2}})$.
Let $q^t$ be the second-order Taylor approximation of $f^t$ around $H=0$. 
Since $F$ is geodesically convex, $q^t$ is convex in the ordinary sense. 
Thus, one can optimize $q^t(H)$ under the constraint $\|H\|_2 \leq \frac{1}{2}$ (this is the {\em box} constraint). 
If $H_t$ is the optimizer to this constrained optimization problem, $X_{t+1}:=X_t^{\frac{1}{2}} e^{H_t}X_t^{\frac{1}{2}}$.
\cite{Allen-ZhuGLOW18} show that after about $R \log \frac{1}{\eps}$ iterations, this algorithms produces an $\eps$-approximate minimizer to $F$.
Here $R$ is a bound on the {\em distance} of each iterate to the optimal solution.

For the operator scaling problem, the function $F(X):=\log \det\left(\sum_{i=1}^m A_i X A_i ^\top\right)-\log\det(X)$ which is geodesically convex over the Riemannian manifold of positive definite matrices.
They show how to modify this function slightly and  provide a bound for $R$ (or rather an alternative to it). 

 As an application, \cite{Allen-ZhuGLOW18} present a polynomial time algorithm for an equivalence problem for the left-right group action underlying the operator scaling problem. 
 This yields a deterministic polynomial-time algorithm for (commutative) PIT problems; we omit the details, see \cite{Allen-ZhuGLOW18}.

\subsection{The null-cone problem, invariant theory, and noncommutative optimization}\label{sec:nullcone}

We present a summary of the  paper by   
B\"urgisser, Franks, Garg, Oliveira,  Walter, and Wigderson
\cite{BurgisserFGOWW18,BurgisserFGOWW19} that generalizes and unifies many prior  works and initiates a systematic development of a theory of noncommutative optimization under symmetries.
We start by presenting some basics in Section \ref{sec:group_basics}.
In Section \ref{sec:null_cone}, we introduce the general definition of capacity and that of the null cone.
In section \ref{sec:nc_duality}, we introduce the notion of a moment map that leads to connections with geodesic convexity and noncommutative duality.
Finally, in Section \ref{sec:nco}, we mention the computational problems and the algorithmic results from \cite{BurgisserFGOWW19}.

\subsubsection{Groups, orbits, and invariants}\label{sec:group_basics}
We consider a vector space $V \cong \mathbb{C}^m$ for some $m$.
Given a group $G$, the {\em action} of $G$ on $V$ is a function $\phi: G \times V \to V$ for which we write $\phi(g,v)$ as just $g \cdot v$. 
A {group action} must further satisfy the properties that $g \cdot (h \cdot v) = (gh) \cdot v$ and $e \cdot v = v$, where $e$ is the identity element in $G$. 
An {\em orbit} of $v \in V$ under a given  action of $G$ is the set $$\mathcal{O}_v := \{w \in V : w = g \cdot v \text{ for some } g \in G\}.$$ 
The closure of an orbit $\mathcal{O}_v$ is denoted by $\overline{\mathcal{O}_v}$. 

A group representation $\pi$ is a map from an element  $g \in G$ to an invertible linear transformation $\pi(g)$ of the vector space $V$ (or $\mathrm{GL}(V)$). 
Enforcing $\pi$ to be a group homomorphism (i.e., for any $g_1,g_2 \in G$ we have $\pi(g_1 g_2)=\pi(g_1)\pi(g_2)$) implies the action $g \cdot v := \pi(g) v$ is a group action.

Invariant polynomials are polynomial functions on $V$ that are invariant by
the action of $G$.
The ring of invariant of polynomials is denoted by $\mathbb{C}[V]^G$ and is finitely generated due to a theorem of Hilbert \cite{Hilbert1890,Hilbert1893}.
It is known that for two vectors $v_1, v_2 \in  V$, their orbit-closures intersect if and only if $p(v_1) = p(v_2)$ for all $p \in \mathbb{C}[V]^G$; see \cite{mumford1994geometric}.

As an example, operator scaling can be viewed as  a special case of  the left-right action of $G=\mathrm{SL}_n(\mathbb{C}) \times \mathrm{SL}_n(\mathbb{C})$ on $V=(\mathbb{C}^{n\times n})^m$:
$$ \pi(C,D)\cdot (A_1,\ldots, A_m):=(C A_1  D^{\dagger},\ldots,C A_m D^{\dagger}).$$
Here, the invariants for the left-right action are generated by polynomials
of the form $\det (\sum_{i=1}^m E_i \otimes A_i)$, where $E_i$ are complex $d \times d$ matrices for some $d$.
Dersken and Makam \cite{DERKSEN201744} prove that $d \leq n^5$ suffices.
This implies that, to check if the orbit-closures for two  $(A_1,\ldots,A_m)$ and $(B_1,\ldots,B_m)$ under  the left-right action of $\mathrm{SL}_n(\mathbb{C}) \times \mathrm{SL}_n(\mathbb{C})$, it suffices to check if  $\det (\sum_{i=1}^m Y_i \otimes A_i)=\det (\sum_{i=1}^m Y_i \otimes B_i)$ for all $d \times d$ matrices $Y_i$s  on  disjoint set of variables for $d \leq n^5$.
This is an instance of the ordinary PIT problem and a deterministic algorithm for this problem is provided by the algorithm in  \cite{Allen-ZhuGLOW18} discussed in Section \ref{sec:2ndorder}.

\subsubsection{Capacity and the null cone}\label{sec:null_cone}
\cite{BurgisserFGOWW19} generalize operator scaling and the algorithmic results for it to the case when $\pi$ is {\em any} representation of  $G=\mathrm{GL}_n(\mathbb{C})$.
To do so, one needs to   
 assume that $V$ is equipped with an inner product $\langle \cdot, \cdot \rangle$ which defines a norm $\|v\|:=\sqrt{\langle v,v\rangle}$.
For a representation $\pi$, \cite{BurgisserFGOWW19} define capacity of an element $v \in V$ as 
\begin{equation}\label{eq:cap_general}
 \Cap(v):= \inf_{g \in G} \|\pi(g) v\|.
 \end{equation}
In the commutative (torus) case, this is precisely the notion of polynomial capacity introduced by Gurvits \cite{Gurvits06} (Section \ref{sec:poly_cap}).
In the left-right action case, the notion of operator capacity (Definition \ref{def:cap}) and
the one in Equation \eqref{eq:cap_general} can also be seen to  coincide.

A natural question is: For what $v$ is $\Cap(v)=0$?
This brings us to the notion of the {\em null cone} of $V$ which is defined as follows:
$$ \mathcal{N}:=\{v \in V: \Cap(v)=0\}.$$
Thus, the null cone is the set of all vectors $v \in V$ whose orbit closure contains $0$.

\subsubsection{Geodesic convexity,  moment map, and noncommutative duality}\label{sec:nc_duality}
For a representation $\pi$ of $\mathrm{GL}_n(\mathbb{C})$,  a vector $v \in V$, and an $H \in \mathcal{H}(n)$, consider  $\log \|\pi(e^{tH})v\|$ as a function of $t$.
Here, $e^{tH}$ is the matrix exponential and is a geodesic in $\mathrm{GL}_n(\mathbb{C})$ starting at the identity element in the direction $H$.
This function can be proved to be convex in $t$, making a connection to geodesic convexity; see \cite{BurgisserFGOWW19} for details. 
The derivative of this function at time $t=0$ gives rise to the {\em moment map} $\mu(v)$ for $v \in V$ as follows:
For an $H \in \mathcal{H}(n)$, $$
\langle \mu(v),H \rangle :=  \frac{\partial  \log \|\pi( e^{tH})v\|}{\partial t}  (0).
$$
Thus, a moment map can  be viewed as a noncommutative version of the gradient in a suitably defined 
Riemannian manifold that arises from the symmetries of noncommutative groups \cite{BurgisserFGOWW19}.
Hence, as $\|\pi( g)v\|$ tends to $\Cap(v)$ with $g$, $\mu(v)$ tends to zero.

For $H \in {H}(n)$, let $\mathrm{spec}(H):=(\lambda_1,\ldots,\lambda_n)$ where $\lambda_1 \geq \cdots \geq \lambda_n$ are the eigenvalues of $H$.
The {\em moment polytope} of $v$, denoted by $\Delta(v)$ is the closure of the set of eigenvalues of $\mu(w)$ as $w$ varies in the orbit of $v$:
$$ \Delta(v):= \overline{\{ \mathrm{spec}(\mu(w): w \in \mathcal{O}_v \}}.
$$
It is a nontrivial result that $\Delta(v)$ is a convex polytope \cite{Kostant1973,Atiyah1982,NM}.

It was proved by Kempf and Ness \cite{KempfNess} that $v$ is not in the null-cone $\mathcal{N}$ if and only if  $\mu(w)=0$ for some $w$ in the orbit closure of $v$, or $0 \in \Delta(v)$.
This is an important result and can be viewed as a noncommutative analog of  Farkas' Lemma in the commutative world.
Thus, we can draw an analogy to convex optimization: If we view the moment map as the gradient of the action of $pi$ at the identity element, then the $\|w\|$ is minimized when the gradient is zero.
$v$ is in the null cone if and only if 
$\Cap(v)=0$.

One of the key structural results in \cite{BurgisserFGOWW19} is a quantitative version of the Kempf-Ness theorem.

\begin{theorem}\label{thm:nc}{\bf (Noncommutative duality \cite{BurgisserFGOWW19})}
For a unit vector $v$ in $V$, 
$$ 1- \frac{\|\mu(v)\|}{\gamma(\pi)} \leq \Cap^2(v) \leq 1- \frac{\|\mu(v)\|^2}{4N(\pi)}.$$
\end{theorem}
Here, the {\em weight norm} $N(\pi)$ is defined to be the maximum Euclidean norm of a weight that occurs in $\pi$.
A weight vector $\lambda \in \mathbb{Z}^n$ occurs in $\pi$ if one of its irreducible subspaces is of type $\lambda$.
And, the {\em weight margin} $\gamma(\pi)$  is the minimum Euclidean distance between the origin and the convex hull of any subset of the weights of $\pi$ that does not contain the origin.
The weights arise in the study of irreducible representations of $\pi$ and we direct the reader to \cite{BurgisserFGOWW19} for a discussion on them.

For matrix scaling (the left-right action by special torus group), it can be shown that $\gamma(\pi)\geq \frac{1}{\mathrm{poly}(n)}$.
For operator scaling too  (with the left-right action  by $\mathrm{SL}_n(\mathbb{C}) \times \mathrm{SL}_n(\mathbb{C})$), it can be shown that $\gamma(\pi)\geq \frac{1}{\mathrm{poly}(n)}$.

\subsubsection{Noncommutative optimization under symmetries}\label{sec:nco}

\cite{BurgisserFGOWW19} study a variety of general and related problems related to  orbits of group actions.
\begin{enumerate}
\item {\bf Null cone membership problem:} Given  $(\pi,v)$, check if $v \in \mathcal{N}$.

\item {\bf Moment polytope membership problem:} Given $(\pi,v,p)$, check if $p \in \Delta(v)$.

\item {\bf  Norm-minimization problem:} Given $(\pi,v,\eps)$ such that $\Cap(v)>0$, output a $g \in G$ such that $\log \|\pi(g) \cdot v\| - \log \Cap(v) \leq \eps$.

\item {\bf Scaling problem:} Given $(\pi, v, p,\eps)$ such that $p \in \Delta(v)$, output an element $g \in G$ such that $\|\mathrm{spec}(\mu(\pi(g)v)) - p\| \leq \eps$.

\end{enumerate}

\noindent
\cite{BurgisserFGOWW19} discuss 
how these problems capture a diverse set of problems in different areas of computer science, mathematics, and physics. 
We already discussed the application to approximating the permanent (Section \ref{sec:permanent}), noncommutative singularity testing (Section \ref{sec:ncit}), and computing Brascamp-Lieb constants (Section \ref{sec:BL}).
Other applications include the Horn problem: Do there exist three Hermitian matrices $A,B,C$ with prescribed eigenvalues such that $A+B=C$?,
the quantum marginal problem: Given density matrices describing local quantum states, is there a global pure state consistent with the local states?
Moreover, these problems also connect to  geometric complexity theory (GCT) \cite{MulmuleySohoni} that formulates a variant of {\sf VP} vs. {\sf VNP} question as
 checking if the (padded) permanent lies in the orbit-closure of the determinant (of an appropriate
 size), under the action of the general linear group on polynomials induced by its natural
linear action on the variables.

\cite{BurgisserFGOWW19} also show how, sometimes, these abovementioned problems may reduce to each other and discuss multiple ways in which the input may be specified.
For instance, in the operator scaling problem $\pi$ is fixed (and not part of the input) while, in general, one could be given an oracle to $\pi(g)v$ for a $g \in G$ and an input vector $v$.
$p$ and $\eps$ are assumed to be  given in binary and they present algorithms that run in time both a polynomial 
in $\frac{1}{\eps}$ and in $\log \frac{1}{\eps}$.
\cite{BurgisserFGOWW19} note that techniques from \cite{SinghV14,StraszakV19}  can  be used to
design polynomial time algorithms for commutative null cone and moment polytope membership in
the oracle setting.

Prior works for these problems, including the ones discussed in Section \ref{sec:permanent} and \ref{sec:operator}, the underlying groups
need to be products of at least two copies of rather specific linear groups ($\mathrm{SL}(n)$s or tori),
to support the algorithms and analysis.
More importantly, these actions were {\em linear} in each of the copies.
In \cite{BurgisserFGOWW19}, arbitrary group actions of $\mathrm{GL}_n$ that can be described by a  representation, are handled.
They develop two general methods, a first-order and a second-order method, which require information about the gradient and the Hessian of the function to be optimized. 
Their algorithms  rely on the connection of the moment map to geodesic convexity and the running time bounds depend on the quantitative parameters -- weight norm and weight margin -- arising in their quantitative version of noncommutative duality (Theorem \ref{thm:nc}).
The main technical work goes into showing how these parameters control convergence to the optimum in each of these methods.

The first-order method of \cite{BurgisserFGOWW19} is a natural analog of gradient descent.
For the problem of computing  $\Cap(v)$, it starts with an element $g_0=I$ (the identity element in $G$) and repeats for $\tau$ iterations and a suitable ``step-size'' $\eta>0$ the following:
$$ g_{t+1}:=e^{-\eta \mu(\pi(g_t)v)}g_t.$$
They show that there is a choice of $\eta$ such that this method, when $Cap(v)>0$, finds a $g$ such that $\|\mu(\pi(g)v)\| \leq \eps$
for $\tau=\left(\frac{N(\pi)^2}{\eps^2}|\log \Cap (v)| \right).$
This approximately solves the scaling problem for $p=0$.
The generalization to $p \neq 0$ is also presented.

Their second-order method, at a high-level, repeatedly optimizes  
quadratic Taylor expansions of the objective in a small neighborhood (similar to Newton's method in convex optimization). 
It is an extension of their method for computing operator capacity mentioned in Section \ref{sec:2ndorder}.
The number of iterations it  takes for the above-mentioned scaling problem is $\tilde{O}\left( \frac{N(\pi)\sqrt{n}}{\gamma(\pi)} \left( |\log \Cap(v) | + \log \frac{n}{\eps}\right) \right)$.

The work of \cite{BurgisserFGOWW19} has also led to a host of new challenges in noncommutative optimization.
An important one is to design analogs of the ``cutting plane'' or the ``interior point methods'' in the noncommutative setting.
Such algorithms would likely yield true polynomial time algorithms for Problems (1)-(4) mentioned above; see \cite{HNW23} for some progress towards the latter goal.
Finally, there are several other  works where the lens of symmetry has been helpful in the design of nonconvex optimization and sampling algorithms, see \cite{Barvinok05convexgeometry,sanyal2011,Waterhouse:1983:DSP,Edelman1999,Parrilo,LeakeV22,LeakeV20b,LeakeMV21} and the references therein.

\bibliographystyle{abbrv}
\bibliography{refs,pseudorandomness}
\end{document}